%% file: main.tex
\documentclass[11pt,a4paper]{article}

\usepackage[english]{babel}
\usepackage{marvosym}
\usepackage{enumerate,xspace,stmaryrd,xcolor}
\usepackage{amssymb,amsmath,amsthm,latexsym}
\usepackage[shortlabels]{enumitem}
\usepackage[final]{hyperref}
\usepackage{fixme}
\usepackage[T1]{fontenc}
\FXRegisterAuthor{fc}{afc}{FC}

\usepackage{geometry}
\usepackage{tikz}
\usetikzlibrary{arrows,positioning}
\usetikzlibrary{automata,backgrounds,petri}
\usetikzlibrary{trees,shapes,decorations,calc}
\usetikzlibrary{patterns}

\input{fc-commands}

\title{%
Characterization theorems for \pdl and \fotc
}

\author{%
    Facundo Carreiro%
    \thanks{E-mail: \texttt{fcarreiro@dc.uba.ar}}\\\\
    \emph{Institute for Logic, Language and Computation}\\
    \emph{University of Amsterdam, The Netherlands}
}

\date{\small%
First version: January 12, 2015\\
Last revision: \today
}

\begin{document}

\maketitle

\noindent
\begin{center}
\begin{tabular}{|p{0.85\textwidth}|}
\hline
\textbf{Warning}: Parts of this manuscript are not valid due to a gap in the proofs of Section~\ref{ssec:cafoe2wcl}. Namely, the automata characterization of $\mucafoe$ and the bisimulation result for $\binfotc$ do not follow from the contents of this manuscript. More detail about the specific problems and possible solutions are to appear in my dissertation on December 2015.\\
\hline
\end{tabular}
\end{center}

\vspace{0.8cm}
\input{abstract}

\newpage
{\small%
\tableofcontents}
\newpage

\section{Introduction}\label{sec:intro}
\input{introduction}

\section{Preliminaries}\label{sec:prel}
\subsection{General terminology, transition systems and trees} \label{ssec:prelim_trees}
\input{prel-tstrees}
\subsection{Games}
\input{prel-games}
\subsection{Parity automata}
\input{prel-parityaut}
\subsection{Propositional Dynamic Logic}\label{subsec:pdl}
\input{prel-pdl}
\subsection{The modal $\mu$-calculus}\label{subsec:mu}
\input{prel-muml}
\subsection{Bisimulation}
\input{prel-bisimulation}
\subsection{Weak chain logic}\label{sec:wcl}
\input{prel-wcl}
\subsection{Fixpoint extension of first-order logic}
\input{prel-mufoe}
\subsection{First-order logic with transitive closure}\label{sec:prel:fotc}
\input{prel-fotc}

\subsection{Notational convention}\label{sec:prel:conventions}
\input{prel-convention}

\section{Characterization of $\binfotc$ inside $\mufoe$}\label{sec:fotcinmufoe}
\input{fotc-intro}
\subsection{Fixpoint theory of completely additive maps}\label{sec:fixpoint-caf}
\input{fixpoint-caf}
\subsection{Completely additive restriction of $\mufoe$}
\input{fotc-defcaf}

\subsection{Equivalence of $\binfotc$ and $\mucafoe$}
\input{fotc-mucafoe}

\section{One-step logics, normal forms and additivity}\label{sec:onestep}
\input{os-intro}

\subsection{Normal forms}\label{subsec:normalforms}
\input{os-normalforms}
\subsection{One-step monotonicity}
\input{os-monotonicity}
\subsection{One-step additivity}\label{subsec:one-stepcont}
\input{os-additivity}
\subsection{One-step multiplicativity and Boolean duals}\label{subsec:one-stepcocont}
\input{os-multiplicativity}
\subsection{Effectiveness of the normal forms}
\input{os-decidability}

\section{Additive-weak parity automata}\label{sec:aut}
\input{soaut-def}
\subsection{Simulation theorem}
\input{soaut-simulation}
\subsection{Closure properties}
\input{soaut-clintro}
\subsubsection{Closure under projection}
\input{soaut-projection}
\subsubsection{Closure under Boolean operations}
\input{soaut-booleans}

\section{Logical characterizations of $\AutWA(\ofoe)$ on trees}\label{sec:soaut-logic}
\input{soaut-logic-intro}
\subsection{Fixpoint theory of maps that restrict to descendants}
\input{fixpoint-rdes}
\subsection{Forward-looking fragment of $\mufoe$}
\input{soaut-logic-forward}
\subsection{Translations}
\input{soaut-logic-wcl2aut}
\input{soaut-logic-aut2caf}
\input{soaut-logic-caf2wcl}

\section{Expressiveness modulo bisimilarity}\label{sec:char}
\input{exp-intro}
\subsection{Bisimulation-invariant fragment of $\AutWA(\ofoe)$}
\input{exp-autofoe}

\subsection{Bisimulation-invariant fragment of \binfotc}
\input{exp-binfotc}
\subsection{Bisimulation-invariant fragment of \wcl}
\input{exp-wcl}
\subsection{\pdl versus \wcl versus \binfotc}\label{ssec:versus}
\input{exp-versus}


\section{Conclusions and open problems}
\input{conclusions}

\bibliographystyle{alpha}
\bibliography{main}

\end{document}

%% file: fc-commands.tex
\newcommand{\nat}{\ensuremath{\mathbb{N}}\xspace}
\newcommand{\tup}[1]{\langle{#1}\rangle}
\newcommand{\ext}[1]{\llbracket#1\rrbracket}
\newcommand{\props}{\mathsf{P}}
\newcommand{\oprops}{\mathsf{P'}}
\newcommand{\qprops}{\mathsf{Q}}
\newcommand{\acts}{\mathsf{D}}
\newcommand{\aact}{\ell}
\newcommand{\moddom}{S}
\newcommand{\tmoddom}{T}
\newcommand{\npmoddom}{M}
\newcommand{\model}{\mathbb{\moddom}}
\newcommand{\tmodel}{\mathbb{\tmoddom}}
\newcommand{\npmodel}{\mathbb{\npmoddom}}
\newcommand{\osmodel}{\mathbf{D}}
\newcommand{\compset}[1]{\{#1\}}

\newcommand{\mmodels}{\Vdash}
\newcommand{\umods}{\mathfrak{M}_1}
\newcommand{\sumods}{\umods^s}

\newcommand{\omegaunrav}[1]{{#1}^\omega}
\newcommand{\unravel}[1]{\hat{#1}}
\newcommand{\bis}{\mathrel{\underline{\leftrightarrow}}}

\newcommand{\resto}{{\upharpoonright}}
\newcommand{\dcup}{\uplus}
\newcommand{\stl}{\rightarrowtriangle} 
\newcommand{\lfp}{\mathsf{LFP}}
\newcommand{\gfp}{\mathsf{GFP}}
\newcommand{\tc}{\mathsf{TC}}
\newcommand{\st}{\mathrm{ST}}
\newcommand{\sorts}{\mathcal{S}}
\newcommand{\asort}{\mathtt{s}}
\newcommand{\aSort}{\mathtt{S}}
\newcommand{\pto}{\rightharpoonup}
\newcommand{\last}{\ensuremath{\mathtt{last}}\xspace}

\newcommand{\pdl}{\ensuremath{\mathrm{PDL}}\xspace}
\newcommand{\wcl}{\ensuremath{\mathrm{WCL}}\xspace}

\newcommand{\mso}{\ensuremath{\mathrm{MSO}}\xspace}

\newcommand{\muML}{\ensuremath{\mu\mathrm{ML}}\xspace}
\newcommand{\ML}{\ensuremath{\mathrm{ML}}}

\newcommand{\fo}{\ensuremath{\mathrm{FO}}\xspace}
\newcommand{\fotc}{\ensuremath{\mathrm{FO(TC)}}\xspace}
\newcommand{\folfp}{\ensuremath{\mathrm{FO(LFP)}}\xspace}
\newcommand{\unfolfp}{\ensuremath{\mathrm{FO(LFP^1)}}\xspace}
\newcommand{\binfotc}{\ensuremath{\mathrm{FO(TC^1)}}\xspace}

\newcommand{\ofo}{\ensuremath{\fo_1}\xspace}
\newcommand{\foe}{\ensuremath{\mathrm{FOE}}\xspace}
\newcommand{\ofoe}{\ensuremath{\foe_1}\xspace}

\newcommand{\mucaML}{\mu_\mathit{ca}\textup{ML}}
\newcommand{\mufoe}{\ensuremath{\mu\foe}\xspace}
\newcommand{\mucafoe}{\ensuremath{\mu_{ca}\foe}\xspace}

\newcommand{\muffoe}{\ensuremath{\mufoe^\ff}\xspace}
\newcommand{\mucaffoe}{\ensuremath{\mucafoe^\ff}\xspace}

\newcommand{\ff}{\textup{\guillemotright}}
\newcommand{\add}[2]{#1\mathsf{ADD}_{#2}}
\newcommand{\mult}[2]{#1\mathsf{MUL}_{#2}}
\newcommand{\monot}[2]{#1\mathsf{MON}_{#2}}

\newcommand{\dbnf}{\nabla}										
\newcommand{\dgbnfofo}[2]{\dbnf_{\fo}(#1,#2)}					
\newcommand{\dbnfofoe}[2]{\dbnf_{\foe}(#1,#2)}					

\newcommand{\mondgbnfofo}[3]{\dbnf^{#3}_\fo(#1,#2)}				
\newcommand{\mondbnfofoe}[3]{\dbnf^{#3}_{\foe}(#1,#2)}			

\newcommand{\posdgbnfofo}[2]{\dbnf^+_{\fo}(#1,#2)}				
\newcommand{\posdbnfofoe}[2]{\dbnf^+_{\foe}(#1,#2)}				

\newcommand{\fovar}{\mathsf{iVar}}

\newcommand{\ass}{g}
\newcommand{\val}{V}
\newcommand{\uval}{U}

\newcommand{\tscolors}{\kappa}
\newcommand{\tsval}{\tscolors^\natural}

\newcommand{\aut}{\ensuremath\mathbb{A}\xspace}
\newcommand{\baut}{\ensuremath\mathbb{B}\xspace}

\newcommand{\tmap}{\Delta}
\newcommand{\tmapProj}{\tmap^\exists}
\newcommand{\pmap}{\Omega}
\newcommand{\pmapProj}{\pmap^\exists}
\newcommand{\ord}{\preceq}
\newcommand{\reach}{\leadsto}
\newcommand{\lthen}{\to}
\newcommand{\llang}{\ensuremath{\mathcal{L}}\xspace}
\newcommand{\trees}{\ensuremath{\mathcal{T}}}

\newcommand{\mccomp}{C} 

\newcommand{\inc}{\sqsubseteq}
\newcommand{\here}[1]{{\Downarrow}#1}
\newcommand{\arediff}[1]{\mathrm{diff}(#1)}
\newcommand{\aresame}[1]{\mathrm{equal}(#1)}
\newcommand{\foeq}{\approx}
\newcommand{\tadd}{\vartriangle}
\newcommand{\tmono}{\circ}
\newcommand{\tmodal}{\blacktriangledown}
\newcommand{\vlist}[1]{\overline{{\mathbf{#1}}}}

\newcommand{\efgame}{\mathrm{EF}}

\newcommand{\agame}{\mathcal{A}}
\newcommand{\game}{\mathcal{G}}
\newcommand{\eloise}{\ensuremath{\exists}\xspace}
\newcommand{\abelard}{\ensuremath{\forall}\xspace}
\newcommand{\player}{\ensuremath{\Pi}\xspace}
\newcommand{\pmatches}[2]{\mathsf{PM}^{#1}_{#2}}
\newcommand{\win}{\text{\sl Win}} 

\newcommand{\Aut}{\mathit{Aut}}

\newcommand{\AutWA}{\mathit{Aut}_{wa}}

\newcommand{\seq}{{;}}
\newcommand{\choice}{\oplus}

\newcommand{\prog}{\pi}
\newcommand{\pprog}{\rho}
\newcommand{\relprog}{R}

\newcommand{\existsc}{\exists_{c}}
\newcommand{\existswc}{\exists_{wc}}
\newcommand{\forallwc}{\forall_{wc}}

\newcommand{\fconst}{F} 
\newcommand{\fwa}[2]{{#1}^{\!\div}_{#2}} 
\newcommand{\matches}{\mathcal{B}}
\newcommand{\match}{\pi}
\newcommand{\shA}{A^{\wp}}

\newcommand{\Dom}{\mathsf{Dom}}
\newcommand{\Ran}{\mathsf{Ran}}
\newcommand{\dual}[1]{#1^{\delta}}
\newcommand{\nada}{\varnothing}
\newcommand{\diff}[1]{\mathtt{diff}(#1)}

\newtheorem{theorem}{Theorem}[section]
\newtheorem{fact}[theorem]{Fact}
\newtheorem{proposition}[theorem]{Proposition}
\newtheorem{lemma}[theorem]{Lemma}
\newtheorem{corollary}[theorem]{Corollary}
\newtheorem{conjecture}[theorem]{Conjecture}

\theoremstyle{definition}
\newtheorem{definition}[theorem]{Definition}
\newtheorem{remark}[theorem]{Remark}

\newenvironment{proofof}[1]{\begin{trivlist}\item[\hskip\labelsep{\bf
Proof~of~{#1}.\ }]}{\hspace*{\fill} {\sc qed}\end{trivlist}}

\newtheorem{claim2}{\sc Claim}
\newenvironment{claim}{\begin{claim2}\rm}{\end{claim2}\rm}
\newenvironment{claimfirst}{\setcounter{claim2}{0}
               \begin{claim2}\rm}{\end{claim2}\rm}

\newenvironment{pfclaim}{\begin{trivlist}\item[]{\sc Proof of
Claim.}}{\hfill {\mbox{$\blacktriangleleft$}}\end{trivlist}}

\newenvironment{tbs}{%
   \small\tt
   \begin{itemize}}{\end{itemize}}
\newcommand{\btbs}{\begin{tbs}}                                                                      
\newcommand{\etbs}{\end{tbs}}

%% file: abstract.tex

\begin{abstract}
\noindent 
Our main contributions can be divided in three parts:
\begin{enumerate}[(1)]
  \itemsep 0 pt
  \item
  \textit{Fixpoint extensions of first-order logic}: we give a precise syntactic and semantic characterization of the relationship between \binfotc and \folfp.
  \item
  \textit{Automata and expressiveness on trees}: we introduce a new class of parity automata which, on trees, captures the expressive power of \binfotc and \wcl (weak chain logic). The latter logic is a variant of \mso which quantifies over finite chains. 
  \item 
  \textit{Expressiveness modulo bisimilarity}: we show that \pdl is expressively equivalent to the bisimulation-invariant fragment of both \binfotc and \wcl. 
\end{enumerate}
In particular, point (3) closes the open problems of the bisimulation-invariant characterizations of \pdl, \binfotc and \wcl all at once.

\bigskip
\noindent \textbf{Keywords:} Propositional Dynamic Logic, Transitive-closure logic, Chain Logic, characterization theorem, bisimulation-invariant fragment, fixed points, parity automata, Janin--Walukiewicz theorem, van~Benthem theorem, complete additivity.
\end{abstract}

%% file: introduction.tex

In this article we study the relative expressive power of various modal and first-order fixpoint logics; that is, logics which have some kind of iteration or recursion mechanism.

On the modal side, we consider Propositional Dynamic Logic, a well-known extension of the basic modal language with non-deterministic programs (in particular, with the iteration construct $\prog^*$.) On the first-order side, we consider extensions of first-order logic with a reflexive-transitive closure operator (similar to that of \pdl) and first-order logic extended with a fixpoint operator (similar to the modal $\mu$-calculus).

\bigskip\noindent
This article can be roughly divided in three parts:
\begin{enumerate}[(1)]
  \itemsep 0 pt
  \item
  \textit{Fixpoint extensions of first-order logic}: we give a precise syntactic and semantic characterization of the relationship between \binfotc and \folfp.
  \item
  \textit{Automata and expressiveness on trees}: we introduce a new class of parity automata which, on trees, captures the expressive power of \binfotc and \wcl (weak chain logic). The latter logic is a variant of \mso which quantifies over finite chains. 
  \item 
  \textit{Expressiveness modulo bisimilarity}: we show that \pdl is expressively equivalent to the bisimulation-invariant fragment of both \binfotc and \wcl. 
\end{enumerate}

\paragraph{Extensions of first-order logic.}
It is well known that the reflexive-transitive closure $R^*$ of a binary relation $R$ is not expressible in first-order logic~\cite{MALQ:MALQ19750210112}. Therefore, a straightforward way to extend first-order logic is to add a reflexive-transitive closure operator:
\[
[\tc_{\vlist{x},\vlist{y}}.\varphi(\vlist{x},\vlist{y})](\vlist{u},\vlist{v})
\]
which states that $(\vlist{u},\vlist{v})$ belongs to the transitive closure of the relation denoted by $\varphi(\vlist{x},\vlist{y})$. In the above expression, the sequences of variables $\vlist{x},\vlist{y},\vlist{u},\vlist{v}$ should all be of the same length; this length is called the \emph{arity} of the transitive closure.
This extension of first-order logic, called \fotc or sometimes transitive-closure logic, was introduced by Immerman in~\cite{Immerman87languagesthat} where he showed that it captures the class of NLOGSPACE queries.

The arity hierarchy of \fotc was proven strict for finite models~\cite{Grohe1996103}.
Moreover, in some restricted classes of trees, full \fotc is even more expressive than \mso~\cite{fotcmorethanmso}. In this paper, however, we restrict our attention to \binfotc, which is \fotc restricted to sequences of length one; that is, the reflexive-transitive closure can only be applied to formulas $\varphi(x,y)$ defining a \emph{binary relation}. This fragment of \fotc is easily seen to be included in \mso.

\medskip
A more general way to extend first-order logic is to add a fixpoint operator as in~\cite{Chandra198299}. Consider, as an example, a first-order formula $\varphi(p,x)$ where $p$ is a monadic predicate and $x$ is a free variable. The set of elements $s\in\npmoddom$ of some model $\npmodel$ which satisfy $\varphi(p,s)$ clearly depends on the interpretation of $p$. This dependency can be formalized as a map 
\[
F^\varphi_{p{:}x}(Y) := \{s \in \npmoddom \mid \varphi(Y, s) \text{ is true at } \npmodel \}.
\]
Assuming that $\varphi$ is monotone in $p$, the least and greatest fixpoints of this map will exist by the Knaster-Tarski theorem. It is now easy to extend first-order logic with a fixpoint construction
\[
[\lfp_{p{:}x}.\varphi(p,x)](z)
\]
which holds iff the interpretation of $z$ belongs to the least fixpoint of the map $F^\varphi_{p{:}x}$. This extension is called first-order logic with \emph{unary} fixpoints and is usually denoted by \unfolfp. This is because the arity of the fixpoint relation (in our case $p$) is unary. As with the transitive closure operator, we can consider a logic where the arity of the fixpoint is not bounded. This logic is known as \folfp and was shown to capture PSPACE queries~\cite{Immerman87languagesthat}. In this article we focus on \unfolfp, which we will denote as $\mufoe$.

\medskip
It is not difficult to see (\textit{cf.}~Section~\ref{sec:prel:fotc}) that \fotc is included in \folfp and \binfotc is included in $\mufoe$. However, to the best of our knowledge the exact fragment of $\mufoe$ that 
corresponds to $\binfotc$ has not been characterized. 
As we will prove in this article, the key notion leading to such a 
characterization is that of \emph{complete additivity}.

A formula $\varphi$ is said to be completely additive in $p$ if for any non-empty family
of subsets $\{P_i\}_{i\in I}$ and variable $x$ 
it satisfies the following equation:
\[
F^\varphi_{p{:}x}(\bigcup_i P_i) = \bigcup_{ i } F^\varphi_{p{:}x}(P_i).
\]
Complete additivity has been studied (under the name `continuity') by van
Benthem~\cite{vanbenthem:1996}, in the context of operations on relations that 
are \emph{safe for} (that is, preserve) bisimulations. On the modal side,
both Hollenberg~\cite{MarcoPhD} and Fontaine and Venema~\cite{GaellePhD,FV12} gave syntactic 
characterizations of complete additivity for the modal $\mu$-calculus.\fcwarning{muML not introduced} Finally,
Carreiro and Venema~\cite{AiML2014} showed that \pdl is equivalent to the fragment $\mucaML$ of the $\mu$-calculus where the fixpoint operator is restricted to completely additive formulas.

\begin{fact}[\cite{AiML2014}]
  $\pdl$ is effectively equivalent to $\mucaML$ over all models.
\end{fact}

The first contribution of this article goes in the same direction as the above result. First of all, we consider a syntactic fragment $\mucafoe$ of $\mufoe$ by restricting the application of the unary fixpoint to formulas which are completely additive. We prove that \binfotc effectively corresponds to this fragment.

\begin{theorem}\label{thm:fotcmucafoe}
  $\binfotc$ is effectively equivalent to $\mucafoe$ over all models.
\end{theorem}

As a minor contribution towards the above theorem we give, in Section~\ref{sec:fixpoint-caf}, a general characterization of fixpoints of arbitrary completely additive \emph{maps}, that is, maps which need not be induced by a formula.



\paragraph{Automata and expressiveness on trees.} 
It is difficult to overstress the importance of automata-theoretic techniques in logical questions.
The literature is vast in this topic, and we only name a few results which are relevant for the current article:
The classical work of Rabin~\cite{Rabin69} introduced tree automata to show that the monadic second-order theory of the infinite binary tree is decidable. This is one of the most fundamental decidability results to which many other decidability results in logic and computer science can be reduced. 

In a more contemporary paper~\cite{Walukiewicz96}, Walukiewicz introduced \emph{parity automata} for \mso on arbitrary trees. These automata were crucial in proving that the modal $\mu$-calculus is the bisimulation-invariant fragment of \mso~\cite{Jan96}. On the modal side, other classes of parity automata were used to prove the small model property~\cite{ydesmall} and uniform interpolation~\cite{MuIntp2000JSL} for the $\mu$-calculus.

If we restrict to \emph{sibling-ordered} trees, nested tree-walking automata were used
to separate \binfotc and \mso~\cite{Cate:2010:TCL:1706591.1706598}. 
Also, \fotc is known to correspond to automata with nested pebbles~\cite{FOTCpebbles} on finite ranked trees. However, not much is known about automata models for \fotc on arbitrary unordered unranked trees.

\medskip
In this part of the article we introduce a new class of parity automata which we call \emph{additive-weak parity automata}. The main result of this section shows that these automata capture many of the logics of this article (over trees) and will work as a common bridge between them in later sections. Before we turn to a description of these automata, we first have a look at the automata introduced by Walukiewicz~\cite{Walukiewicz96}, corresponding to $\mso$ (over tree models).

We fix the set of proposition letters of our models as $\props$ and think of
$\wp(\props)$ as an \emph{alphabet} or set of \emph{colors}.
An \mso-automaton is then a tuple $\aut = \tup{A, \tmap,
\pmap, a_I}$, where $A$ is a finite set of states, $a_I$ an initial state, and $\pmap:
A \to \nat$ is a parity function.
The transition function $\tmap$ maps a pair $(a,c) \in A \times \wp(\props)$ to a
sentence in the monadic first-order language with equality $\ofoe(A)$, of which the
state space $A$ provides the set of (monadic) predicates.
We shall refer to $\ofoe$ as the \emph{one-step language} of $\mso$-automata,
and denote the class of $\mso$-automata with $\Aut(\ofoe)$. The semantics of
these automata will be provided later in the article.
Walukiewicz's key result linking $\mso$ to $\Aut(\ofoe)$ states the following:
\begin{fact}[\cite{Walukiewicz96}]
  \mso and $\Aut(\ofoe)$ are effectively equivalent over tree models.
\end{fact}
One of the directions of this result is proved by inductively showing that every formula
$\varphi$ in $\mso$ can be effectively transformed into an equivalent
automaton $\aut_{\varphi} \in \Aut(\ofoe)$. In order to do that, these automata
are shown to be closed under the operations of \mso, i.e., Boolean operators and quantification over sets.

\medskip
The additive-weak parity automata that we introduce in Section~\ref{sec:aut} are a restriction of \mso-automata. In order to state the constraints, observe that given a parity automaton $\aut$ we can induce a graph on $A$ by setting transition from $a$ to $b$ if $b$ occurs in $\tmap(a,c)$ for some $c \in \wp(\props)$. A parity automaton $\aut$ is called an additive-weak automaton if it satisfies the following constraints for every \emph{maximal} strongly connected component $\mccomp \subseteq A$, states $a,b \in \mccomp$ and color $c\in \wp(\props)$:
\begin{description}
\itemsep 0 pt
\item[(weakness)] $\pmap(a) = \pmap(b)$.
\item[(additivity)] for every color $c\in\wp(\props)$:\\
    If $\pmap(b)=1$ then $\tmap(a,c)$ is completely additive in $\mccomp$.\\
    If $\pmap(b)=0$ then $\tmap(a,c)$ is completely multiplicative in $\mccomp$.
\end{description}
The class of these automata is denoted by $\AutWA(\ofoe)$.

As opposed to $\Aut(\ofoe)$, this class is not closed under the existential quantification of \mso. We prove that, instead, it is closed under the operations of \emph{weak chain logic} (\wcl), a monadic second-order logic which quantifies over \emph{finite chains} (as opposed to arbitrary sets). Here, we define a chain on a tree $\tmodel$ to be a set $X$ such that all elements of $X$ belong to the same branch. The original (non-weak) chain logic was introduced by Thomas in~\cite{Thomas:1984:LAS:2868.2871} and further studied in~\cite{Thomas96languagesautomata,Mikolaj-Thesis}. 

The second main contrubution of this article proves that, on trees, the class of additive-weak automata captures the expressive power of many of the fixpoint logics that we have considered, and also of weak chain logic.

\begin{theorem}\label{thm:allthesame}
On trees, the following formalisms are expressively equivalent:
\begin{enumerate}[(i)]
  \itemsep 0 pt
  \item $\AutWA(\ofoe)$: additive-weak automata based on $\ofoe$,
  \item $\wcl$: weak chain logic,
  \item $\mucafoe$: completely additive restriction of $\mufoe$,
  \item $\binfotc$: first-order logic with binary reflexive-transitive closure.
\end{enumerate}
Moreover, the equivalence is given by \emph{effective} translations.
\end{theorem}
It is worth remarking that even though \binfotc and \wcl coincide over trees, the logics \binfotc and \wcl themselves are not equivalent. This is shown in Section~\ref{ssec:versus}.

Another point worth observing is that the original automata characterization of \mso was done for the signature with only one binary relation $R$ (and many monadic predicates). In this article we prove all the results for a signature with a family of binary relations $R_{\aact\in\acts}$, mainly because of our interest in the connection with poly-modal logics like \pdl. Adapting the automata to this setting required new techniques which are (necessarily) different from the ones applicable to \mso-automata. For example, we use \emph{multi-sorted} one-step languages in our automata.

In order to prove the above results, we need a detailed analysis of certain fragments of $\ofoe$. The third contribution of this article is to provide, in Section~\ref{sec:onestep}, a syntactical characterization of the monotone, completely additive and completely multiplicative fragments of \emph{multi-sorted} monadic first-order logic with and without equality.


\paragraph{Expressiveness modulo bisimilarity.}
The last part of this article concerns the relative expressive power of some languages
when restricted to properties which are bisimulation-invariant.
The interest in such expressiveness questions stems from applications where
transition systems model computational processes, and bisimilar
structures represent the \emph{same} process.
Seen from this perspective, properties of transition structures are relevant
only if they are invariant under bisimilarity.
This explains the importance of bisimulation-invariance results of the form
\[
M \equiv L / {\bis} ,
\]
stating that, 
one language $M$ is expressively complete with respect to the
relevant (i.e., bisimulation-invariant) properties that can be formulated in
another language $L$.
In this setting, generally $L$ is some rich yardstick formalism such as
first-order or monadic second-order logic, and $M$ is some modal-style
fragment of $L$, usually displaying much better computational behavior
than the full language $L$.

A seminal result in the theory of modal logic is van Benthem's Characterization
Theorem~\cite{vanBenthemPhD}, stating that every bisimulation-invariant
first-order formula is actually equivalent to (the standard
translation of) a modal formula:
\[
\ML \equiv \mathrm{FO}/{\bis} .
\]
Over the years, a wealth of variants of the Characterization Theorem have been
obtained.
For instance, Rosen proved that van Benthem's theorem is one of the few
preservation results that transfers to the setting of finite
models~\cite{rose:moda97}; for a recent, rich source of van Benthem-style
characterization results, see Dawar \& Otto~\cite{DawarO09}.
In this paper we are mainly interested is the work of Janin \&
Walukiewicz~\cite{Jan96}, who extended van Benthem's result to the setting
of fixpoint logics, by proving that the modal $\mu$-calculus ($\muML$) is the
bisimulation-invariant fragment of monadic second-order
logic (\mso):
\[
\muML \equiv \mso/{\bis} .
\]
Despite the continuous study of the connection between modal and classical logics there are still important logics which are not well understood and represent exciting problems. In particular, the bisimulation-invarant fragments of \wcl and \binfotc have not been characterized. Also, it is not known wether there is a natural classical logic whose bisimulation-invariant fragment corresponds to \pdl (see~\cite[p.~91]{MarcoPhD}), even though there are results leading towards this direction~\cite{safe,MarcoPhD,vanbenthem:1996}.

The language now called Propositional Dynamic Logic was first investigated by 
Fisher and Ladner~\cite{DBLP:journals/jcss/FischerL79} as a logic to reason 
about computer program execution.
\pdl extends the basic modal logic with an infinite collection of diamonds 
$\tup{\prog}$ where the intended intuitive interpretation of $\tup{\prog}\varphi$
is that ``some terminating execution of the program $\prog$ from the current 
state leads to a state satisfying~$\varphi$''.

One of the most important and characteristic features of \pdl is that the 
program construction $\prog^*$ (corresponding to iteration) endows \pdl with 
second-order capabilities while still keeping it computationally well-behaved.
For an extensive treatment of \pdl we refer the reader 
to~\cite{Harel:2000:DL:557365}.


As the reader may have observed, automata will play a crucial role connecting the logics in this article. The case of \pdl will be no exception. The following recent result by Carreiro and Venema gives a characterization of \pdl as additive-weak parity automata based on the one step language of first-order logic \emph{without} equality:

\begin{fact}[\cite{AiML2014}]\label{fact:autwaofo2pdl}
  $\pdl$ is effectively equivalent to $\AutWA(\ofo)$ over all models.
\end{fact}

In the last part of this article we build on the previous sections and obtain the following characterization result which closes the open questions for \pdl, \wcl and \binfotc.

\begin{theorem}\label{thm:mainbinfotc}
  $\pdl$ is effectively equivalent to $\binfotc/{\bis}$.
\end{theorem}

\begin{theorem}\label{thm:mainwcl}
  $\pdl$ is effectively equivalent to $\wcl/{\bis}$.
\end{theorem}

It is worth remarking that even though the bisimulation-invariant fragments of \binfotc and \wcl coincide, the logics \binfotc and \wcl themselves are not equivalent. This is shown in Section~\ref{ssec:versus}.
Summing up, we give characterizations of \pdl as the bisimulation-invariant fragment of both an extension of first-order logic and a variant of monadic second-order logic.

%% file: prel-tstrees.tex

Throughout this article we fix a set $\props$ of elements that will be called
\emph{proposition letters} and denoted with small Latin letters $p$, $q$, etc.
%
We also fix a set $\acts$ of \emph{atomic actions}.

We use overlined boldface letters to represent sequences, for example a list of variables $\vlist{x} := x_1,\dots,x_n$ or a sequence of sets $\vlist{T} \in \wp(S)^n$. We blur the distinction between sets and sequences: a sequence may be used as a set comprised of the elements of the list; in a similar way, we may assume a fixed order on a set and see it as a list. As an abuse of notation (and to simplify notation), given a map $f:A^{n+m}\to B$ and $\vlist{a} \in A^n$, $\vlist{a}'\in A^m$ we write $f(\vlist{a},\vlist{a}')$ to denote $f(a_1,\dots,a_n,a'_1,\dots,a'_m)$.

Given a binary relation $R \subseteq X \times Y$, for any element $x \in X$,
we indicate with $R[x]$ the set $\compset{ y \in Y \mid (x,y) \in R}$ while $R^+$
and $R^{*}$ are defined respectively as the transitive closure of~$R$ and
the reflexive and transitive closure of~$R$. The set $\Ran(R)$ is defined as $\bigcup_{x\in X}R[x]$.

\paragraph{Transition systems.}
A \emph{labeled transition system} (LTS) on the set of propositions $\props$ and actions $\acts$ is a tuple $\model = \tup{\moddom,R_{\aact\in\acts},\tscolors,s_I}$ where $\moddom$ is the universe or domain of $\model$; the map $\tscolors:\moddom\to\wp(\props)$ is a marking (or coloring) of the elements in $\moddom$;
$R_\aact\subseteq \moddom^2$ is the accessibility relation for the atomic action $\aact\in\acts$;
and $s_I \in \moddom$ is a distinguished node.
We use $R$ without a subscript to denote the binary relation defined as $R := \bigcup_{\aact\in\acts} R_\aact$.
%

%
Observe that a marking ${\tscolors:\moddom\to\wp(\props)}$ can be seen as a valuation $\tsval:\props\to\wp (\moddom)$ given by $\tsval(p) = \{s \in \moddom \mid p\in \tscolors(s)\}$. 
We say that $\model$ is \emph{$p$-free} if $p\notin \props$ or $p\notin \tscolors(s)$ for all $s\in\moddom$.
Given a set of propositions $\oprops$ and $p\notin \oprops$, a \emph{$p$-extension} of a LTS $\model = \tup{\moddom,R_{\aact\in\acts},\tscolors,s_I}$ over $\oprops$
is a transition system $\tup{\moddom,R_{\aact\in\acts},\tscolors',s_I}$ over $\oprops\cup\{p\}$
such that $\tscolors'(s)\setminus\{p\} = \tscolors(s)$ for all $s \in \moddom$.
Given a set $X_p \subseteq \moddom$, we use $\model[p\mapsto X_p]$ to denote the $p$-extension
where $p \in \tscolors'(s)$ iff $s \in X_p$.
%


\paragraph{Trees.}
A \emph{$\props$-tree} $\tmodel$ is a LTS over $\props$ in which every node can
be reached from $s_I$, and every node except $s_I$ has a unique $R$-predecessor;
the distinguished node $s_I$ is called the \emph{root} of $\tmodel$.
A tree is called \emph{strict} when $\Ran(R_\aact) \cap \Ran(R_{\aact'}) = \nada$ for every $\aact\neq\aact'$.

\input{fig-trees}

Each node $s \in \tmoddom$ uniquely defines a subtree of $\tmodel$ with carrier
$R^{*}[s]$ and root $s$. We denote this subtree by ${\tmodel.s}$.
We use the term \emph{tree language} as a synonym of class of trees.

The \emph{tree unravelling} of an LTS $\model$ is given by $\unravel{\model} := \tup{\hat{\moddom},\hat{R}_{\aact\in\acts},\hat{\tscolors},s_I}$ where $\hat{\moddom}$ is the set of ($\acts$-decorated) finite paths $s_I \to_{\aact_1} e_1 \to_{\aact_2} \dots \to_{\aact_n} e_n$ in $\model$ stemming from $s_I$; $\hat{R}_\aact(t,t')$ holds iff $t'$ is an extension of $t$ through the relation $\aact$; and the color of a path $t\in \hat{\moddom}$ is given by the color of its last node in $\moddom$. The $\omega$-unravelling $\omegaunrav{\model}$ of $\model$ is an unravelling which has $\omega$-many copies of each node different from the root.

\begin{remark}
	The ($\omega$-)unravelling of a labelled transition system is a \emph{strict} tree.
\end{remark}

\noindent
Also observe that if there is only one relation $R$, the notion of tree and strict tree coincide.

\paragraph{Chains and generalized chains.}
Let $\model$ be an arbitrary model. A \emph{chain} on $\model$ is a set $X\subseteq \moddom$ such that $(X,R^*)$ is a totally ordered set; i.e., the following conditions are satisfied for every $x,y\in X$:
\begin{description}
	\itemsep 0pt
	\item[(antisymmetry)] if $x R^* y$ and $y R^* x$ then $x=y$,
	\item[(transitivity)] if $x R^* y$ and $y R^* z$ then $x R^* z$,
	\item[(totality)] $x R^* y$ or $y R^* x$.
\end{description}
A \emph{finite chain} is a chain based on a finite set. 
A \emph{generalized chain} is a set $X \subseteq \moddom$ such that $X\subseteq P$, for some path $P$ of $\model$. A \emph{generalized finite chain} is a finite subset $X \subseteq \moddom$ such that $X\subseteq P$, for some finite path $P$ of $\model$.


\begin{proposition}
	Every chain on $\model$ is also a generalized chain on $\model$.
\end{proposition}

\input{fig-chains}

In Fig.~\ref{fig:chains} we show some examples of (generalized) chains and non-chains: 
in (a) the set $X_a=\{2,4\}$ is a finite chain. In (b) the generalized finite chain $X_b=\{1,2,3,4,5,6\}$ is witnessed, among others, by the path $3 \to 4 \to 5 \to 6 \to 1 \to 2$. Observe, however, that $X_b$ is \emph{not} a chain, since there is no possible total ordering of $X_b$ by $R^*$ (antisymmetry fails). In (c) the generalized finite chain $X_c = \{1,3,4,5,6,7,9\}$ is witnessed by the path $1 \to 2 \to \dots \to 7 \to 2 \to 8 \to 9$; observe that the element $2$ is repeated in the path. Again, $X_c$ is also not a finite chain. In the last example (d), the set $X_d = \{1,2,4,6\}$ is not a generalized chain (and hence not a chain).

The following proposition states a useful relationship between chains and generalized chains: on trees this distinction vanishes.

\begin{proposition}\label{prop:chainsontrees}
	On trees, chains and general chains coincide.
\end{proposition}
\begin{proof}
	Observe that every path on a tree $\tmodel$ sits inside some branch of $\tmodel$. Therefore every generalized chain $X$ can be embedded in some branch of $\tmodel$ and hence $(X,R^*)$ will be a total order. The key concept in the background is that on trees there are no cycles.
\end{proof}

%% file: fig-trees.tex

\begin{figure}
\centering
\begin{tikzpicture}[
->,>=latex,
level/.style={level distance=1.5cm, sibling distance=2cm/#1}
]
\begin{scope}[xshift=-4.5cm]
\node {$s_I$}
    child {
        node {$s_0$}
            child {
                node {$s_{00}$}
                edge from parent
                node [near start, left] {$\aact,\aact'$}
            }
            child {
                node {$s_{01}$}
                edge from parent
                node [near start, right] {$\aact'$}
            }
            edge from parent
            node [near start, left] {$\aact$}
    }
    child {
        node {$s_{1}$}
        child {
                node {$s_{10}$}
                edge from parent
                node [near start, left] {$\aact$}
            }
            child {
                node {$s_{11}$}
                edge from parent
                node [near start, right] {$\aact,\aact'$}
            }
        edge from parent
        node [near start, right] {$\aact$}
    };
\end{scope}

\begin{scope}
\node {$s_I$}
    child {
        node {$s_0$}
            child {
                node {$s_{00}$}
                edge from parent
                node [near start, left] {$\aact$}
            }
            child {
                node {$s_{01}$}
                edge from parent
                node [near start, right] {$\aact'$}
            }
            edge from parent
            node [near start, left] {$\aact$}
    }
    child {
        node {$s_{1}$}
        child {
                node {$s_{10}$}
                edge from parent
                node [near start, left] {$\aact$}
            }
            child {
                node {$s_{11}$}
                edge from parent
                node [near start, right] {$\aact'$}
            }
        edge from parent
        node [near start, right] {$\aact$}
    };
\end{scope}

\begin{scope}[xshift=4.5cm]
\node {$s_I$}
    child {
        node (s0) {$s_0$}
            child {
                node (s01) {$s_{01}$}
                edge from parent
                node [near start, right] {$\aact'$}
            }
        edge from parent
        node [near start, left] {$\aact$}
    }
    child {
        node (s1) {$s_{1}$}
            child {
                node {$s_{11}$}
                edge from parent
                node [near start, right] {$\aact'$}
            }
        edge from parent
        node [near start, right] {$\aact$}
    };

    \draw[->] (s1) -- (s01) node [midway,above] {$\aact$};
\end{scope}
\end{tikzpicture}
\caption{A tree, a strict tree, and an LTS which is not a tree.}
\end{figure}
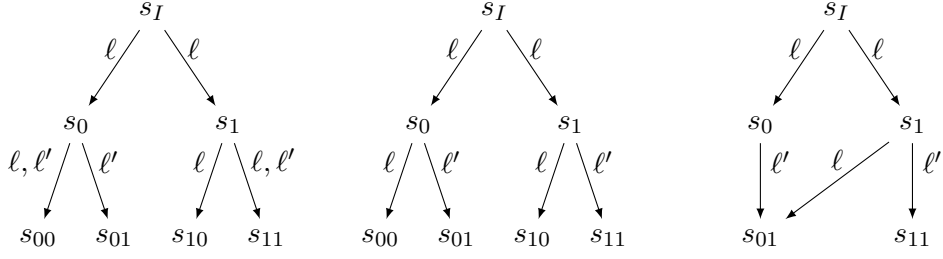




%% file: fig-chains.tex

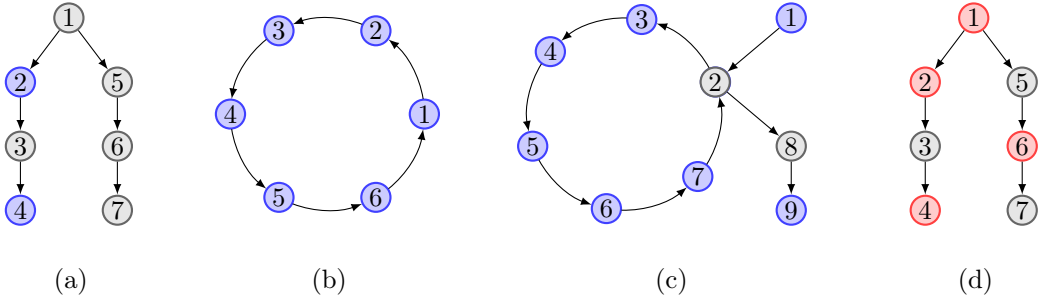
\begin{figure}[h]
\centering
\begin{tikzpicture}[
scale=0.85,
font=\small,
every node/.style={circle,draw=black!60,fill=black!10,thick,inner sep=1pt},
cnode/.style={draw=blue!75,fill=blue!20},
nnode/.style={draw=red!75,fill=red!20},
->,>=latex
]
\node[draw=none,fill=none,yshift=-3.5cm,xshift=-0.4cm] {(a)};
\node[draw=none,fill=none,yshift=-3.5cm,xshift=3cm] {(b)};
\node[draw=none,fill=none,yshift=-3.5cm,xshift=7.5cm] {(c)};
\node[draw=none,fill=none,yshift=-3.5cm,xshift=11.5cm] {(d)};
\begin{scope}[
grow=down,
xshift=-0.5cm,
level/.style={level distance=1cm, sibling distance=1.5cm/#1},
]
\node {$1$}
    child {
        node[cnode] {$2$}
            child {
                node {$3$}
                    child {
                        node[cnode] {$4$}
                    }         
            }         
    }
    child {
        node {$5$}
            child {
                node {$6$}
                    child {
                        node {$7$}
                    }         
            }         
    };
\end{scope}
\begin{scope}[xshift=3.5cm,yshift=-1.5cm]
\def \n {6}
\def \radius {1.5cm}
\def \margin {9} 

\foreach \s in {1,...,\n}
{
  \node[cnode] at ({360/\n * (\s - 1)}:\radius) {$\s$};
  \draw ({360/\n * (\s - 1)+\margin}:\radius) 
    arc ({360/\n * (\s - 1)+\margin}:{360/\n * (\s)-\margin}:\radius);
}
\end{scope}
\begin{scope}[
xshift=8.09cm,
yshift=-1.5cm,
]
\def \n {7}
\def \radius {1.5cm}
\def \margin {9+19.5} 

\foreach \s in {2,3,...,\n}
{
  \node[cnode] at ({(360/(\n-1) * (\s - 2)+19.5)}:\radius) {$\s$};
  \draw[->, >=latex] ({360/(\n-1) * (\s - 2)+\margin}:\radius) 
    arc ({360/(\n-1) * (\s - 2)+\margin}:{360/(\n-1) * (\s-1)-\margin}:\radius);
}
\end{scope}
\begin{scope}[
xshift=9.5cm,
level/.style={level distance=1cm}
]
\node[xshift=1cm,cnode] {$1$}
    child {
        node[xshift=-1cm] {$2$}
            child {
                node[xshift=1cm] {$8$}
                    child {
                        node[cnode] {$9$}
                    }         
            }         
    };
\end{scope}
\begin{scope}[
grow=down,
xshift=13.5cm,
level/.style={level distance=1cm, sibling distance=1.5cm/#1},
]
\node[nnode] {$1$}
    child {
        node[nnode] {$2$}
            child {
                node {$3$}
                    child {
                        node[nnode] {$4$}
                    }         
            }         
    }
    child {
        node {$5$}
            child {
                node[nnode] {$6$}
                    child {
                        node {$7$}
                    }         
            }         
    };
\end{scope}
\end{tikzpicture}
\caption{Examples and counter-examples of chains.}
\label{fig:chains}
\end{figure}

%% file: prel-games.tex

We introduce some terminology and background on infinite games.
All the games that we consider involve two players called \emph{\'Eloise}
($\eloise$) and \emph{Abelard} ($\abelard$).
In some contexts we refer to a player $\player$ to specify a
a generic player in $\{\eloise,\abelard\}$.
Given a set $A$, by $A^*$ and $A^\omega$ we denote respectively the set of
words (finite sequences) and streams (or infinite words) over $A$.

A \emph{board game} $\game$ is a tuple $(G_{\eloise},G_{\abelard},E,\win)$,
where $G_{\eloise}$ and $G_{\abelard}$ are disjoint sets whose union
$G=G_{\eloise}\cup G_{\abelard}$ is called the \emph{board} of $\game$,
$E\subseteq G \times G$ is a binary relation encoding the \emph{admissible
moves}, and $\win \subseteq G^{\omega}$ is a \emph{winning condition}.
An \emph{initialized board game} $\game@u_I$ is a tuple
$(G_{\eloise},G_{\abelard},u_I, E,\win)$ where
$u_I \in G$ is the
\emph{initial position} of the game.

A special case of winning condition is induced by a parity function
$\pmap: G \to \nat$ by defining
$\win_\pmap := \{ g \in G^\omega \mid \text{the \emph{minimum} parity occurring infinitely often in $g$ is \emph{even}}\}$.
In this case, we say that $\game$ is a parity game and sometimes
simply write $\game=(G_{\eloise},G_{\abelard},E,\pmap)$.

Given a board game $\game$, a \emph{match} of $\game$ is simply a path
through the graph $(G,E)$; that is, a sequence $\match = (u_i)_{i< \alpha}$ of
elements of $G$, where $\alpha$ is either $\omega$ or a natural number,
and $(u_i,u_{i+1}) \in E$ for all $i$ with $i+1 < \alpha$.
A match of $\game@u_{I}$ is supposed to start at $u_{I}$.
Given a finite match $\match = (u_i)_{i< k}$ for some $k<\omega$, we call
$\last(\match) := u_{k-1}$ the \emph{last position} of the match; the
player $\player$ such that $\last(\match) \in G_{\player}$ is supposed to move
at this position, and if $E[\last(\match)] = \nada$, we say that
$\player$ \emph{got stuck} in $\match$.
A match $\match$ is called \emph{total} if it is either finite, with one of the
two players getting stuck, or infinite. Matches that are not total are called
\emph{partial}.
Any total match $\match$ is \emph{won} by one of the players:
If $\match$ is finite, then it is won by the opponent of the player who gets stuck.
Otherwise, if $\match$ is infinite, the winner is $\eloise$ if $\match \in
\win$, and $\abelard$ if $\match \not\in \win$.

Given a board game $\game$ and a player $\player$, let $\pmatches{G}{\player}$ denote
the set of partial matches of $\game$ whose last position belongs to player
$\player$.
A \emph{strategy for $\player$} is a function $f:\pmatches{G}{\player}\to G$.
A match $\match  = (u_i)_{i< \alpha}$ of $\game$ is
\emph{$f$-guided} if for each $i < \alpha$ such that $u_i \in G_{\player}$ we
have that $u_{i+1} = f(u_0,\dots,u_i)$.
Let $u \in G$ and a $f$ be a strategy for $\player$.
We say that $f$ is a \emph{surviving strategy} for $\player$ in $\game@u$ if
\begin{enumerate}
  \item[(i)] For each $f$-guided partial match $\match$ of $\game@u$, if $\last(\match)$
  is in $G_{\player}$ then $f(\match)$ is legitimate, that is,
  $(\last(\match),f(\match)) \in E$.
\end{enumerate}
We say that $f$ is a \emph{winning strategy} for $\player$ in $\game@u$ if, additionally, 
%
\begin{enumerate}
  \item[(ii)] $\player$ wins each $f$-guided total match of $\game@u$.
\end{enumerate}
If $\player$ has a winning winning strategy for $\game@u$ then $u$ is called a \emph{winning position} for $\player$ in $\game$.
The set of positions of $\game$ that are winning for $\player$ is denoted by $\win_{\player}(\game)$.
A strategy $f$ is called \emph{positional} if $f(\match) = f(\match^{\prime})$ for each $\match,\match^{\prime} \in \Dom(f)$ with $\last(\match) = \last(\match^{\prime})$.
A board game $\game$ with board $G$ is \emph{determined} if $G = \win_{\eloise}(\game) \cup \win_{\abelard}(\game)$, that is, each $u \in G$ is a winning position for one of the two players.

\begin{fact}[Positional Determinacy of Parity Games~\cite{EmersonJ91,Mostowski91Games}]
For each parity game $\game$, there are positional strategies $f_{\eloise}$
and $f_{\abelard}$ respectively for player $\eloise$ and $\abelard$, such that
for every position $u \in G$ there is a player $\player$ such that $f_{\player}$ is a
winning strategy for $\player$ in $\game@u$.
\end{fact}
From now on, we always assume that each strategy we work with in parity games
is positional. Moreover, we will think of a positional strategy $f_\player$ for player $\player$
as a function $f_\player:G_\player\to G$.

%% file: prel-parityaut.tex

We recall the definition of a parity automaton, adapted to our setting.
Since we will be comparing parity automata defined in terms of various
one-step languages, it makes sense to make the following abstraction.

\begin{definition}
Given a set $A$ and sorts $\sorts = \{\asort_1,\dots,\asort_n\}$,
we define a \emph{one-step model} to be a tuple $\osmodel = (D_{\asort_1},\dots,D_{\asort_n},\val)$
consisting of a domain $D$ and sets $D_{\asort_1},\dots,D_{\asort_n}$ such that $\bigcup_\asort D_\asort = D$,
and a valuation $\val: A \to \wp D$. A one-step model is called \emph{strict} when the sets $D_{\asort\in\sorts}$ are pairwise disjoint, that is, when $D_{\asort_1},\dots,D_{\asort_n}$ is a partition of $D$.
Depending on context, elements of $A$ will be called \emph{monadic predicates}, \emph{names}
or \emph{propositional variables}. When the sets $D_{\asort\in\sorts}$ are not relevant we will just write the one-step model as $(D,\val)$. The class of all one-step models will be denoted by $\umods$ and the class of all strict one-step models will be denoted by $\sumods$.

\input{fig-osmodel-1}

A \emph{(multi-sorted) one-step language} is a map $\llang$ assigning to each set $A$ and sorts $\sorts$, a set $\llang(A,\sorts)$ of objects called \emph{one-step formulas} over $A$ (on sorts $\sorts$). When the sorts are understood from context (or fixed) we simply write $\llang(A)$ instead of $\llang(A,\sorts)$.
We require that $\llang(\bigcap_{i} A_{i},\sorts) = \bigcap_{i} \llang(A_{i},\sorts)$,
so that for each $\varphi \in \llang(A,\sorts)$ there is a smallest $A_{\varphi} \subseteq A$ such
that $\varphi \in \llang(A_{\varphi},\sorts)$; this $A_{\varphi}$ is the set of names that \emph{occur} in $\varphi$.

We assume that one-step languages come with a \emph{truth}
relation: given a one-step model $\osmodel$, a formula $\varphi \in \llang$
is either \emph{true} or \emph{false} in $\osmodel$, denoted by,
respectively, $\osmodel \models \varphi$ and $\osmodel \not\models \varphi$.
We also assume that $\llang$ has a \emph{positive fragment} $\llang^+$
characterizing monotonicity in the sense that a formula $\varphi \in \llang(A,\sorts)$ is
(semantically) monotone iff it is equivalent to a formula $\varphi' \in
\llang^+(A,\sorts)$.
\end{definition}

The one-step languages $\llang$ featuring in this paper all are induced by
well-known logics. Examples include (multi-sorted) monadic first-order logic (with and without equality)
and fragments of these languages.

\begin{definition}
\label{def:parity_aut}
A \emph{parity automaton} based on the one-step language $\llang$, actions $\acts$ and 
alphabet $\wp(\props)$ is a tuple $\aut = \tup{A,\tmap,\pmap,a_I}$ such that $A$ is a
finite set of states of the automaton, $a_I \in A$ is the initial state,
$\tmap: A\times \wp(\props) \to \llang^+(A,\acts)$
is the transition map, and $\pmap: A \to \nat$ is the parity map.
The collection of such automata will be denoted by $\Aut(\llang,\props,\acts)$.
For the rest of the article we fix the set of actions $\acts$ and omit it in our notation,
we also omit the set $\props$ when clear from context or irrelevant.
\end{definition}

Acceptance and rejection of a transition system by an automaton is defined
in terms of the following parity game.

\begin{definition}
Given an automaton $\aut = \tup{A,\tmap,\pmap,a_I}$ in $\Aut(\llang,\props)$ and a $\props$-transition
system $\model = \tup{\moddom,R_{\aact\in\acts},\tscolors,s_I}$, the \emph{acceptance game}
$\agame(\aut,\model)$ of $\aut$ on $\model$ is the parity game defined
according to the rules of the following table.
%
\begin{center}
\small
\begin{tabular}{|l|c|l|c|} \hline
Position & Pl'r & Admissible moves & Parity \\
\hline
    $(a,s) \in A \times \moddom$
  & $\eloise$
  & $\{\val : A \to \wp(R[s]) \mid (R_{\aact_1}[s],\dots,R_{\aact_n}[s],\val) \models \tmap (a, \tscolors(s)) \}$
  & $\pmap(a)$ 
\\
    $\val : A \rightarrow \wp(\moddom)$
  & $\abelard$
  & $\{(b,t) \mid t \in \val(b)\}$
  & $\max(\pmap[A])$
\\ \hline
 \end{tabular}
\end{center}
A transition system $\model$ is \emph{accepted} by $\aut$ if $\exists$ has
a winning strategy in $\agame(\aut,\model)@(a_I,s_I)$, and \emph{rejected}
if $(a_I,s_I)$ is a winning position for $\abelard$.
\end{definition}

Many properties of parity automata are determined at the one-step level.
An important example concerns the notion of complementation.

\begin{definition}
\label{d:bdual1}
Two one-step formulas $\varphi$ and $\psi$ are each other's \emph{Boolean dual}
if for every structure $(D,\val)$ we have
\[
(D,\val) \models \varphi \text{ iff } (D,\val^{c}) \not\models \psi,
\]
where $\val^{c}$ is the valuation given by $\val^{c}(a) \mathrel{:=} D
\setminus \val(a)$, for all $a$.
A one-step language $\llang$ is \emph{closed under Boolean duals} if for every
set $A$, each formula $\varphi \in \llang(A)$ has a Boolean dual $\dual{\varphi}
\in \llang(A)$.
\end{definition}

Following ideas from~\cite{Muller1987267,DBLP:conf/calco/KissigV09}, we can use Boolean duals, together with a
\emph{role switch} between $\abelard$ and $\eloise$, in order to define a
negation or complementation operation on automata.

\begin{definition}
\label{d:caut}
Assume that, for some one-step language $\llang$, the map $\dual{(-)}$
provides, for each set $A$, a Boolean dual $\dual{\varphi} \in \llang(A)$ for each
$\varphi \in \llang(A)$.
Given $\aut = \tup{A,\tmap,\pmap,a_I}$ in $\Aut(\llang)$ we define its
\emph{complement} $\dual{\aut}$ as the automaton
$\tup{A,\dual{\tmap},\dual{\pmap},a_I}$
where $\dual{\tmap}(a,c) := \dual{(\tmap(a,c))}$, and $\dual{\pmap}(a)
:= 1 + \pmap(a)$, for all $a \in A$ and $c \in \wp(\props)$.
\end{definition}

\begin{proposition}
\label{PROP_complementation}
Let $\llang$ and $\dual{(-)}$ be as in the previous definition.
For each automaton $\aut \in \Aut(\llang)$ and each transition structure
$\model$ we have that
\[
\dual{\aut} \text{ accepts } \model
\quad\text{iff}\quad
\aut \text{ rejects } \model.
\]
\end{proposition}

The proof of Proposition~\ref{PROP_complementation} is based on the fact
that the power of $\eloise$ in $\agame(\dual{\aut},\model)$ is the same
as that of $\abelard$ in $\agame(\aut,\model)$.

As an immediate consequence of this proposition, one may show that if the
one-step language $\llang$ is closed under Boolean duals, then the class
$\Aut(\llang)$ is closed under taking complementation.
Further on we will use Proposition~\ref{PROP_complementation} to show that
the same may apply to some subsets of $\Aut(\llang)$.

%% file: fig-osmodel-1.tex

\begin{figure}[h]
\centering
\begin{tikzpicture}[thick,inner sep=1pt,font=\small,node distance=1.5cm]
	\begin{scope}[every node/.style={circle,draw=black!75,fill=white}]
	\node (n2) {$1$};
	\node[right of=n2] (n3) {$2$};
	\node[right of=n3] (n4) {$3$};
	\node[right of=n4] (n5) {$4$};
	\node[right of=n5] (n6) {$5$};
	\node[right of=n6] (n7) {$6$};
	\end{scope}

	\node[below=9pt of n2] {$\{\}$};
	\node[below=9pt of n3] {$\{b\}$};
	\node[below=9pt of n4] {$\{a\}$};
	\node[below=9pt of n5] {$\{a,b\}$};
	\node[below=9pt of n6] {$\{\}$};
	\node[below=9pt of n7] {$\{a\}$};

	\node[above=9pt of n2] {$\asort_1$};
	\node[above=9pt of n4] {$\asort_2$};
	\node[above=9pt of n6] {$\asort_3$};

	\begin{pgfonlayer}{background}
	\filldraw [line width=4mm,join=round,black]
		(n2.north  -| n2.west) rectangle (n2.south  -| n2.east);

	\filldraw [line width=4mm,join=round,black!50]
		(n3.north -| n3.west) rectangle (n5.south  -| n5.east);
	\filldraw [line width=3mm,join=round,white]
		(n3.north -| n3.west) rectangle (n5.south  -| n5.east);

	\filldraw [line width=4mm,join=round,black!50]
		(n5.north  -| n5.west) rectangle (n7.south  -| n7.east);
	\filldraw [line width=3mm,join=round,black!25]
		(n5.north  -| n5.west) rectangle (n5.south  -| n5.east);

	\end{pgfonlayer}
\end{tikzpicture}
\caption{One-step model with sorts (above) and valuation (below).}
\end{figure}
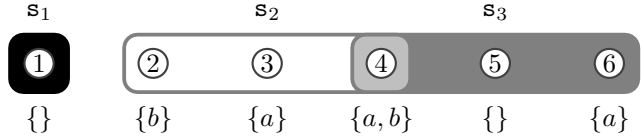

%% file: prel-pdl.tex

\begin{definition}
The formulas of Propositional Dynamic Logic (\pdl) on propositions $\props$ and atomic actions $\acts$ are given by mutual induction on formulas and programs
\begin{align*}
    \varphi &::= p \mid \lnot\varphi \mid \varphi \lor \varphi \mid \tup{\prog}\varphi \\ 
    \prog &::= \aact \mid \prog\seq\prog \mid \prog\choice\prog \mid \prog^* \mid \varphi?
\end{align*}
where $p\in \props$ and $\aact \in \acts$. 
%
\end{definition}

To give the semantics of \pdl we take a standard approach and define it together with the relation~$\relprog^\model_\prog$ induced by a program $\prog$ on a model $\model = \tup{\moddom, R_{\aact\in\acts}, \tscolors, s_I}$, by mutual induction
\begin{align*}
    \relprog^\model_\aact & := R_\aact &
    \relprog^\model_{\prog\seq\pprog} & := \relprog^\model_\prog \circ \relprog^\model_\pprog \\
    \relprog^\model_{\prog\choice\pprog} & := \relprog^\model_\prog \cup \relprog^\model_\pprog &
    \relprog^\model_{\prog^*}& := (\relprog^\model_\prog)^*\\
    \relprog^\model_{\varphi?} & := \{(s,s) \in S\times S \mid \model[s_I\mapsto s] \mmodels \varphi\}.
\end{align*}
The semantics of \pdl is then given as usual on the boolean operators and as follows on modal operators and propositions.
\begin{align*}
    \model \mmodels p & \quad\text{iff}\quad s_I \in \tsval(p)\\
    \model \mmodels \tup{\prog}\varphi &\quad\text{iff}\quad \text{there exists } t \in S \text{ such that } \relprog^\model_\prog(s,t) \text{ and } \model[s_I\mapsto t] \mmodels \varphi
\end{align*}
We drop the superscript in $\relprog^\model_\prog$ when it is clear from context. 

%% file: prel-muml.tex

The language of the modal $\mu$-calculus ($\muML$) on $\props$ and $\acts$ is given by the following grammar:
\[
    \varphi ::= q \mid \lnot\varphi \mid \varphi \lor \varphi \mid
    \tup{\aact}\varphi \mid \mu p.\varphi
\]
where $p,q \in \props$, $\aact\in\acts$ and $p$ is positive in $\varphi$ (i.e., $p$ is under an even number of negations). We use the standard convention that no variable is both free and bound in a formula and that every bound variable is fresh.
Let $p$ be a bound variable occuring in some formula $\varphi \in \muML$, we use $\delta_p$ to denote the binding definition of $p$, that is, the formula such that either $\mu p.\delta_p$ is a subformula of $\varphi$.

The semantics of this language is completely standard. Let $\model = \tup{\moddom,R_{\aact\in\acts},\tscolors, s_I}$ be a transition system and $\varphi \in \muML$. We inductively define the \emph{meaning} $\ext{\varphi}^{\model}$ which includes the following clause for the fixpoint operator: 
\[
  \ext{\mu p.\psi}^{\model} := \bigcap \{X \subseteq \moddom \mid X \supseteq \ext{\psi}^{\model[p \mapsto X]} \}
\]
%
We say that $\varphi$ is \emph{true} in $\model$ (notation $\model \mmodels \varphi$) iff $s_I \in \ext{\varphi}^{\model}$.

%% file: prel-bisimulation.tex

Bisimulation is a notion of behavioral equivalence between processes.
For the case of transition systems, it is formally defined as follows.

\begin{definition}
Let $\model = \tup{\moddom, R_{\aact\in\acts}, \tscolors, s_I}$ and
$\model' = \tup{\moddom', R'_{\aact\in\acts}, \tscolors', s'_I}$ be transition systems.
A \emph{bisimulation} is a relation $Z \subseteq \moddom \times \moddom'$
such that for all $(t,t') \in Z$ the following holds:
\begin{description}
  \itemsep 0 pt
  \item[(atom)] $p \in \tscolors(t)$ iff $p \in \tscolors'(t')$ for all $p\in\props$;
  \item[(forth)] for all $\aact\in\acts$ and $s \in R_{\aact}[t]$ there is $s' \in R'_{\aact}[t']$ such that $(s,s') \in Z$;
  \item[(back)] for all $\aact\in\acts$ and $s' \in R'_{\aact}[t']$ there is $s \in R_{\aact}[t]$ such that $(s,s') \in Z$.
\end{description}
Two transition systems $\model$ and $\model'$ are
\emph{bisimilar} (denoted $\model \bis \model'$) if there is a
bisimulation $Z \subseteq \moddom \times \moddom'$ containing $(s_I,s'_I)$.
\end{definition}

The following fact about tree unravellings is central in many theorems of modal logics. It will also play an important role in this paper.

\begin{fact}\label{prop:tree_unrav}
$\model$ and its unravelling $\unravel{\model}$ are bisimilar, for every transition system $\model$.
\end{fact}

As observed in the introduction, a key concept in this paper is that of bisimulation invariance. Formally, it is defined as follows for an arbitrary language $\llang$:

\begin{definition}
A formula $\varphi \in \llang$ is \emph{bisimulation-invariant} if $\model \bis
\model'$ implies that
$\model \mmodels \varphi$ iff $\model' \mmodels \varphi$,
for all $\model$ and $\model'$.
\end{definition}

\begin{fact}
Every formula of $\muML$, and therefore of \pdl, is bisimulation-invariant.
\end{fact}

%% file: prel-wcl.tex


The non-weak version of chain logic (CL) was defined in~\cite{Thomas96languagesautomata}, and studied in the context of \emph{trees}. As we said before, this logic is a variant of MSO which changes the usual second-order quantifier to the following quantifier over chains:
\[
\tmodel \models \existsc p.\varphi \quad\text{ iff }\quad  \text{there is a \emph{chain} $X\subseteq \tmoddom$ such that $\tmodel[p\mapsto X] \models \varphi$}.
\]
In this paper we will only work with a \emph{weak} version of CL, that is, the quantification will be over \emph{finite} chains. On the other hand, we also want to consider this logic on the class of \emph{all} models. To give a definition of weak chain logic we adhere to what we think is the ``spirit'' of the definition of CL, as opposed to the ``letter.'' As observed in Section~\ref{ssec:prelim_trees}, the concept of chain on trees coincides with that of ``subset of a path.'' Therefore, on the class of all models, we choose to define the weak second-order quantifier as:
\[
\model \models \existswc p.\varphi \quad\text{ iff }\quad  \text{there is a \emph{generalized finite chain} $X\subseteq \moddom$ such that $\model[p\mapsto X] \models \varphi$}.
\]
More formally, the weak version of chain logic is given as follows.

\begin{definition}
The one-sorted \emph{weak chain logic} (\wcl) on a set of predicates $\props$ and actions $\acts$ is given by
\[
\varphi ::= \here{p} \mid p \inc q \mid R_\aact(p,q) \mid \lnot\varphi \mid \varphi\lor\varphi \mid \existswc p.\varphi
\]
where $p,q \in \props$ and $\aact\in\acts$. We denote this logic by $\wcl(\props,\acts)$ and omit $\props$ and $\acts$ when clear from context.
We adopt the standard convention that no letter is both free and bound in $\varphi$.
\end{definition}

\begin{definition}\label{def:wcl}
Let $\model = \tup{\moddom,R_{\aact\in\acts},\tscolors, s_I}$ be a labelled transition system. The semantics of \wcl is defined as follows:
\begin{align*}
\model \models \here{p} & \quad\text{ iff }\quad  \tsval(p) = \compset{s_I} \\
\model \models p \inc q & \quad\text{ iff }\quad  \tsval(p) \subseteq \tsval(q) \\
\model \models R_\aact(p,q) & \quad\text{ iff }\quad  \text{for all $s\in \tsval(p)$ there is $t\in \tsval(q)$ such that $s{R_\aact}t$} \\
\model \models \lnot\varphi & \quad\text{ iff }\quad  \model \not\models \varphi \\
\model \models \varphi\lor\psi & \quad\text{ iff }\quad  \model \models \varphi \text{ or } \model \models \psi \\
\model \models \existswc p.\varphi & \quad\text{ iff }\quad  \text{there is a \emph{generalized finite chain} $X\subseteq \moddom$ such that $\model[p\mapsto X] \models \varphi$}.
\end{align*}
%
\end{definition}

\paragraph{A digression on second-order languages.}
The reader may have expected a more standard two-sorted language for second-order logic, for example given by
\[
\varphi ::= p(x)
\mid R_\aact(x,y)
\mid x \foeq y
\mid \neg \varphi
\mid \varphi \lor \varphi
\mid \exists x.\varphi
\mid \existswc p.\varphi
\]%
where $p \in \props$, $\aact\in\acts$, $x,y \in \fovar$ (individual variables), 
and $\foeq$ is the symbol for equality. We call this language $2\wcl$.
This semantics of this language is completely standard, with $\exists x$ denoting first-order quantification (that is, quantification over individual states) and $\existswc p$ denoting second-order quantification (that is, in this case, quantification over generalized finite chains).
Both definitions can be proved to be equivalent, however, we choose to keep Definition~\ref{def:wcl} as it will be better suited to work with in the context of automata. 

Formulas of this languages are interpreted over (non-pointed) models $\npmodel = \tup{\npmoddom, R_{\aact\in\acts}, \tscolors}$ with an assignment, that is, a map $\ass: \fovar \to \npmoddom$ interpreting the individual variables as elements of $\npmoddom$.  The key point is that $\wcl$ can interpret $2\wcl$ by encoding every individual variable $x\in\fovar$ as a set variable $p_x$ denoting a \emph{singleton}. The following is a more detailed proof of the remark found in~\cite{Ven12}, adapted for $\wcl$.

\begin{proposition}
  There is a translation $(-)^t:2\wcl(\props,\acts) \to \wcl(\props\uplus \props_X,\acts)$ such that
  %
  \[
  \npmodel,\ass \models \varphi \quad\text{iff}\quad \npmodel[p_{x\in\fovar} \mapsto \{\ass(x)\}] \models \varphi^t,
  \]
  where $\props_X := \{p_x \mid x\in\fovar\}$.
\end{proposition}
\noindent
The translation is inductively defined as follows:
\begin{itemize}
  \itemsep 0pt
  \item $(p(x))^t := p_x \inc p$,
  \item $(R_\aact(x,y))^t := R_\aact(p_x,p_y)$,
  \item $(x \foeq y)^t := p_x \inc p_y \land p_y \inc p_x$,
  \item Negation and disjunction as usual,
  \item $(\existswc p.\varphi)^t := \existswc p.\varphi^t$,
  \item $(\exists x.\varphi)^t := \existswc p_x. \texttt{singleton}(p_x) \land \varphi^t$
\end{itemize}
where the translation crucially uses the predicates
\begin{align*}
  \texttt{empty}(p) & := \forallwc q. (p \inc q) \\
  \texttt{singleton}(p) & := \forallwc q. (q \inc q \to (\texttt{empty}(q) \lor p\inc q))
\end{align*}
Observe that the translation does not use the operator $\here{p}$ and hence is well-defined on non-pointed models. We finish by proving the following claim.
\begin{claimfirst}
  For every $\varphi\in 2\wcl$ we have $\npmodel,g \models \varphi \text{ iff } \npmodel[p_{x\in\fovar} \mapsto \{\ass(x)\}] \models \varphi^t$.
\end{claimfirst}
\begin{pfclaim}
  We prove the inductive step for the first-order quantification.\\
  \fbox{$\Rightarrow$} Suppose $\npmodel,g \models \exists x.\varphi$ then there is $s\in\npmoddom$ such that $\npmodel,g[x\mapsto s] \models \varphi$. By inductive hypothesis then there exists $s\in\npmoddom$ such that $\npmodel[p_{x} \mapsto \{s\}; p_{y\neq x} \mapsto \{\ass(y)\}] \models \varphi^t$. This clearly implies that $\npmodel[p_{x\in\fovar} \mapsto \{\ass(x)\}] \models \existswc p_x. \texttt{singleton}(p_x) \land \varphi^t$.\\
  \fbox{$\Leftarrow$} This direction is very similar.
\end{pfclaim}

\fcnote{later consider comparing with other logics like MSO and PMSO, see remarks in source.}



%% file: prel-mufoe.tex

In this subsection we give an extension of $\foe$ with a unary fixed point operator. This extension is known in the literature as FO(LFP$^1$) but we will call it $\mufoe$.


As usual with (extensions of) first-order logic, $\mufoe$ will be interpreted over models with an assignment. See Section~\ref{sec:wcl} ($2\wcl$ \textsl{vs.} $\wcl$) for a discussion on how languages with individual variables fit in our setting. Also, because of the presence of individual variables, the syntax and semantics of the fixpoint operator is considerably more involved than for the modal $\mu$-calculus.

\begin{definition}
The \emph{first-order logic with equality and unary fixpoints} ($\mufoe$) on a set of predicates $\props$, actions $\acts$ and individual variables $\fovar$ is given by
\[
\varphi ::= q(x) \mid R_\aact(x,y) \mid x \foeq y \mid \exists x.\varphi \mid \lnot\varphi \mid \varphi \lor \varphi \mid [\lfp_{p{:}x}.\varphi(p,x)](z)
\]
where $p,q\in\props$, $\aact\in\acts$ and $x,y\in\fovar$.
Observe that $z$ is free in the fixpoint clause and the fixpoint operator binds the designated variables $x$ and $p$.
\end{definition}


The semantics of the fixpoint formula $[\lfp_{p{:}x}.\varphi(p,x)](z)$ is the expected one~\cite{Chandra198299}. Given a model $\npmodel$ and an assigment $\ass$, the map $F^\varphi_{p{:}x}:\wp(\npmoddom)\to \wp(\npmoddom)$ is defined as
\[
F^\varphi_{p{:}x}(Y) := \{t \in \npmoddom \mid \npmodel[p \mapsto Y],\ass[x\mapsto t] \models \varphi(p, x) \}.
\]
The formula $\npmodel,\ass \models [\lfp_{p{:}x}.\varphi(p,x)](z)$ is defined to hold iff $\ass(z) \in \lfp(F^\varphi_{p{:}x})$. 

\begin{remark}\label{rem:parameters}
	Suppose that a formula $\varphi \in \mufoe$ has free variables $FV(\varphi) = \{x,\vlist{y}\}$. If we consider the fixpoint formula $\psi := [\lfp_{p{:}x}.\varphi(p,x)](z)$ then $\psi$ would have as free variables $FV(\psi) = \{z,\vlist{y}\}$. The free variables of $\varphi$ which are not bound by the fixpoint (in this case~$\vlist{y}$) are called the \emph{parameters} of the fixpoint.

	Parameters can always be avoided at the expense of increasing the arity of the fixpoint~\cite[p.~184]{Libkin_ElsFiniteModelTheory}. That is, for example, take the fixpoint over a relation $P(x_1,\dots,x_n)$ instead of just a predicate $p$. However, in this paper we will only consider fixpoints over unray predicates, and therefore we will allow the use of parameters.
\end{remark}

%% file: prel-fotc.tex

\begin{definition}
The syntax of first-order logic extended with reflexive-transitive closure of \emph{binary} formulas is given by the following grammar:
\[
\varphi ::= p(x) \mid x\foeq y \mid R_\aact(x,y) \mid \lnot \varphi \mid \varphi \lor \varphi \mid \exists x.\varphi \mid [\tc_{x,y}.\varphi(x,y)](z,w)
\]
where $p,q\in\props$, $\aact\in\acts$ and $x,y,z,w\in\fovar$. We denote this logic by \binfotc. The semantics are standard for the first-order part and as follows for the new operator:
\[
\npmodel,\ass \models [\tc_{x,y}.\varphi(x,y)](u,v) \quad\text{iff}\quad (\ass(u),\ass(v)) \in R^*_\varphi
\]
where $R_\varphi := \{(s_x,s_y) \in \npmoddom\times\npmoddom \mid \npmodel,\ass[x\mapsto s_x, y\mapsto s_y] \models \varphi\}$.
\end{definition}

\begin{remark}[{\cite[Example~3.3.8]{FMTapp}}]\label{rem:tcinlfp}
  The meaning of the formula $[\tc_{x,y}.\varphi(x,y)](u,v)$ can be rephrased as saying that $v\in \varphi^*[u]$; that is, $v$ is a $\varphi$-descendant of~$u$. This can be expressed with the formula $[\lfp_{p{:}y}.y\foeq u \lor (\exists x. p(x) \land \varphi(x,y))](v).$
\end{remark}

%% file: prel-convention.tex

The following table works as a summary of the most used notation in this article. It should be taken as a set of general rules from which we try to divert as little as possible.

\begin{center}
\begin{tabular}{|l|l|}
	\hline
	\textbf{Concept} & \textbf{Notation}\\
	\hline
	\hline
	Transition system (pointed model) & $\model = \tup{\moddom,R_{\aact\in\acts},\tscolors,s_I}$\\
	Tree (pointed tree) & $\tmodel = \tup{\tmoddom,R_{\aact\in\acts},\tscolors,s_I}$\\
	Model (non-pointed) & $\npmodel = \tup{\npmoddom,R_{\aact\in\acts},\tscolors}$\\
	One-step model & $\osmodel = (D,\val:A\to\wp(D))$\\
	Automaton & $\aut,\baut,\dots$\\
	\hline
	\hline
	Formula & $\varphi, \psi, \alpha, \beta, \xi, \chi,\dots$ $\Phi,\Psi,\dots$\\
	Set & $A, B, C, D,\dots$ $X, Y, Z, W,\dots$\\
	Sequence of objects & $\vlist{x}, \vlist{y}, \dots$ $\vlist{a}, \vlist{b}, \dots$ $\vlist{X}, \vlist{Y}, \dots$\\
	Propositional variable & $p, q, r, \dots$\\
	Individual (first-order) variable & $x, y, z, w, \dots$\\
	Second-order (set) variable & $X,Y,Z,W,\dots$ $p,q,r,\dots$\\
	\hline
	\hline
	Assignment (of individual variables) & $\ass:\fovar \to \npmoddom$\\
	Valuation (of names/propositions) & $\val:A\to \wp(D)$, $\tsval:\props\to\wp(\moddom)$\\
	Marking/coloring & $\val^\natural:D\to \wp(A)$, $\tscolors:\moddom\to\wp(\props)$\\
	\hline
\end{tabular}
\end{center}

%% file: fotc-intro.tex

In this section we prove that \binfotc is equivalent to \mucafoe, the fragment of \mufoe where the least fixpoint operator is restricted to completely additive formulas. That is,

\begin{trivlist}
	\item[]\textbf{Theorem~\ref{thm:fotcmucafoe}.} {\it $\binfotc \equiv \mucafoe$ over all models.}
\end{trivlist}


\fcwarning{See commented remark on Thomason}

We start by giving a characterization of completely additive \emph{maps}, which we will later use as a tool to characterize \binfotc.

%% file: fixpoint-caf.tex

\begin{definition}
A function $F:\wp(\moddom)^n \to \wp(\moddom)$ is \emph{completely additive in the $i^\text{th}$-coordinate} if 
for every \emph{non-empty} family of subsets $\{Y_i \subseteq S\}_{i\in I}$ and $X_1,\dots,X_n \subseteq S$ it satisfies:
%
\[
F(X_1, \dots, \bigcup_i Y_i, \dots, X_n) =
\bigcup_i F(X_1, \dots, Y_i, \dots, X_n).
\]
We say that $F$ is \emph{completely additive} (sometimes called completely additive in the product) if for every \emph{non-empty} family $\{\vlist{P}_i \in \wp(\moddom)^n\}_{i\in I}$ it satisfies:
\[
F(\bigcup_i \vlist{P}_i) =
\bigcup_i F(\vlist{P}_i).
\]
%
\end{definition}
\begin{remark}
	Observe that complete additivity in the $i^\text{th}$-coordinate implies monotonicity in the $i^\text{th}$-coordinate. Also, if a function is completely additive then it is so in every coordinate; however, the converse does not hold. A simple counterexample is $F(A,B) = A \cap B$.
\end{remark}

An alternative characterization of complete additivity in the $i^\text{th}$-coordinate is given by asking that $F$ restricts to singletons (or the empty set) in that coordinate. More formally, 
\[
F(X_1, \dots, Y_i, \dots, X_n) =
F(X_1, \dots, \nada, \dots, X_n) \cup
\bigcup_{y\in Y_i} F(X_1, \dots, \{y\}, \dots, X_n).
\]
Along the same line, we can give an alternative characterization of complete additivity in the product. First, we need the following definition.

\begin{definition}
Given $\vlist{X}\in \wp(\moddom)^n$ we say that $\vlist{Y} \in \wp(\moddom)^n$ is an atom of $\vlist{X}$ if and only if ${\vlist{Y} = (\nada,\dots,\{x_i\},\dots,\nada)}$ for some element $x_i \in X_i$ standing at some coordinate $i$. We say that $\vlist{Q}$ is a \emph{quasi-atom} if it is an atom or $\vlist{Q} = (\nada,\dots,\nada)$.
\end{definition}

In this terminology, we can formulate the concept of complete additivity in the product by asking that $F$ restricts to quasi-atoms; i.e., for every $\vlist{P}\in \wp(\moddom)^n$, it should satisfy:
\[
F(\vlist{P}) =
\bigcup \{F(\vlist{Q}) \mid \vlist{Q} \text{ is a quasi-atom of $\vlist{P}$}\}.
\]
Another way to read this last definition is that every $s \in F(\vlist{P})$ only depends on at most one singleton on one of the coordinates.


\paragraph{Finite approximants of completely additive maps.}
Given a monotone map $F:\wp(\moddom) \to \wp(\moddom)$, the approximants of the least fixpoint of $F$ are the sets $F^\alpha(\nada) \subseteq \moddom$, where $\alpha$ is an ordinal. The map $F^\alpha$ is intuitively the $\alpha$-fold composition of $F$. Formally,

\begin{itemize}
	\itemsep 0pt
	\item $F^0(X) := \nada$,
	\item $F^{\alpha+1}(X) := F(F^\alpha(X))$,
	\item $F^\lambda(X) := \bigcup_{\alpha<\lambda} F^\alpha(X)$ for limit ordinals $\lambda$.
\end{itemize}

\noindent The sets $F^\alpha(\nada)$ are called approximants because of the following fact.

\begin{fact}
	For every $s\in\moddom$ we have that $s\in \lfp(F)$ iff $s\in F^\beta(\nada)$ for some ordinal $\beta$.
\end{fact}

Moreover, this approximation starts at $F^0(\nada) = \nada$ and grows strictly until it stabilizes for some ordinal $\beta$. This ordinal is called the closure or unfolding ordinal of $F$. Moreover, completely additive maps satisfy nicer properties regarding the approximants.

\begin{fact}\label{fact:constructivity}
	If $F$ is completely additive then it is constructive, i.e., $\lfp(F) = \bigcup_{i\in\nat} F^i(\nada)$.
\end{fact}

Suppose now that we are given a map $G(X,Y)$ which is completely additive. A natural question is whether the (least) fixpoint operation preserves complete additivity. That is, whether $G'(Y) := \lfp_X.G(X,Y)$ is completely additive as well. To answer that question, we will have to look at the finite approximants of $F(X) = G(X,Y)$ where $Y$ is now fixed. In this subsection we give a fairly technical and precise characterization of the finite approximants of completely additive maps, and use it to prove the following interesting theorem.

Let $F:\wp(\moddom) \to \wp(\moddom)$ and $Y\subseteq S$. We define the restriction of $F$ to $Y$ as the function $F_{\resto Y} : \wp(Y) \to \wp(Y)$ given by $F_{\resto Y}(X) := F(X) \cap Y$.
%
\begin{theorem}\label{thm:propscafmap}~
	\begin{enumerate}[(1)]
		\item If $G(X,\vlist{Y})$ is completely additive then so is $H(\vlist{Y}) := \lfp_X.G(X,\vlist{Y})$.
		\item For every completely additive functional $F:\wp(\npmoddom)\to\wp(\npmoddom)$ and $s\in\npmoddom$ we have that
		\[
		s \in \lfp(F) \quad\text{iff}\quad \text{there exists $Y$ such that } s \in \lfp(F_{\resto Y})
		\]
		where $Y = \{t_1,\dots,t_k\}$ satisfies $t_{i+1} \in F_{\resto Y}^{i+1}(\nada)\setminus F_{\resto Y}^i(\nada)$ and $t_k=s$.
	\end{enumerate}
\end{theorem}

The following lemma gives a precise characterization of the finite approximants of fixpoints of completely additive functions.

\begin{lemma}\label{lem:charcaffp}
	Let $G:\wp(\npmoddom)^{n+1}\to\wp(\npmoddom)$ be a completely additive functional. 
	For every $s\in \npmoddom$ and $\vlist{Y}\in \wp(\npmoddom)^n$ we have that $s \in \lfp_X.G(X,\vlist{Y})$ iff there exist $t_1,\dots,t_k \in \npmoddom$ such that $t_k=s$ and the following conditions hold:
	\begin{itemize}
		\itemsep 0 pt
		\item $t_1 \in G(\nada,\vlist{Q})$ where $\vlist{Q}\in\wp(\npmoddom)^n$ is a quasi-atom of $\vlist{Y}$; and
		\item $t_{i+1} \in G(\{t_i\}, \vlist{\nada})$, for all $1 \leq i < k$.
	\end{itemize}
\end{lemma}
\begin{proof}
	\fbox{$\Rightarrow$} As an abbreviation, define $F(X) := G(X,\vlist{Y})$.
	Let $s\in \lfp(F)$ and $k'\in\nat$ be the smallest $k'$ such that $s\in F^{k'}(\nada)$. Such $k'$ exists because of Fact~\ref{fact:constructivity}. We define elements $u_i \in F^i(\nada)$ 
	by downwards induction:
	\begin{itemize}
		\itemsep 0 pt
		\item Case $i = k'$: we set $u_i := s$, which belongs to $F^{k'}(\nada)$.
		\item Case $i < k'$: we want to define $u_i$ in terms of $u_{i+1} \in F^{i+1}(\nada)$. By definition we have that $u_{i+1} \in G(F^i(\nada),\vlist{Y})$. By complete additivity of $G$ there is a quasi-atom $(T,\vlist{Q}')$ of $(F^i(\nada),\vlist{Y})$ such that $u_{i+1} \in G(T,\vlist{Q}')$. We consider the shape of the quasi-atom:
		\begin{enumerate}[(1)]
			\itemsep 0 pt
			\item If $T = \{t\}$ and $\vlist{Q}' = \vlist{\nada}$ we set $u_i := t$ which satisfies $u_{i+1} \in G(\{u_i\}, \vlist{\nada})$.
			\item If $T = \nada$ and $\vlist{Q}'$ is a quasi-atom of $\vlist{Y}$ we set $\vlist{Q} := \vlist{Q}'$ and finish the process.
		\end{enumerate}
		Observe that case (2) will eventually occur. In the worst case this it will occur when $i=1$, because $F^0(\nada)$ is defined as $\nada$.
	\end{itemize}
	This process defines a series of elements $u_{k'},u_{k'-1},\dots,u_j$ where $j\geq 1$. To define the elements $t_j$ we just shift this sequence. That is, we set $k := k'-j+1$ and $t_i := u_{j+i-1}$ for $1 \leq i \leq k$.
	
	\medskip\noindent
	\fbox{$\Leftarrow$}
	This direction will easily follow from this claim:
	\begin{claimfirst}\label{claim:tiinFi}
		$t_{i} \in F^i(\nada)$ for all $1\leq i \leq k$.
	\end{claimfirst}
	\begin{pfclaim}
		We prove it by induction. For the base case, we have by hypothesis that $t_1 \in G(\nada,\vlist{Q})$ where $\vlist{Q}$ is a quasi-atom of $\vlist{Y}$. By monotonicity of $G$ we then have $t_1 \in G(\nada,\vlist{Y})$ which means, by definition of $F$, that $t_1 \in F(\nada)$. For the inductive case let $t_{i+1} \in G(\{t_i\}, \vlist{\nada})$. By inductive hypothesys $t_i \in F^i(\nada)$ therefore, by monotonicity of $G$, we have that $t_{i+1} \in G(F^i(\nada), \vlist{\nada})$. Again by monotonicity, we get that $t_{i+1} \in G(F^i(\nada), \vlist{Y})$. By definition of $F$ we can conclude that $t_{i+1} \in F^{i+1}(\nada)$.
	\end{pfclaim}
	In particular, $t_k = s \in F^k(\nada)$ and therefore we get $s \in \lfp_X.G(X,\vlist{Y})$.
\end{proof}

Note that the above lemma is not restricted to any particular logic, as it expresses a property about an arbitrary completely additive functional $G$. We can now prove our main theorem about completely additive functionals.

\begin{proofof}{Theorem~\ref{thm:propscafmap}(1)}
	Let $G(X,\vlist{Y})$ be a completely additive functional and define $H(\vlist{Y}) := \lfp_X.G(X,\vlist{Y})$. Suppose that $s \in H(\vlist{Y})$. Let $\vlist{Q}$ be the quasi-atom of $\vlist{Y}$ given by Lemma~\ref{lem:charcaffp}, we will prove that $s \in H(\vlist{Q}) = \lfp_X.G(X,\vlist{Q})$. Observe that, by the lemma, $t_1\in G(\nada,\vlist{Q})$. The key observation is that as $t_{i+1} \in G(\{t_i\},\nada)$, by monotonicity we get that $t_{i+1} \in G(\{t_i\},\vlist{Q})$. From this it can be easily seen that, as $s\in G(\{t_{k-1}\},\vlist{Q})$ and $G$ is monotone, we get $s\in \lfp_X.G(X,\vlist{Q})$.
\end{proofof}

\begin{proofof}{Theorem~\ref{thm:propscafmap}(2)}
	Let $F:\wp(\npmoddom)\to\wp(\npmoddom)$ be completely additive and let $s\in\npmoddom$; we prove that
	$s \in \lfp(F) \text{ iff } \text{there exists $Y$ such that } s \in \lfp(F_{\resto Y})$
	where $Y = \{t_1,\dots,t_k\}$ satisfies $t_{i+1} \in F_{\resto Y}^{i+1}(\nada)\setminus F_{\resto Y}^i(\nada)$.
	
	\medskip\noindent
	\fbox{$\Rightarrow$} Let $Y=\{t_1,\dots,t_k\}$ be the set obtained using Lemma~\ref{lem:charcaffp}. In the lemma we already proved that $t_i \in F^i(\nada)$ for all $i$. We now prove the following stronger version of the claim:
	\begin{claim}
		$t_i \in F_{\resto Y}^{i}(\nada)$ for all $i$.
	\end{claim}
	\begin{pfclaim}
		For the base case, we know that $t_1 \in F(\nada)$ by Claim~\ref{claim:tiinFi}~(Lemma~\ref{lem:charcaffp}); moreover, by definition $t_1 \in Y$. Hence $t_1 \in F(\nada)\cap Y$ which, by definition of $F_{\resto Y}$, is equivalent to $t_1 \in F_{\resto Y}(\nada)$. For the inductive case let $t_{i+1} \in F^{i+1}(\nada)$. By definition of $F^{i+1}$ we have that $t_{i+1} \in F(F^{i}(\nada))$. Now we use the iductive hypothesis and get that $t_{i+1} \in F(F^{i}_{\resto Y}(\nada))$. As we did in the base case, because $t_{i+1}\in Y$, we know that $t_{i+1} \in F(F^{i}_{\resto Y}(\nada)) \cap Y$ which by definition of $F_{\resto Y}$ and regrouping we can conclude that $t_{i+1} \in F^{i+1}_{\resto Y}(\nada)$.
	\end{pfclaim}
	\noindent In particular $s \in F_{\resto Y}^k(\nada)$ and therefore $s\in \lfp(F_{\resto Y})$. 
	
	\medskip\noindent
	\fbox{$\Leftarrow$} This direction goes through using a monotonicity argument. That is, using that for all $X$ we have $F_{\resto Y}(X) \subseteq F(X)$, it is not difficult to prove that $F^\alpha_{\resto Y}(X) \subseteq F^\alpha(X)$ for all $\alpha$, which entails that $\lfp(F_{\resto Y}) \subseteq \lfp(F)$.
\end{proofof}

Of course, the functionals with which we will work are induced by formulas of $\muML$, $\mufoe$ and $\wcl$ which are completely additive. This leads us to analyze the completely additive fragment of $\mufoe$.

%% file: fotc-defcaf.tex

Before stating the main definitions we need to introduce some useful notation. Given a non-pointed model $\npmodel = \tup{\npmoddom, R_{\aact\in\acts}, \tscolors}$, elements $\vlist{q} \in \props^n$ and $\vlist{X} = (X_1,\dots,X_n) \in \wp(\npmoddom)^n$, we introduce the following notation:
\begin{align*}
	\tsval(\vlist{q}) &:= \tsval(q_1),\dots,\tsval(q_n)\\
	\npmodel[\vlist{q}\mapsto \vlist{X}] &:= \npmodel[q_i \mapsto X_i \mid 1\leq i \leq n]\\
	\npmodel[\vlist{q}\resto \vlist{X}] &:= \npmodel[q_i \mapsto \tsval(q_i) \cap X_i \mid 1\leq i \leq n].
\end{align*}
We are now ready to state the main definition of this section.

\begin{definition}
We say that $\varphi\in\mufoe$ is
\emph{completely additive in $\{q_1,\dots,q_n\}=\qprops\subseteq\props$} if
for every model $\npmodel$ and assignment $\ass$ it satisfies
\[
\npmodel,\ass \models \varphi \quad\text{ iff }\quad \npmodel[\qprops\resto\vlist{Y}],\ass \models \varphi \text{ for some quasi-atom $\vlist{Y}$ of $\tsval(\qprops)$.}
\]
\end{definition}

\begin{proposition}\label{prop:caformcamap}
If $\varphi\in\mufoe$ is completely additive in $\qprops$ then
for every model $\npmodel$, assignment $\ass$ and variable $x\in\fovar$, the map $G_x:\wp(\npmoddom)^n\to\wp(\npmoddom)$ given by
\[
G_x(\vlist{Z}) := \{ t\in\npmoddom \mid \npmodel[\qprops\mapsto\vlist{Z}],\ass[x\mapsto t] \models \varphi\}
\]
is completely additive.
\end{proposition}
\begin{proof}
	Fix a model $\npmodel$, assignment $g$ and free variable $x\in FV(\varphi)$. We want to prove that $G_x(\vlist{Z})$ is completely additive. An element $t$ belongs to $G_x(\vlist{Z})$ iff $\npmodel[\qprops\mapsto\vlist{Z}],\ass[x\mapsto t] \models \varphi$. By complete additivity of $\varphi$, this occurs iff $\npmodel[\qprops\mapsto\vlist{Y}],\ass[x\mapsto t] \models \varphi$ for some quasi-atom $\vlist{Y}$ of $\vlist{Z}$. By definition of $G_x$, this is equivalent to saying that $t \in G_x(\vlist{Y})$. Therefore, the map $G_x$ is completely additive.
	%
\end{proof}

%
Next, we provide a definition of a fragment of $\mufoe$, and shortly after that we prove that every formula in this fragment is completely additive. 

\begin{definition}
Let $\qprops\subseteq \props$ be a set of monadic predicates. The fragment $\add{\mufoe}{\qprops}(\props,\acts)$ is defined by the following rules:
\[
\varphi ::= \psi \mid q(x) \mid \exists x.\varphi(x) \mid \varphi \lor \varphi \mid \varphi \land \psi \mid [\lfp_{p{:}x}.\xi(p,x)](z)
\]
where $q \in \qprops$, $\psi \in \mufoe(\props\setminus \qprops,\acts)$, $p \in \props\setminus \qprops$ and $\xi(p,x) \in \add{\mufoe}{\qprops \cup\{p\}}(\props,\acts)$.
\end{definition}

\noindent Observe that the atomic formulas given by equality and relations are taken into account by this definition in the $\psi$ clause.

\begin{proposition}
	Every $\varphi \in \add{\mufoe}{\qprops}$ is completely additive in $\qprops$.
\end{proposition}
\begin{proof}
The proof goes by induction, most cases are solved similar to Lemma~\ref{lem:caddofoiscadd}. We focus on the inductive step of the fixpoint operator.
	%
	%
	%
	%
	Let $\varphi$ be $[\lfp_{p{:}x}.\psi(p,x)](z)$, $\npmodel$ be a model and $\ass$ an assigment,
	we have to prove that
	\[
		\npmodel,\ass \models \varphi \quad\text{ iff }\quad \npmodel[\qprops\resto\vlist{Y}],\ass \models \varphi \text{ for some quasi-atom $\vlist{Y}$ of $\tsval(\qprops)$.}
	\]
	By semantics of the fixpoint operator $\npmodel,\ass \models \varphi$ iff $\ass(z) \in \lfp(F_{p:x})$ where
	\[
		F_{p:x}(P) := \{ t\in\npmoddom \mid \npmodel[p\mapsto P],\ass[x\mapsto t] \models \psi\}.
	\]
	It will be useful to take a slightly more general definition: consider the map
	\[
	G_{p:x}^\psi(P,\vlist{Z}) := \{ t\in\npmoddom \mid \npmodel[p\mapsto P;\qprops\mapsto\vlist{Z}],\ass[x\mapsto t] \models \psi\}
	\]
	and observe that $F_{p:x}(P) = G_{p:x}^\psi(P,\tsval(\qprops))$ and therefore their least fixpoints will be the same. By inductive hypothesis and Proposition~\ref{prop:caformcamap}, we know that $G_{p:x}^\psi(P,\vlist{Z})$ is completely additive. Using Theorem~\ref{thm:propscafmap}(1) we get that $\lfp_P.G_{p:x}^\psi(P,\tsval(\qprops))$ is completely additive as well. That is, there is a quasi-atom $\vlist{Y}$ of $\tsval(\qprops)$ such that
	\[
	t \in \lfp_P.G_{p:x}^\psi(P,\tsval(\qprops))
	\quad\text{iff}\quad
	t \in \lfp_P.G_{p:x}^\psi(P,\vlist{Y})
	\]
	from which we can conclude that $\npmodel,\ass \models \varphi$ iff $\npmodel[\qprops\resto\vlist{Y}],\ass \models \varphi$.
\end{proof}

This proves that the above fragment is ``sound'' with respect to the property of complete additivity. We conjecture that the fragment is also ``complete'' with respect to this property, i.e., that every formula of $\mufoe$ which is completely additive in $\qprops$ is equivalent to a formula in $\add{\mufoe}{\qprops}$. We do not pursue this matter because it goes beyond the objectives of the current article.

\begin{conjecture}
	Every formula $\varphi\in\mufoe$ which is completely additive in $\qprops$ is equivalent to some formula $\varphi'\in \add{\mufoe}{\qprops}$.
\end{conjecture}

\noindent Finally, we define $\mucafoe$:

\begin{definition}
The fragment $\mucafoe$ of $\mufoe$ is given by the following restriction of the fixpoint operator to the completely additive fragment:
\[
\varphi ::= q(x) \mid R_{\aact}(x,y) \mid x \foeq y \mid \exists x.\varphi \mid \lnot\varphi \mid \varphi \lor \varphi \mid [\lfp_{p{:}x}.\xi(p,x)](z)
\]
where $p,q\in\props$, $\aact\in\acts$, $x,y\in\fovar$; and $\xi(p,x) \in \add{\mufoe}{\{p\}} \cap \mucafoe$.
\end{definition}

%% file: fotc-mucafoe.tex

In this subsection we prove Theorem~\ref{thm:fotcmucafoe}. That is, we give effective translations that witness the equivalence $\binfotc \equiv \mucafoe$.

\paragraph{From $\binfotc$ to $\mucafoe$.}

In Remark~\ref{rem:tcinlfp} we observed that the reflexive-transitive closure of a formula can be expressed as a fixed point. That is,
\[[\tc_{x,y}.\varphi(x,y)](u,v) \equiv [\lfp_{p{:}y}.y\foeq u \lor (\exists x. p(x) \land \varphi(x,y))](v).\]
It is easy to see (syntactically) that the formula inside the fixpoint is completely additive in $p$ (which is a fresh variable); therefore it belongs to $\mucafoe$. Moreover, the equivalence holds for all models, in particular, for trees.


\paragraph{From $\mucafoe$ to $\binfotc$.}

An alternative way to read Theorem~\ref{thm:propscafmap}(2) is that an element $s$ belongs to the least fixed point of a map $F$ iff there is a sequence of elements from $F(\nada), F^2(\nada), \dots$ which eventually reaches~$s$. We will now introduce a new notation which is closer to this reading, and rephrase Theorem~\ref{thm:propscafmap}(2) in those terms. The connection between fixpoints of completely additive maps and transitive closure will become clear.

\begin{definition} 
	The relation ${\stl_F}$ is defined for every $t,t' \in \npmoddom$ as
	$t \stl_F t'$ iff $t'\in F(\{t\})$.
\end{definition}

\noindent With this notation, Theorem~\ref{thm:propscafmap}(2) can be reformulated as follows:

\begin{corollary}\label{cor:fpcafrel}
	Let $F:\wp(\npmoddom)\to\wp(\npmoddom)$ be completely additive.
	For every $s\in\npmoddom$ we have that $s \in \lfp(F)$ iff there exist $t_1 \in F(\nada)$ such that $t_1 \stl^*_F s$.
\end{corollary}

Let $\varphi(p,z)$ belong to $\add{\mufoe}{p}$. It only remains to observe that the required relations can be defined in $\binfotc$, as follows:
\begin{itemize}
	\itemsep 0 pt
	\item $v \in F(\nada)$ is equivalent to $\varphi(\bot,v)$,
	\item $u \stl_F v$ is equivalent to $\varphi(p,v)[p(y) \mapsto u\foeq y]$.
\end{itemize}
To finish, we define
$[\lfp_{p{:}z}.\varphi(p,z)](v) := \exists t_1. \varphi(\bot,t_1) \land [\tc_{x,y}.{\stl_{F_z^\varphi}}(x,y)](t_1,v)$.

%% file: os-intro.tex


In this section we define the one-steps logics that we use in the rest of the article, namely: one-step first-order logic with and without equality ($\ofoe$, $\ofo$). The main theorems of this section prove normal forms for these logics and give syntactical characterizations of the monote and completely additive fragments that we use in later sections.

\begin{definition}
The set $\ofoe(A,\sorts)$ of (multi-sorted) one-step first-order sentences (with equality) is given by the sentences formed by
\[
\varphi ::= a(x)
\mid x \foeq y
\mid \neg \varphi
\mid \varphi \lor \varphi
\mid \exists x{:}\asort.\varphi
\]
where $x,y\in \fovar$, $a \in A$ and $\asort\in\sorts$. The one-step logic $\ofo(A,\sorts)$ of multi-sorted first-order sentences without equality is defined similarly.
\end{definition}

Without loss of generality, from now on we always assume that every bound variable occurring in a sentence is bound by an unique quantifier.
%
%
Recall that given a one-step logic $\llang_1$ we write $\llang_1^+(A)$ to denote the fragment
where every predicate $a\in A$ occurs only positively. 

The multi-sorted semantics that we will use in this article is slightly non-standard: for example, the individual variables and names (predicates) do not have a fixed sort. We define the semantics formally to avoid confusions.

\begin{definition}
	Let $\varphi \in \ofoe(A,\sorts)$ be a formula, $\osmodel = (D_{\asort_1},\dots,D_{\asort_n},\val)$ be a one-step model on sorts $\sorts$ and $\ass:\fovar\to \wp(D)$ be an assignment. The semantics of $\ofoe(A,\sorts)$ is given as follows:
	\begin{align*}
	    \osmodel,\ass \models a(x) & \quad\text{iff}\quad \ass(x) \in \val(a),\\
	    \osmodel,\ass \models x \foeq y & \quad\text{iff}\quad \ass(x) = \ass(y),\\
	    \osmodel,\ass \models \exists x{:}\asort.\varphi & \quad\text{iff}\quad \osmodel,\ass[x\mapsto d] \models \varphi \text{ for some $d\in D_\asort$},
	\end{align*}
	where the Boolean connectives are defined as expected.
\end{definition}

In the following subsections we provide a detailed model theoretic analysis of the one-step logics that we use in this article, specifically, we give:
\begin{itemize}
	\itemsep 0 pt
	\item Normal forms for arbitrary formulas of multi-sorted $\ofo$ and $\ofoe$.
	\item Strong forms of syntactic characterizations for the monotone and completely additive fragments of $\ofo$ and $\ofoe$. Namely, for $\llang_1 \in \{\ofo,\ofoe\}$ we give:
		\begin{enumerate}[(a)]
			\item A fragment $\monot{\llang_1}{A'}$ and a translation $(-)^\tmono:\llang_1(A)\to\monot{\llang_1}{A'}(A)$ such that for every $\varphi \in\llang_1$ we have $\varphi\equiv\varphi^\tmono$ iff $\varphi$ is monotone in $A' \subseteq A$,
			%
		%
		%
			\item A fragment $\add{\llang_1}{A'}$ and a translation $(-)^\tadd:\llang_1(A)\to\add{\llang_1}{A'}(A)$ such that for every $\varphi \in\llang_1$ we have $\varphi\equiv\varphi^\tadd$ iff $\varphi$ is completely additive in $A' \subseteq A$.
		\end{enumerate}
		Moreover, we show that the latter translation also restricts to the fragment $\llang_1^+$, i.e.,
		\begin{enumerate}[(a)]
			\item[(c)] The restriction $(-)^\tadd_+:\llang_1^+(A)\to\add{\llang_1^+}{A'}(A)$ of $(-)^\tadd$ is such that for every $\varphi \in\llang_1^+$ we have $\varphi\equiv\varphi^\tadd_+$ iff $\varphi$ is completely additive in $A' \subseteq A$.
		\end{enumerate}
	\item Syntactic characterizations of the completely multiplicative fragments of $\ofo$ and $\ofoe$.
	\item Normal forms for the monotone and completely additive fragments.
\end{itemize}



%% file: os-normalforms.tex

Given a set of names $A$ and $S \subseteq A$, we introduce the notation
\[
  \tau_{S}(x) := \bigwedge_{a\in S}a(x) \land \bigwedge_{a\in A\setminus S}\lnot a(x).
\]
The formula $\tau_{S}(x)$ is called an \emph{$A$-type}. We usually blur the distinction between $\tau_{S}(x)$ and $S$ and call $S$ an $A$-type as well.
A \emph{positive} $A$-type is defined as $\tau_{S}^+(x) := \bigwedge_{a\in S}a(x)$.
We use the convention that if $S=\nada$, then $\tau_S^+(x)$ is $\top$ and we call it the \emph{empty} positive $A$-type.
Given a one-step model $\osmodel$ we use $|S|^\asort_\osmodel$ to denote the number of elements of sort $s\in\sorts$ that realize the $A$-type $\tau_S$ in $\osmodel$. Formally, it is defined as $|S|^\asort_\osmodel := |\{d\in D_\asort : \osmodel \models \tau_S(d) \}|$.

A \emph{partial isomorphism} between two multi-sorted one-step models $\osmodel = (D_{\asort_1},\dots,D_{\asort_n},\val)$ and $\osmodel' = (D'_{\asort_1},\dots,D'_{\asort_n},\val')$ is a \emph{partial} function $f: D \pto D'$ which is injective and for all $d\in \Dom(f)$ it satisfies the following conditions:
\begin{description}
	\itemsep 0pt
	\item[(sorts)] $d$ and $f(d)$ have the same sorts,
	\item[(atom)] $d \in \val(a) \Leftrightarrow f(d) \in \val'(a)$, for all $a\in A$.
\end{description}
Given two sequences $\vlist{d} \in D^k$ and $\vlist{d'} \in {D'}^k$ 
we use
$f:\vlist{d} \mapsto \vlist{d'}$ to denote the partial function $f:D\pto D'$ defined as $f(d_i) := d'_i$. If there exist $d_i,d_j$ such that $d_i = d_j$ but $d'_i \neq d'_j$ then the result is undefined.

\begin{definition}
The quantifier rank $qr(\varphi)$ of $\varphi \in \ofoe$ is defined as follows
\begin{itemize}
	\itemsep 0 pt
	\item If $\varphi$ is atomic $qr(\varphi) = 0$,
	\item If $\varphi = \lnot\psi$ then $qr(\varphi) = qr(\psi)$,
	\item If $\varphi = \psi_1 \land \psi_2$ or $\varphi = \psi_1 \lor \psi_2$ then $qr(\varphi) = \max\{qr(\psi_1),qr(\psi_2)\}$,
	\item If $\varphi = Qx{:}\asort.\psi$ for $Q \in \{\exists,\forall\}$ then $qr(\varphi) = 1+qr(\psi)$.
\end{itemize}
Given a one-step logic $\llang$ we write $\osmodel \equiv_k^{\llang} \osmodel'$ to indicate that the one-step models $\osmodel$ and $\osmodel'$ satisfy exactly the same formulas $\varphi \in \llang$ with $qr(\varphi) \leq k$. The logic $\llang$ will be omitted when it is clear from context.
\end{definition}

\subsubsection{Normal form for $\ofo$}

We start by stating a normal form for one-step first-order logic without equality. A formula in \emph{basic form} gives a complete description of the types that are satisfied in a one-step model.

\begin{definition}
A formula $\varphi \in \ofo(A,\sorts)$ is in \emph{basic form} if $\varphi = \bigvee \bigwedge_\asort \dgbnfofo{\Sigma}{\Pi}_\asort$ where in each conjunct
\[
\dgbnfofo{\Sigma}{\Pi}_\asort :=
\bigwedge_{S\in\Sigma} \exists x{:}\asort. \tau_S(x) \land \forall x{:}\asort. \bigvee_{S\in\Pi} \tau_S(x)
\]
for some set of types $\Sigma,\Pi \subseteq \wp(A)$.
\end{definition}

\begin{remark}\label{rem:ofostrict}
$\ofo$ cannot distinguish between arbitrary and strict one-step models. More formally, every arbitrary one-step model $(D_1,\dots,D_n,\val)$ is equivalent (for $\ofo$) to the model $(D_1\times\{1\},\dots,D_n\times\{n\},\val_\pi)$ where $\val_{\pi}(a) := \{(d,k) \mid d \in \val(a), k \in \{0,\dots,n\}\}$. Therefore, when proving results for $\ofo$, it is not difficult to see that we can restrict to the class of strict one-step models.
\end{remark}

It is not difficult to prove, using Ehrenfeucht-Fra\"iss\'e games, that every formula of monadic first-order logic without equality (i.e., unsorted $\ofo$) is equivalent to a formula in basic form over strict models. By Remark~\ref{rem:ofostrict} the normal form also holds over arbitrary models. Proof sketches for the unsorted case can be found in~\cite[Lemma~16.23]{ALG02} and~\cite[Proposition~4.14]{Venxx}. We omit a full proof for the sorted case because it is very similar to the case of $\ofoe$.

\begin{proposition}
Every formula of $\ofo(A,\sorts)$ is equivalent to a formula in basic form.
\end{proposition}

\subsubsection{Normal form for $\ofoe$}

In this subsection we will have to pay particular attention to the kind of one-step models that we are working with. The case of $\ofoe$ is much more complicated, as this logic \emph{can} distinguish between strict and arbitrary one-step models. We first give a normal form for \emph{strict} models and afterwards generalize it to arbitrary models.

\paragraph{The strict case.}
We prove that every formula of \emph{multi-sorted} monadic first-order logic with equality (i.e., $\ofoe$) is equivalent to a formula in strict basic form over strict models.

\begin{definition}
A formula $\varphi \in \ofoe(A,\sorts)$ is in \emph{strict basic form} if $\varphi = \bigvee\bigwedge_\asort \dbnfofoe{\vlist{T}}{\Pi}_\asort$ where in each conjunct we have
\[
\dbnfofoe{\vlist{T}}{\Pi}_\asort := \exists \vlist{x}{:}\asort.\big(\arediff{\vlist{x}} \land \bigwedge_i \tau_{T_i}(x_i) \land \forall z{:}\asort.(\arediff{\vlist{x},z} \lthen \bigvee_{S\in \Pi} \tau_S(z))\big)
\]
such that $\vlist{T} \in \wp(A)^k$ for some $k$, $\asort \in \sorts$ and $\Pi \subseteq \vlist{T}$. The predicate $\arediff{\vlist{y}}$, which states that the elements $\vlist{y}$ are different, is given as $\arediff{y_1,\dots,y_n} := \bigwedge_{1\leq m < m^{\prime} <n} (y_m \not\approx y_{m^{\prime}})$.
\end{definition}

\noindent We start by defining the following relation between strict one-step models.

\begin{definition}
	Let $\osmodel$ and $\osmodel'$ be strict one-step models. For every $k \in \nat$ we define 
	\[
		\osmodel \sim^=_k \osmodel' \Longleftrightarrow
		\forall S\subseteq A, \asort\in\sorts.
		\big(|S|^\asort_\osmodel = |S|^\asort_{\osmodel'} < k \text{ or } |S|^\asort_\osmodel,|S|^\asort_{\osmodel'} \geq k \big)
	\]
\end{definition}

Intuitively, two models are related by $\sim^=_k$ when their type information coincides `modulo~$k$'. Later we will prove that this is the same as saying that they cannot be distinguished by a formula of $\ofoe$ with quantifier rank lower or equal to $k$. For the moment, we prove the following properties of $\sim^=_k$.

\begin{proposition}\label{prop:eqrelofoe} The following hold
	\begin{enumerate}[(i)]
		\itemsep 0 pt
		\item $\sim^=_k$ is an equivalence relation,
		\item $\sim^=_k$ has finite index,
		\item Every $E \in \sumods/{\sim^=_k}$ is characterized by a formula $\varphi^=_E \in \ofoe(A,\sorts)$ with $qr(\varphi^=_E) = k$.
	\end{enumerate}
\end{proposition}
\begin{proof}
	We only prove the last point. Let $E \in \sumods/{\sim^=_k}$ and let $\osmodel \in E$ be a representative. For every $\asort\in\sorts$ call $S_1,\dots,S_n \subseteq A$ to the types such that $|S_i|^\asort_\osmodel = n_i < k$ and $S'_1,\dots,S'_m \subseteq A$ to those satisfying $|S'_i|^\asort_\osmodel \geq k$. Now define
	\begin{align*}
	\varphi^=_{E,\asort} :=
		& \bigwedge_{i\leq n} \big(\exists x_1,\dots,x_{n_i}{:}\asort.\arediff{x_1,\dots,x_{n_i}}\ \land \\
		& \phantom{mm} \bigwedge_{j\leq n_i} \tau_{S_i}(x_j) \land \forall z{:}\asort. \arediff{x_1,\dots,x_{n_i},z} \lthen \lnot\tau_{S_i}(z)\big)\ \land \\
        & \bigwedge_{i\leq m} \big(\exists x_1,\dots,x_k{:}\asort.\arediff{x_1,\dots,x_k} \land \bigwedge_{j\leq k} \tau_{S'_i}(x_j) \big)
	\end{align*}
	Finally set $\varphi^=_{E} := \bigwedge_\asort \varphi^=_{E,\asort}$. It is easy to see that $qr(\varphi^=_{E}) = k$ and that $\osmodel' \models \varphi^=_{E}$ iff $\osmodel' \in E$. Observe that $\varphi^=_E$ gives a specification of $E$ ``sort by sort and type by type''.
\end{proof}

In the following definition we recall the notion of Ehrenfeucht-Fra\"iss\'e game for $\ofoe$, slightly adapted for the multi-sorted setting, which will be used to establish the connection between ${\sim^=_k}$ and $\equiv_k^\foe$.

\begin{definition}
	Let $\osmodel_0 = (D_0,\val_0)$ and $\osmodel_1 = (D_1,\val_1)$ be strict multi-sorted one-step models. We define the game $\efgame^=_k(\osmodel_0,\osmodel_1)$ between \abelard and \eloise. If $\osmodel_i$ is one of the models we use $\osmodel_{-i}$ to denote the other model, we do the same with elements and elements. Note that in this definition the index $i$ will never refer to a sort, but to one of the models. A position in this game is a pair of sequences $\vlist{s_0} \in D_0^n$ and $\vlist{s_1} \in D_1^n$ with $n \leq k$. The game consists of $k$ rounds where in round $n+1$ the following steps are made
	\begin{enumerate}[1.]
		\itemsep 0 pt
		\parsep 0 pt
		\item \abelard chooses an element $d_i$ in one of the $\osmodel_i$,
		\item \eloise responds with an element $d_{-i}$ in the model $\osmodel_{-i}$.
		\item Let $\vlist{s_i} \in D_i^n$ be the sequences of elements chosen up to round $n$, they are extended to ${\vlist{s_i}' := \vlist{s_i}\cdot d_i}$. Player \eloise survives the round iff she does not get stuck and the function $f_{n+1}: \vlist{s_0}' \mapsto \vlist{s_1}'$ is a partial isomorphism of one-step models.
	\end{enumerate}
	Player \eloise wins iff she can survive all $k$ rounds.
	%
	%
	Given $n\leq k$ and $\vlist{s_i} \in D_i^n$ such that $f_n:\vlist{s_0}\mapsto\vlist{s_1}$ is a partial isomorphism, we use $\efgame_{k}^=(\osmodel_0,\osmodel_1)@(\vlist{s_0},\vlist{s_1})$ to denote the (initialized) game where $n$ moves have been played and $k-n$ moves are left to be played.
\end{definition}

\begin{lemma}\label{lem:connofoe}
	The following are equivalent
	\begin{enumerate}
		\itemsep 0 pt
		\item\label{lem:connofoe:i} $\osmodel_0 \equiv_k^\foe \osmodel_1$,
		\item\label{lem:connofoe:ii} $\osmodel_0 \sim_k^= \osmodel_1$,
		\item\label{lem:connofoe:iii} \eloise has a winning strategy in $\efgame_k^=(\osmodel_0,\osmodel_1)$.
	\end{enumerate}
\end{lemma}
\begin{proof}
	Step~(\ref{lem:connofoe:i}) to~(\ref{lem:connofoe:ii}) is direct by Proposition~\ref{prop:eqrelofoe}. For~(\ref{lem:connofoe:ii}) to~(\ref{lem:connofoe:iii}) we give a winning strategy for \eloise in $\efgame_k^=(\osmodel_0,\osmodel_1)$. We do it by showing the following claim
	\begin{claimfirst}
	Let $\osmodel_0 \sim_k^= \osmodel_1$ and $\vlist{s_i} \in D_i^n$ be such that $n<k$ and $f_n:\vlist{s_0}\mapsto\vlist{s_1}$ is a partial isomorphism; then \eloise can survive one more round in $\efgame_{k}^=(\osmodel_0,\osmodel_1)@(\vlist{s_0},\vlist{s_1})$.
	\end{claimfirst}
	\begin{pfclaim}
		Let \abelard pick $d_i\in D_i$ such that $d_i$ has type $T \subseteq A$ and sort $\asort\in\sorts$. If $d_i$ had already been played then \eloise picks the same element as before and $f_{n+1} = f_n$. If $d_i$ is new and $|T|^\asort_{\osmodel_i} \geq k$ then, as at most $n<k$ elements have been played, there is always some new $d_{-i} \in D_{-i}$ that \eloise can choose that matches $d_i$. If $|T|^\asort_{\osmodel_i} = m < k$ then we know that $|T|^\asort_{\osmodel_{-i}} = m$. Therefore, as $d_i$ is new and $f_n$ is injective, there must be a $d_{-i} \in D_{-i}$ of sort $\asort$ that \eloise can choose. 
	\end{pfclaim}
	
	Step~(\ref{lem:connofoe:iii}) to~(\ref{lem:connofoe:i}) is a standard result~\cite[Corollary 2.2.9]{fmt} in the unsorted setting, we prove it for the multi-sorted setting and for completeness sake.
	\begin{claim}
		Let $\vlist{s_i} \in D_i^n$ and $\varphi(z_1,\dots,z_n) \in \ofoe(A)$ be such that $qr(\varphi) \leq k-n$. If \eloise has a winning strategy in $\efgame_k^=(\osmodel_0,\osmodel_1)@(\vlist{s_0},\vlist{s_1})$ then $\osmodel_0 \models \varphi(\vlist{s_0})$ iff $\osmodel_1 \models \varphi(\vlist{s_1})$.
	\end{claim}
	\begin{pfclaim}
		If $\varphi$ is atomic the claim holds because of $f_n:\vlist{s_0}\mapsto \vlist{s_1}$ being a partial isomorphism (more specifically, the \emph{atom} condition). Boolean cases are straightforward.
		Let $\varphi(z_1,\dots,z_n) = \exists x{:}\asort. \psi(z_1,\dots,z_n,x)$ and suppose $\osmodel_0 \models \varphi(\vlist{s_0})$. Hence, there exists $d_0 \in D_0$ of sort $\asort$ such that $\osmodel_0 \models \psi(\vlist{s_0},d_0)$.
		By hypothesis we know that \eloise has a winning strategy for $\efgame_k^=(\osmodel_0,\osmodel_1)@(\vlist{s_0},\vlist{s_1})$. Therefore, if \abelard picks $d_0\in D_0$ she can respond with some $d_1\in D_1$ and has a winning strategy for $\efgame_{k}^=(\osmodel_0,\osmodel_1)@(\vlist{s_0}{\cdot}d_0,\vlist{s_1}{\cdot}d_1)$.
		First observe that, as \eloise survives the round, then $\vlist{s_0}{\cdot}d_0\mapsto\vlist{s_1}{\cdot}d_1$ is a partial isomorphism and hence (by the \emph{sorts} condition) the lements $d_0$ and $d_1$ will have the same sort.
		By induction hypothesis, because $qr(\psi) \leq k- (n+1)$, we have that $\osmodel_0 \models \psi(\vlist{s_0},d_0)$ iff $\osmodel_1 \models \psi(\vlist{s_1},d_1)$ and hence $\osmodel_1 \models \exists x{:}\asort.\psi(\vlist{s_1},x)$. The other direction is symmetric. 
		\end{pfclaim}
		Combining these claims finishes the proof of the lemma.
\end{proof}

\begin{theorem}\label{thm:sbnfofoe}
Over strict models, every formula $\varphi \in \ofoe(A,\sorts)$ is equivalent to a formula $\psi \in \ofoe(A,\sorts)$ in strict basic form.
\end{theorem}
\begin{proof}
	Let $qr(\varphi) = k$ and let $\ext{\varphi}$ be the models satisfying $\varphi$. As $\sumods/{\equiv_k^\foe}$ is the same as $\sumods/{\sim_k^=}$ by Lemma~\ref{lem:connofoe}, it is easy to see that $\varphi \equiv \bigvee \{ \varphi^=_E \mid E \in \ext{\varphi}/{\sim_k^=} \}$. Remember that $\varphi^=_E$ is defined as $\bigwedge_\asort \varphi^=_{E,s}$. Therefore, it is enough to see that each $\varphi^=_{E,s}$ is equivalent to some $\dbnfofoe{\vlist{T}}{\Pi}_\asort$ where $T_i \subseteq A$ and $\Pi \subseteq \vlist{T}$. From this, we can conclude that $\varphi$ is equivalent to $\psi := \bigvee \{ \varphi^=_E \mid E \in \ext{\varphi}/{\sim_k^=} \}$.

	The crucial observation is that we will use $\vlist{T}$ and $\Pi$ to give a specification of the types ``element by element''. Let $\osmodel \in E$ be a representative. Call $S_1,\dots,S_n \subseteq A$ to the types such that $|S_i|^\asort_\osmodel = n_i < k$ and $S'_1,\dots,S'_m \subseteq A$ to those satisfying $|S'_i|^\asort_\osmodel \geq k$. The size of the sequence $\vlist{T}$ is defined to be $(\sum_{i=1}^n n_i) + k\times m$ where $\vlist{T}$ is contains exactly $n_i$ occurrences of type $S_i$ and $k$ occurrences of each $S'_j$. On the other hand $\Pi = \{S'_1,\dots,S'_m\}$. It is straightforward to check that $\varphi^=_{E,\asort}$ is equivalent to $\dbnfofoe{\vlist{T}}{\Pi}_\asort$, however, the quantifier rank of the latter is only bounded by $k\times 2^{|A|} + 1$.
\end{proof}

\paragraph{The arbitrary case.}
We now prove that we can also give a normal form for arbitrary models.
As an intuition on why the strict normal form ``lifts'' to arbitrary models observe that any one-step model on sorts $\sorts$ can be seen as a \emph{strict} one-step model on sorts $\wp(\sorts)$.

\begin{definition}
	For an arbitrary one-step model $\osmodel$ on sorts $\sorts$ we define $\osmodel^\uparrow$ to be the strict one-step model on sorts $\wp(\sorts)$ obtained by redefining the sorts of $\osmodel$ as follows: an element $d$ of $\osmodel^\uparrow$ belongs to the sort $\aSort \subseteq \sorts$ iff it belongs to all the sorts $\asort\in\aSort$ in $\osmodel$ and it does \emph{not} belong to any sort $\asort' \in \sorts\setminus\aSort$ in $\osmodel$.

	For every $\varphi \in \ofoe(A,\sorts)$ we define the translation $\varphi^\uparrow \in \ofoe(A,\wp(\sorts))$ inductively: it behaves homomorphically in every operator which is not the existential quantifier. For the existential quantifier, it is defined as follows:
	\[
		(\exists x{:}\asort.\varphi(x))^\uparrow := \bigvee\{ \exists x{:}\aSort.\varphi^\uparrow(x) \mid \{\asort\} \subseteq \aSort \subseteq \sorts\} .
	\]
\end{definition}

\noindent The following proposition states the expected relationship.

\begin{proposition}\label{prop:liftsorts}
	For every $\varphi \in \ofoe(A,\sorts)$ and arbitrary one-step model $\osmodel$ on sorts $\sorts$ we have that $\osmodel \models \varphi$ iff $\osmodel^\uparrow \models \varphi^\uparrow$.
\end{proposition}

The next step is to use Theorem~\ref{thm:sbnfofoe} (over strict models) to get a strict normal form $\psi$ of $\varphi^\uparrow$. After that we want to transfer the normal form to arbitrary models, therefore we need something like a converse of Proposition~\ref{prop:liftsorts}.
With this in mind, we introduce the following abbreviation $\exists x{:}\aSort!$:
\begin{align*}
	\exists x{:}\aSort!.\varphi(x) := \
		& \exists x,x_1{:}\asort_1,\dots,x_n{:}\asort_n. \aresame{x,x_1,\dots,x_n} 
		\land \big(\bigwedge_{\asort \in \sorts \setminus \aSort}\forall z{:}\asort. x\neq z\big) \land \varphi(x)
\end{align*}
where $\aresame{y_1,\dots,y_n} := \bigwedge_{1\leq m < n} (y_m \approx y_{m+1})$. Intuitively speaking, the quantifier $\exists x{:}\aSort!$ says that there is an element $x$ which belongs \emph{exactly} to the sorts $\aSort \subseteq \wp(\sorts)$.

\begin{definition}
	For every $\psi \in \ofoe(A,\wp(\sorts))$ we define the translation $\psi^\downarrow \in \ofoe(A,\sorts)$ inductively: it behaves homomorphically in every operator which is not the existential quantifier. For the existential quantifier, it is defined as follows:
	\[
		(\exists x{:}\aSort.\psi(x))^\downarrow := \exists x{:}\aSort!.\psi^\downarrow(x)
	\]
	for $\aSort \in \wp(\sorts)$.
\end{definition}

\noindent The following proposition states the expected relationship.

\begin{proposition}\label{prop:downsorts}
	For every $\psi \in \ofoe(A,\wp(\sorts))$ and arbitrary one-step model $\osmodel$ on sorts $\sorts$ we have that $\osmodel \models \psi^\downarrow$ iff $\osmodel^\uparrow \models \psi$.
\end{proposition}


We are now ready to state the normal form of $\ofoe$ for arbitrary models and generalize Theorem~\ref{thm:sbnfofoe}.

\begin{definition}\label{def:ofobform}
A formula $\varphi \in \ofoe(A,\sorts)$ is in \emph{basic form} if $\varphi = \bigvee\bigwedge_\aSort \dbnfofoe{\vlist{T}}{\Pi}_\aSort$ where in each conjunct we have
\[
\dbnfofoe{\vlist{T}}{\Pi}_\aSort := \exists \vlist{x}{:}\aSort!.\big(\arediff{\vlist{x}} \land \bigwedge_i \tau_{T_i}(x_i) \land \forall z{:}\aSort!.(\arediff{\vlist{x},z} \lthen \bigvee_{S\in \Pi} \tau_S(z))\big)
\]
such that $\vlist{T} \in \wp(A)^k$ for some $k$, $\aSort \subseteq \sorts$ is non-empty and $\Pi \subseteq \vlist{T}$.
\end{definition}

\begin{theorem}\label{thm:bnfofoe}
Every $\varphi \in \ofoe(A,\sorts)$ is equivalent to a formula in basic form.
\end{theorem}
\begin{proof}
	We use the notation that we have developed in this subsection and proceed as follows:
	\begin{align*}
	\osmodel \models \varphi
		& \quad\text{iff}\quad \osmodel^\uparrow \models \varphi^\uparrow
		& \tag{Proposition~\ref{prop:liftsorts}}
	\\
		& \quad\text{iff}\quad \osmodel^\uparrow \models \psi
		& \tag{Theorem~\ref{thm:sbnfofoe}: strict normal form}
	\\
		& \quad\text{iff}\quad \osmodel \models \psi^\downarrow.
		& \tag{Proposition~\ref{prop:downsorts}}
	\end{align*}
	Observe that by construction $\psi^\downarrow$ is in basic normal form.
\end{proof}

%% file: os-monotonicity.tex

Given a one-step logic $\llang(A)$ and formula $\varphi \in \llang(A)$.
We say that $\varphi$ is \emph{monotone in $A' \subseteq A$} if for every one step model $(D,\val:A\to\wp D)$, $a\in A'$ and assignment $\ass:\fovar\to D$,
\[\text{If } (D,\val),\ass \models \varphi \text{ and } \val(a) \subseteq E \text{ then } (D,\val[a\mapsto E]),\ass \models \varphi.\]
We use $\llang^+(A)$ to denote the fragment of $\llang(A)$ composed of formulas monotone in all $a\in A$.

Monotonicity is usually tightly related to positivity. If the quantifiers are well-behaved 
then a formula $\varphi$ will usually be monotone in $a \in A$ iff $a$ has positive polarity in $\varphi$, that is, if all of its occurrences are under an even number of negations. This is the case for all one-step logics considered in this article. In this section we give a syntactic characterization of monotonicity for several one-step logics.

\medskip\noindent
\textit{Convention}.
Given $P,S \subseteq A$, we use $\tau^P_S$ to denote the \emph{$P$-positive} type given by $S$, defined as
\[
\tau^P_S(x) := \bigwedge \{a(x) \mid a\in S\} \land \bigwedge \{\lnot a(x) \mid a\in A \text{ such that } a\notin S \text{ and } a\notin P\}.
\]

\subsubsection{Monotone fragment of $\ofo$}

\begin{theorem}\label{thm:ofomonot}
A formula of $\ofo(A,\sorts)$ is monotone in ${A' \subseteq A}$ iff it is equivalent to a sentence given by the following grammar:
\[
\varphi ::= \psi \mid a(x) \mid \exists x{:}\asort.\varphi \mid \forall x{:}\asort.\varphi \mid \varphi \land \varphi \mid \varphi \lor \varphi
\]
where $a\in A'$, $s\in\sorts$ and $\psi \in \ofo(A\setminus A',\sorts)$. We denote this fragment as $\monot{\ofo}{A'}(A,\sorts)$.
\end{theorem}

\noindent The result follows from the following lemma.

\begin{lemma}
The following hold:
\begin{enumerate}
	\itemsep 0pt
	\item Every $\varphi \in \monot{\ofo}{A'}(A,\sorts)$ is monotone in $A'$.
	\item There exists a translation $(-)^\tmono:\ofo(A,\sorts) \to \monot{\ofo}{A'}(A,\sorts)$ such that
a formula ${\varphi \in \ofo(A,\sorts)}$ is monotone in $A'$ if and only if $\varphi\equiv \varphi^\tmono$.
\end{enumerate}
\end{lemma}
\begin{proof}
	In~\cite{DBLP:journals/corr/CarreiroFVZ14,LICS14} this result is proved for \emph{unsorted} \ofo. It is not difficult to adapt the proof for multi-sorted $\ofoe$. Intuitively, if we assume that $\varphi$ is in negation normal form, the translation is defined as $\varphi^\tmono := \varphi[\lnot a(x) \mapsto \top \mid a\in A']$.
\end{proof}

Combining the normal form theorem for $\ofo$ and the above lemma, we obtain the following corollary providing a normal form for the monotone fragment of $\ofo$.

\begin{corollary}\label{cor:ofopositivenf} Let $\varphi \in \ofo(A,\sorts)$, the following hold:
	\begin{enumerate}[(i)]
		\item The formula $\varphi$ is monotone in $A'\subseteq A$ iff it is equivalent to a formula in the basic form $\bigvee \bigwedge_\asort\mondgbnfofo{\Sigma}{\Pi}{A'}_\asort$ for some types $\Sigma,\Pi \subseteq \wp A$, where
		\[
		\mondgbnfofo{\Sigma}{\Pi}{A'}_\asort := \bigwedge_{S\in\Sigma} \exists x{:}\asort. \tau^{A'}_S(x) \land \forall x{:}\asort. \bigvee_{S\in\Pi} \tau^{A'}_S(x) .
		\]		
		\item The formula $\varphi$ is monotone in all $a\in A$ (i.e., $\varphi\in\ofo^+(A,\sorts)$) iff $\varphi$ is equivalent to a formula in the basic form $\bigvee \bigwedge_\asort\posdgbnfofo{\Sigma}{\Pi}_\asort$ for some types $\Sigma,\Pi \subseteq \wp A$, where
		\[
		\posdgbnfofo{\Sigma}{\Pi}_\asort := \bigwedge_{S\in\Sigma} \exists x{:}\asort. \tau^+_S(x) \land \forall x{:}\asort. \bigvee_{S\in\Pi} \tau^+_S(x) .
		\]
	\end{enumerate}
\end{corollary}

\noindent The following stronger normal form will be useful in the next section.

\begin{proposition}\label{prop:strongmonofo}
In the above normal form we can assume that every conjunct $\mondgbnfofo{\Sigma}{\Pi}{A'}_\asort$ is such that for every pair of distinct $S,S'\in \Sigma$ at least one of the following conditions hold:
\begin{itemize}
	\itemsep 0 pt
	\item $S \cap (A\setminus A') \neq S' \cap (A\setminus A')$, or
	\item $S \cap A' \not\subseteq S' \cap A'$ and $S' \cap A' \not\subseteq S \cap A'$.
\end{itemize}
\end{proposition}
\begin{proof}
	Assume that for some distinct $S,S'\in\Sigma$ neither of the conditions hold. That is,
	$S \cap (A\setminus A') = S' \cap (A\setminus A')$ and
	either (1) $S \cap A' \subseteq S' \cap A'$ or (2) $S' \cap A' \subseteq S \cap A'$.
	It is easy to observe that if (1) holds then $\tau^{A'}_{S'}(x) \models \tau^{A'}_{S}(x)$ and if (2) holds then $\tau^{A'}_{S}(x) \models \tau^{A'}_{S'}(x)$. Therefore we get that $\mondgbnfofo{\Sigma}{\Pi}{A'}_\asort\equiv \mondgbnfofo{\Sigma\setminus \{S\}}{\Pi}{A'}_\asort$ and $\mondgbnfofo{\Sigma}{\Pi}{A'}_\asort\equiv \mondgbnfofo{\Sigma\setminus \{S'\}}{\Pi}{A'}_\asort$ respectively. 
\end{proof}

\subsubsection{Monotone fragment of $\ofoe$}

\begin{theorem}
A formula of $\ofoe(A,\sorts)$ is monotone in ${A'\subseteq A}$ iff it is equivalent to a sentence given by
\[
\varphi ::= \psi \mid a(x) \mid \exists x{:}\asort.\varphi \mid \forall x.\varphi \mid \varphi \land \varphi \mid \varphi \lor \varphi
\]
where $a\in A'$, $\asort\in\sorts$ and $\psi \in \ofoe(A\setminus A',\sorts)$. We denote this fragment as $\monot{\ofoe}{A'}(A,\sorts)$.
\end{theorem}

\noindent The result follows from the following lemma.

\begin{lemma}
The following hold:
\begin{enumerate}
	\itemsep 0pt
	\item Every $\varphi \in \monot{\ofoe}{A'}(A,\sorts)$ is monotone in $A'$.
	\item There exists a translation $(-)^\tmono:\ofoe(A,\sorts) \to \monot{\ofoe}{A'}(A,\sorts)$ such that
a formula ${\varphi \in \ofoe(A,\sorts)}$ is monotone in $A'$ if and only if $\varphi\equiv \varphi^\tmono$.
\end{enumerate}
\end{lemma}
\begin{proof}
	In~\cite{DBLP:journals/corr/CarreiroFVZ14,LICS14} this result is proved for \emph{unsorted} \ofoe extended with generalized quantifiers. It is not difficult to adapt the proof for multi-sorted $\ofoe$. Intuitively, the translation is defined as $\varphi^\tmono := \varphi[\lnot a(x) \mapsto \top \mid a\in A']$.
\end{proof}

Combining the normal form theorem for $\ofoe$ and the above lemma, we obtain the following corollary providing a normal form for the monotone fragment of $\ofoe$.

\begin{corollary}\label{cor:ofoepositivenf}
	Given $\varphi \in \ofoe(A,\sorts)$, the following hold:
	\begin{enumerate}[(i)]
		\item The formula $\varphi$ is monotone in $A'\subseteq A$ iff it is equivalent to a formula in the basic form $\bigvee \bigwedge_\aSort \mondbnfofoe{\vlist{T}}{\Pi}{A'}_\aSort$
		where
		\[
			\mondbnfofoe{\vlist{T}}{\Pi}{A'}_\aSort := \exists \vlist{x}{:}\aSort!.\big(\arediff{\vlist{x}} \land \bigwedge_i \tau^{A'}_{T_i}(x_i) \land \forall z{:}\aSort!.(\arediff{\vlist{x},z} \lthen \bigvee_{S\in \Pi} \tau^{A'}_S(z))\big) .
		\]
		and for each conjunct there are $\vlist{T}$ and $\Pi$ satisfying $T_i \in \wp(A)$ and $\Pi\subseteq\vlist{T}$,
		\item The formula $\varphi$ is monotone in all $a\in A$ (i.e., $\varphi\in \ofoe^+(A)$) iff it is equivalent to a formula in the basic form $\bigvee \bigwedge_\aSort \posdbnfofoe{\vlist{T}}{\Pi}_\aSort$
		where
		\[
			\posdbnfofoe{\vlist{T}}{\Pi}_\aSort := \exists \vlist{x}{:}\aSort!.\big(\arediff{\vlist{x}} \land \bigwedge_i \tau^{+}_{T_i}(x_i) \land \forall z{:}\aSort!.(\arediff{\vlist{x},z} \lthen \bigvee_{S\in \Pi} \tau^{+}_S(z))\big) .
		\]
		and for each conjunct there are $\vlist{T}$ and $\Pi$ satisfying $T_i \in \wp(A)$ and $\Pi\subseteq\vlist{T}$.
		\item Over strict one-step models the above normal forms hold with $\aSort$ replaced by $\asort$.
	\end{enumerate}
\end{corollary}

%% file: os-additivity.tex

Before stating the main definitions we need to introduce some useful notations. Given a valuation $\val:A\to\wp D$, elements $\vlist{a} \in A^n$ and $\vlist{X} = (X_1,\dots,X_n) \in \wp(D)^n$, we introduce the following notation:
\begin{align*}
	\val(\vlist{a}) &:= \val(a_1),\dots,\val(a_n)\\
	\val[\vlist{a}\mapsto \vlist{X}] &:= \val[a_i \mapsto X_i \mid 1\leq i \leq n]\\
	\val[\vlist{a}\resto \vlist{X}] &:= \val[a_i \mapsto \val(a_i) \cap X_i \mid 1\leq i \leq n].
\end{align*}
%
We are now ready to state the main definition of this section.

\medskip
Consider a one-step logic $\llang(A)$ and formula $\varphi \in \llang(A)$.
%
We say that $\varphi$ is \emph{completely additive in $\{a_1,\dots,a_n\}\subseteq A$} if $\varphi$ is monotone in every $a_i$ and, for every one-step model $(D,\val)$ and assignment $\ass:\fovar\to D$,
\[
\text{If } (D,\val),\ass \models \varphi \text{ then } (D, \val[\vlist{a} \resto \vlist{Q}]),\ass \models \varphi \text{ for some quasi-atom $\vlist{Q}$ of $\val(\vlist{a})$},
\]
where $\vlist{a} := a_1,\dots,a_n$.
%
It will be useful to give a syntactic characterization of additivity for several one-step logics.

\subsubsection{Completely additive fragment of $\ofo$}

\begin{definition}\label{def:ofocadd}
The fragment of $\ofo(A,\sorts)$ completely additive in ${A'\subseteq A}$ is given by the sentences generated by the following grammar:
\[
\varphi ::= \psi \mid a(x) \mid \exists x{:}\asort.\varphi \mid \varphi \lor \varphi \mid \varphi \land \psi
\]
where $a\in A'$, $\asort\in\sorts$ and $\psi \in \ofo(A\setminus A',\sorts)$. We denote this fragment as $\add{\ofo}{A'}(A,\sorts)$.
\end{definition}

\begin{theorem}\label{thm:ofocadd}
A formula of $\ofo(A,\sorts)$ is completely additive in ${A'\subseteq A}$ iff it is equivalent to a sentence of $\add{\ofo}{A'}(A,\sorts)$.
\end{theorem}

\noindent The theorem will follow from the next two lemmas.

\begin{lemma}\label{lem:caddofoiscadd}
If $\varphi \in \add{\ofo}{A'}(A,\sorts)$ then $\varphi$ is completely additive in $A'$.
\end{lemma}
\begin{proof}
First observe that $\varphi$ is monotone in every $a\in A'$ by Theorem~\ref{thm:ofomonot}.
We show, by induction, that any one-step formula $\varphi$ in the fragment (which may not be a sentence) satisfies, for every one-step model $(D,\val:A\to\wp D)$, assignment ${\ass:\fovar\to D}$,
\[
\text{If } (D,\val),\ass \models \varphi \text{ then } (D, \val[A' \resto \vlist{Q}]),\ass \models \varphi \text{ for some quasi-atom $\vlist{Q}$ of $\val(A')$.}
\]
The cases are as follows:
\begin{itemize}
\item If ${\varphi = \psi \in \ofo(A\setminus A')}$ changes in the $A'$-part of the valuation will make no difference and hence the condition is trivial. 

\item Case $\varphi = a_i(x)$ with $a_i\in A'$: if $(D,\val),\ass \models a_i(x)$ then $\ass(x)\in \val(a_i)$. 
If we take $\vlist{Q}$ to be an atom of $\val(A')$ such that $Q_i := \{\ass(x)\}$ it is clear that
$\ass(x) \in V[A'\resto\vlist{Q}](a_i)$ and hence $(D, \val[A'\resto\vlist{Q}]),\ass \models a_i(x)$.

\item Case $\varphi = \varphi_1 \lor \varphi_2$: simply apply the inductive hypothesis to one of the disjuncts.

\item Case $\varphi = \varphi_1 \land \psi$: assume $(D,\val),\ass \models \varphi$. By induction hypothesis we have that $(D,\val[A'\resto\vlist{Q}]),\ass \models \varphi_1$ for some $\vlist{Q}$. Observe that $(D,\val),\ass \models \psi$ and as $\psi$ is $A'$-free we also have $(D,\val[A'\resto \vlist{Q}]),\ass \models \psi$. Therefore we can conclude that $(D,\val[a'\resto \vlist{Q}]),\ass \models \varphi$. 

\item Case $\varphi = \exists x{:}\asort.\varphi'$: assume $(D,\val),\ass \models \varphi$. By definition there exists $d\in D$ such that $(D,\val),\ass[x\mapsto d] \models \varphi'$. By induction hypothesis  $(D,\val[A'\resto\vlist{Q}]),\ass[x\mapsto d] \models \varphi'$ for some $\vlist{Q}$. Therefore we can conclude that $(D,\val[A'\resto\vlist{Q}]),\ass \models \exists x{:}\asort.\varphi'$. \qedhere
\end{itemize}
%
\end{proof}

\begin{lemma}\label{lem:ofotrans}
There exists a translation $(-)^\tadd:\monot{\ofo}{A'}(A,\sorts) \to \add{\ofo}{A'}(A,\sorts)$ such that
a formula ${\varphi \in \monot{\ofo}{A'}(A,\sorts)}$ is completely additive in $A'$ if and only if $\varphi\equiv \varphi^\tadd$.
\end{lemma}
\begin{proof}
We assume that $\varphi$ is in basic normal form, i.e., $\varphi = \bigvee \bigwedge_\asort \mondgbnfofo{\Sigma}{\Pi}{A'}_\asort$ where $\Sigma\subseteq\Pi$. 
First, we intuitively consider some conditions on subformulas of $\varphi$ that would force $\val(A')$ to have more than one element and, in particular, \emph{not} be a quasi-atom. Clearly, any formula that forces this, goes against the spirit of complete additivity.
\begin{enumerate}[(i)]
	\itemsep 0pt
	\item\label{ofo:it:sorts} $\bigwedge_\asort\mondgbnfofo{\Sigma}{\Pi}{A'}_\asort$ with $\Sigma_{\asort_1}$ and $\Sigma_{\asort_2}$ such that $a\in \Sigma_{\asort_1}$ and $b\in \Sigma_{\asort_2}$ for $a,b\in A'$.
	\item\label{ofo:it:exist} $\mondgbnfofo{\Sigma}{\Pi}{A'}_\asort$ with $a,b \in S\cap S'$ for distinct $a,b\in A'$ or distinct $S,S'$.
	\item\label{ofo:it:univ} $\mondgbnfofo{\Sigma}{\Pi}{A'}_\asort$ with $\Pi\cap A' \neq \nada$.
\end{enumerate}
Now, we give a translation which eliminates (replaces with $\bot$) the subformulas satisfying any of the above cases. We first take care of case~\ref{ofo:it:sorts} with the following definition
\[
(\bigvee \bigwedge_\asort \mondgbnfofo{\Sigma}{\Pi}{A'}_\asort)^\tadd :=
\bigvee \{\bigwedge_\asort \mondgbnfofo{\Sigma}{\Pi}{A'}_\asort^\tadd \mid \text{(i) is not the case}\}
\]
and we take care of the remaining cases as follows
\[
\mondgbnfofo{\Sigma}{\Pi}{A'}_\asort^\tadd :=
\begin{cases}
\bot & \text{if~\ref{ofo:it:exist} holds,}\\
\mondgbnfofo{\Sigma}{\Pi^{\times}_{A'}}{A'}_\asort & \text{otherwise,}
\end{cases}
\]
where 
$\Pi^{\times}_{A'} := \{S\in \Pi \mid A' \cap S = \nada\}$. 

From the construction it is clear that $\varphi^\tadd \in \add{\ofo}{A'}(A,\sorts)$ and therefore the right-to-left direction of the lemma is immediate by Lemma~\ref{lem:caddofoiscadd}. For the left-to-right direction assume that $\varphi$ is completely additive in $A'$, we have to prove that $(D,\val) \models \varphi$ iff $(D,\val) \models \varphi^\tadd$, for every one-step model $(D,\val)$.
Moreover, using Remark~\ref{rem:ofostrict} we will assume that $(D,\val)$ is a \emph{strict} one-step model.


\bigskip
\noindent \fbox{$\Leftarrow$}
Let $(D,\val) \models \varphi^\tadd$.
It is enough to show that for every conjunct, if $(D,\val) \models \mondgbnfofo{\Sigma}{\Pi^{\times}_{A'}}{A'}$ then $(D,\val) \models \mondgbnfofo{\Sigma}{\Pi}{A'}$. The key observation is that $\Pi^{\times}_{A'} \subseteq \Pi$.

\bigskip
\noindent \fbox{$\Rightarrow$}
Let $(D,\val) \models \varphi$.
To prove this direction we make a slight detour: first, it is not difficult to prove that $(D,\val) \equiv_{\fo} (D\times\{0,1\},\val_\pi)$ where $D\times\{0,1\} := D_{\asort_1}\times\{0,1\},\dots,D_{\asort_n}\times\{0,1\}$ and $\val_{\pi}(a) := \{ (d,k) \mid d \in \val(a), k \in \{0,1\}\}$. Therefore, to prove this direction it is enough to show that $(D\times\{0,1\},\val_\pi) \models \varphi^\tadd$.

By complete additivity in $A'$
we have that $(D\times\{0,1\},\val_\pi[A'\resto\vlist{Q}]) \models \varphi$ for some quasi-atom $\vlist{Q}$ of $\val_\pi(A')$. To improve readability we define $\val_\pi' := \val_\pi[A'\resto\vlist{Q}]$ and $D_{01} := D\times\{0,1\}$.
We now work with $(D_{01},\val_\pi')$ because (by monotonicity) it will be enough to prove $(D_{01},\val_\pi') \models \varphi^\tadd$ to obtain $(D_{01},\val_\pi) \models \varphi^\tadd$. 

As $(D_{01},\val_\pi') \models \varphi$, we know there is some disjunct $\bigwedge_\asort \mondgbnfofo{\Sigma}{\Pi}{A'}_\asort$ of $\varphi$ witnessing the satisfaction. First, we prove that this disjunct is preserved by the translation.
\begin{claimfirst}\label{lem:ofotrans:c0}
	The disjunct $\bigwedge_\asort \mondgbnfofo{\Sigma}{\Pi}{A'}_\asort$ does not satisfy case~\ref{ofo:it:sorts}.
\end{claimfirst}
\begin{pfclaim}
	Suppose that for this disjunct there are two conjuncts corresponding to sorts $\asort_1$ and $\asort_2$ such that $a\in \Sigma_{\asort_1}$ and $b\in \Sigma_{\asort_2}$ for $a,b\in A'$. As the sorts are disjoint, this implies that there should be two \emph{distinct} elements colored by elements of $A'$. However, as $\val_\pi'(A')$ is a quasi-atom, this cannot be the case.
\end{pfclaim}

From the above claim it follows that, for the previously fixed disjunct, there is at most one sort (i.e, one conjunct) which can possibly $A'$ in the existential part (that is, in $\Sigma$). Hence, the disjunct is (so far) preserved by the translation. We still have to check that every conjunct is preserved, that is, we now focus on cases~\ref{ofo:it:exist} and~\ref{ofo:it:univ}. We fix an arbitrary conjunct $\mondgbnfofo{\Sigma}{\Pi}{A'}_\asort$ and prove the following claim.
\begin{claim}\label{lem:ofotrans:c1}
	For every $b,b'\in A'$ and $S,S'\in\Sigma$, if $b \in S$ and $b'\in S'$ then $b=b'$ and $i=j$.
\end{claim}
\begin{pfclaim}
	Assume, towards a contradiction, that one of the following cases hold:
	\begin{enumerate}[(1)]
		\itemsep 0 pt
		\item There are \emph{distinct} $b,b'\in A'$ such that $b,b' \in S\cup S'$ for some $S,S' \in \Sigma$.
		\item There is $b\in A'$ such that $b \in S \cap S'$ for some \emph{distinct} $S,S' \in \Sigma$.
	\end{enumerate}
	Case (1) would require both $\val'_\pi(b)$ and $\val'_\pi(b')$ to be nonempty, which does not hold, as $\val'_\pi(A')$ is a quasi-atom. For case (2) observe that, by Proposition~\ref{prop:strongmonofo}, we have to consider the following two subcases:
	\begin{enumerate}[({2}a)]
		\itemsep 0 pt
		\item $S \cap (A\setminus A') \neq S' \cap (A\setminus A')$.
		\item $S \cap A' \not\subseteq S' \cap A'$ and $S' \cap A' \not\subseteq S \cap A'$.
	\end{enumerate}
	For the subcase (2a) note that if elements $d_S, d_{S'} \in D_{01}$ satisfy $\tau^{A'}_S(d_S)$ and $\tau^{A'}_{S'}(d_{S'})$ then $d_S \neq d_{S'}$ must hold. This would be absurd since at most one element is colored with $b$ (because $\val'_\pi(A')$ is a quasi-atom) and hence the existentials wouldn't be satisfied, because they require at least two \emph{distinct} elements to satisfy $b(x)$.

	The subcase (2b) is slightly more subtle: for this case to hold, we must have \emph{distinct} $b_1,b_2 \in A'$ such that $b_1\in S$ and $b_2\in S'$. The situation is now like case (1), which we already worked out. Therefore this finishes the proof of the claim.
\end{pfclaim}
To finish, we show that condition~\ref{ofo:it:univ} is taken care of.
\begin{claim}
	If $(D_{01},\val_\pi') \models \mondgbnfofo{\Sigma}{\Pi}{A'}_\asort$ then $(D_{01},\val_\pi') \models \mondgbnfofo{\Sigma}{\Pi^{\times}_{A'}}{A'}_\asort$.
\end{claim}
\begin{pfclaim}
	The existential part is trivial. For the universal part suppose, without loss of generality, that some element $(d,0)$ has type $S_0\in\Pi$. If $S\cap A'=\nada$ we are done, since in that case $S_0\in\Pi^{\times}_{A'}$. Suppose now that $S_0 \cap A'\neq\nada$. The two key observations are (1)~because of how $\val_\pi$ was defined, and because $\val_\pi'(A')$ is a quasi-atom, the type $S_1\in\Pi$ of $(d,1)$ has to be exactly $S_1 = S_0 \cap (A\setminus A')$; and (2)~as we are considering $A'$-positive types, element $(d,0)$ also satisfies $S_1$. From this two observations we can conclude that every element $(d,i)$ satisfies some type in $\Pi^{\times}_{A'}$.
\end{pfclaim}
This finishes the proof.
\end{proof}

Putting together the above lemmas we obtain Theorem~\ref{thm:ofocadd}. Moreover, a careful analysis of the translation gives us normal forms for the completely additive fragment of $\ofo$. 

\begin{corollary}\label{cor:ofoadditivenf}
	Let $\varphi \in \ofo(A,\sorts)$, the following hold:
	\begin{enumerate}[(i)]
		\item The formula $\varphi$ is completely additive in $A' \subseteq A$ iff it is equivalent to a formula in the basic form $\bigvee \bigwedge_\asort \mondgbnfofo{\Sigma}{\Pi}{A'}_\asort$ where $\Sigma,\Pi \subseteq \wp (A)$ and for every disjunct $\bigwedge_\asort \mondgbnfofo{\Sigma}{\Pi}{A'}_\asort$,
		\begin{enumerate}[1.]
			\itemsep 0pt
			\item At most one sort $\underline{\asort}\in\sorts$ may use elements from $A'$ in $\mondgbnfofo{\Sigma}{\Pi}{A'}_{\underline{\asort}}$,
			\item For $\mondgbnfofo{\Sigma}{\Pi}{A'}_{\underline{\asort}}$ we have that $\Pi$ is $A'$-free and, if $L_\Sigma \in A^*$ is the list with repetitions of elements of $A$ in $\Sigma$, then there is at most one element of $A'$ in $L_\Sigma$.
		\end{enumerate}
		\item If $\varphi$ is monotone in $A$ (i.e., $\varphi\in\ofo^+(A,\sorts)$) then $\varphi$ is completely additive in $A' \subseteq A$ iff it is equivalent to a formula of the form $\bigvee \bigwedge_\asort \posdgbnfofo{\Sigma}{\Pi}_\asort$ where $\Sigma,\Pi \subseteq \wp (A)$ and for every disjunct $\bigwedge_\asort \posdgbnfofo{\Sigma}{\Pi}_\asort$,
		\begin{enumerate}[1.]
			\itemsep 0pt
			\item At most one sort $\underline{\asort}\in\sorts$ may use elements from $A'$ in $\posdgbnfofo{\Sigma}{\Pi}_{\underline{\asort}}$,
			\item For $\posdgbnfofo{\Sigma}{\Pi}_{\underline{\asort}}$ we have that $\Pi$ is $A'$-free and, if $L_\Sigma \in A^*$ is the list with repetitions of elements of $A$ in $\Sigma$, then there is at most one element of $A'$ in $L_\Sigma$.
		\end{enumerate}
	\end{enumerate}
\end{corollary}

\subsubsection{Completely additive fragment of $\ofoe$}

\begin{definition}\label{def:ofoeadd}
Given ${A'\subseteq A}$, the fragment $\add{\ofoe}{A'}(A,\sorts)$ of $\ofoe(A,\sorts)$ is given by the sentences generated by the following grammar:
\[
\varphi ::= \psi \mid a(x) \mid \exists x{:}\asort.\varphi \mid \varphi \lor \varphi \mid \varphi \land \psi
\]
where
$a\in A'$, $\asort\in\sorts$ and $\psi \in \ofoe(A\setminus A',\sorts)$. Observe that the equality is included in $\psi$. 
\end{definition}

\begin{theorem}\label{thm:ofoeadd}
A formula of $\ofoe(A,\sorts)$ is completely additive in ${A'\subseteq A}$ iff it is equivalent to a sentence in $\add{\ofoe}{A'}(A,\sorts)$.
\end{theorem}

The theorem will follow from the next two lemmas.

\begin{lemma}\label{lem:cofoeisadd}
Every $\varphi \in \add{\ofoe}{A'}(A,\sorts)$ is completely additive in $A'$.
\end{lemma}
\begin{proof}
The proof is the same as for Lemma~\ref{lem:caddofoiscadd} (complete additivity for \ofo).
\end{proof}

\begin{lemma}\label{lem:ofoetrans}
	There exists a translation $(-)^\tadd:\monot{\ofoe}{A'}(A,\sorts) \to \add{\ofoe}{A'}(A,\sorts)$ such that
a formula $\varphi \in \monot{\ofoe}{A'}(A,\sorts)$ is completely additive in $A'$ if and only if $\varphi\equiv \varphi^\tadd$.
\end{lemma}

\begin{proof}
We assume that $\varphi$ is in basic form, i.e., $\varphi = \bigvee \bigwedge_\aSort \mondbnfofoe{\vlist{T}}{\Pi}{A'}_\aSort$ with $\Pi\subseteq\vlist{T}$. 
First, we intuitively consider some conditions on subformulas of $\varphi$ that would force the existence of at least two elements colored with $A'$. Clearly, any formula that forces this, goes against the spirit of complete additivity:
\begin{enumerate}[(i)]
	\itemsep 0pt
	\item\label{ofoe:it:sorts} Some $\bigwedge_\aSort \mondbnfofoe{\vlist{T}}{\Pi}{A'}_\aSort$ has $\vlist{T}_{\aSort_1}$ and $\vlist{T}_{\aSort_2}$ with $a\in \vlist{T}_{\aSort_1}$ and $b\in \vlist{T}_{\aSort_2}$ for $a,b\in A'$, $\aSort_1\neq\aSort_2$.
	\item\label{ofoe:it:exist} For any $\mondbnfofoe{\vlist{T}}{\Pi}{A'}_\aSort$ there are $a,b \in T_i\cap T_j$ for distinct $a,b\in A'$ or distinct $i,j$.
	\item\label{ofoe:it:univ} For any $\mondbnfofoe{\vlist{T}}{\Pi}{A'}_\aSort$ we have $\Pi\cap A' \neq \nada$.
\end{enumerate}
Now, we give a translation which eliminates (replaces with $\bot$) the subformulas satisfying any of the above cases. We first take care of case~\ref{ofoe:it:sorts} with the following definition
\[
(\bigvee \bigwedge_\aSort \mondbnfofoe{\vlist{T}}{\Pi}{A'}_\aSort)^\tadd :=
\bigvee \{\bigwedge_\aSort \mondbnfofoe{\vlist{T}}{\Pi}{A'}_\aSort^\tadd \mid \text{(i) is not the case}\}
\]
and we take care of the remaining cases as follows
\[
\mondbnfofoe{\vlist{T}}{\Pi}{A'}_\aSort^\tadd :=
\begin{cases}
	\bot &\text{ if~\ref{ofoe:it:exist} holds,} \\
	\mondbnfofoe{\vlist{T}}{\Pi^{\times}_{A'}}{A'}_\aSort &\text{ otherwise,}
\end{cases}
\]
where $\Pi^{\times}_{A'} := \{S\in\Pi \mid A'\cap S=\nada\}$. 
%

First we prove the right-to-left direction of the lemma. Inspecting the syntactic form of $\varphi^\tadd$ and using Lemma~\ref{lem:cofoeisadd} it is not difficult to see that $\varphi^\tadd \in \add{\ofoe}{A'}(A,\sorts)$.
For the left-to-right direction of the lemma we assume $\varphi$ to be completely additive in $A'$ and have to prove $\varphi \equiv \varphi^\tadd$.

\bigskip
\noindent\fbox{$\Leftarrow$}
Assume $(D,\val) \models \varphi^\tadd$. It is enough to show that $(D,\val) \models \mondbnfofoe{\vlist{T}}{\Pi^{\times}_{A'}}{A'}_\aSort$ implies $(D,\val) \models \mondbnfofoe{\vlist{T}}{\Pi}{A'}_\aSort$ for every conjunct. The key observation is that $\Pi^{\times}_{A'} \subseteq \Pi$.

\bigskip
\noindent\fbox{$\Rightarrow$}
Let $(D,\val) \models \varphi$. By complete additivity in $A'$
we have that $(D,\val[A'\resto\vlist{Q}]) \models \varphi$ for some quasi-atom $\vlist{Q}$ of $\val(A')$. To improve readability we define $\val' := \val[A'\resto\vlist{Q}]$.
We now work with $(D,\val')$ because (by monotonicity, which is implied by complete additivity) it will be enough to prove that $(D,\val') \models \varphi^\tadd$. 

As $(D,\val') \models \varphi$, we know there is some disjunct $\bigwedge_\aSort \mondbnfofoe{\vlist{T}}{\Pi}{A'}_\aSort$ of $\varphi$ witnessing the satisfaction. First, we prove that this disjunct is preserved by the translation.
\begin{claimfirst}\label{lem:ofoeatrans:c0}
	The disjunct $\bigwedge_\aSort \mondbnfofoe{\vlist{T}}{\Pi}{A'}_\aSort$ does not satisfy case~\ref{ofoe:it:sorts}.
\end{claimfirst}
\begin{pfclaim}
	Same as in Claim~\ref{lem:ofotrans:c0} of Lemma~\ref{lem:ofotrans}.
\end{pfclaim}
From the above claim it follows that, for the previously fixed disjunct, there is at most one sort (i.e, one conjunct) which can possibly use $A'$ in the existential part (that is, in $\vlist{T}$). Hence, the disjunct is (so far) preserved by the translation. We still have to check that every conjunct is preserved, that is, we now focus on cases~\ref{ofoe:it:exist} and~\ref{ofoe:it:univ}. We fix an arbitrary conjunct $\mondbnfofoe{\vlist{T}}{\Pi}{A'}_\aSort$ and prove the following claim.
\begin{claim}\label{lem:ofoeatrans:c1}
	For every $b,b'\in A'$, if $b \in T_i$ and $b'\in T_j$ then $b=b'$ and $i=j$.
\end{claim}
\begin{pfclaim}
	Suppose that there are distinct $T_i,T_j \in \vlist{T}$ such that $b\in T_i\cap T_j$. This would require at least two \emph{distinct} elements to satisfy $b(x)$. However, this cannot occur because $\val'(A')$ is a quasi-atom. The case where $b\neq b'$ is handled in a similar way: suppose that $b \in T_i$, $b'\in T_j$ and $b\neq b'$. Using what we just proved, let us assume that $i=j$. Therefore, this requires the existance of an element which is colored with both $b$ and $b'$. However, as $\val'(A')$ is a quasi-atom, this cannot occur if $b\neq b'$.
\end{pfclaim}
%
%
To finish, we show that condition~\ref{ofoe:it:univ} is taken care of.
\begin{claim}
	If $(D,\val') \models \mondbnfofoe{\vlist{T}}{\Pi}{A'}_\aSort$ then $(D,\val') \models \mondbnfofoe{\vlist{T}}{\Pi^{\times}_{A'}}{A'}_\aSort$.
\end{claim}
\begin{pfclaim}
Assume $(D,\val') \models \mondbnfofoe{\vlist{T}}{\Pi}{A'}_\aSort$ and that $d\in D$ is one of the elements which is not a witness for $\vlist{T}$; therefore, $d$ has to satisfy some type $S_d \in \Pi$. If $S_d \cap A' = \nada$ we are done, because in that case $S_d \in \Pi^{\times}_{A'}$. Suppose that $S_d \cap A' \neq \nada$, this means that $d$ is colored with some $a\in A'$. As $\val'(A')$ is a quasi-atom, this means that no other element can be colored with $A'$. The final observation is that, as $\Pi\subseteq \vlist{T}$, then $S_d\in \vlist{T}$. This means that there should exist an element $d'\neq d$ which is colored with the same $a\in A'$ but we have just observed that this cannot occur. We conclude that every $d\in D$ has to satisfy some type $S\in\Pi$ with $S\cap A' = \nada$.
\end{pfclaim}

%
\noindent This finishes the proof.
\end{proof}

Putting together the above lemmas we obtain Theorem~\ref{thm:ofoeadd}. Moreover, a careful analysis of the translation gives us the following corollary, providing normal forms for the completely additive fragment of $\ofoe$.

\begin{corollary}\label{cor:ofoeadditivenf}
	Let $\varphi \in \ofoe(A)$, the following hold:
	\begin{enumerate}[(i)]
		\item The formula $\varphi$ is completely additive in $A' \subseteq A$ iff it is equivalent to a formula in the basic form $\bigvee \bigwedge_\aSort \mondbnfofoe{\vlist{T}}{\Pi}{A'}_\aSort$ where $\vlist{T}\in\wp(A)^k$, $\Pi\subseteq \vlist{T}$ and for every disjunct $\bigwedge_\aSort \mondbnfofoe{\vlist{T}}{\Pi}{A'}_\aSort$,
		\begin{enumerate}[1.]
			\itemsep 0pt
			\item At most one sort $\underline{\aSort}\in\sorts$ may use elements from $A'$ in $\mondbnfofoe{\vlist{T}}{\Pi}{A'}_{\underline{\aSort}}$,
			\item For $\mondbnfofoe{\vlist{T}}{\Pi}{A'}_{\underline{\aSort}}$ we have that $\Pi$ is $A'$-free and there is at most one element of $A'$ in the concatenation of the lists $T_1{\cdot} T_2 {\cdots} T_k$.
		\end{enumerate}
		\item If $\varphi$ is monotone in $A$ (i.e., $\varphi\in\ofoe^+(A)$) then $\varphi$ is completely additive in $A'\subseteq A$ iff it is equivalent to a formula in the basic form $\bigvee \bigwedge_\aSort \posdbnfofoe{\vlist{T}}{\Pi}_\aSort$ where $\vlist{T}\in\wp(A)^k$, $\Pi\subseteq \vlist{T}$ and for every disjunct $\bigwedge_\aSort \posdbnfofoe{\vlist{T}}{\Pi}_\aSort$,
		\begin{enumerate}[1.]
			\itemsep 0pt
			\item At most one sort $\underline{\aSort}\in\sorts$ may use elements from $A'$ in $\posdbnfofoe{\vlist{T}}{\Pi}_{\underline{\aSort}}$,
			\item For $\posdbnfofoe{\vlist{T}}{\Pi}_{\underline{\aSort}}$ we have that $\Pi$ is $A'$-free and there is at most one element of $A'$ in the concatenation of the lists $T_1{\cdot} T_2 {\cdots} T_k$.
		\end{enumerate}
		\item Over strict one-step models the above normal forms hold with $\aSort$ replaced by $\asort$.
	\end{enumerate}
\end{corollary}

%% file: os-multiplicativity.tex

Consider a one-step logic $\llang(A)$ and formula $\varphi \in \llang(A)$.
We say that $\varphi$ is \emph{completely multiplicative in $\{a_1,\dots,a_n\}\subseteq A$} if $\varphi$ is monotone in all $a_i$ and for all one-step models $(D,\val)$ and assignments $\ass:\fovar\to D$,
\[
\text{If } (D,\val),\ass \not\models \varphi \text{ then } (D, \val^c[\vlist{a} \resto\vlist{Q}]),\ass \not\models \varphi \text{ for some quasi-atom $\vlist{Q}$ of $\val^c(\vlist{a})$}
\]
where $\vlist{a} := a_1,\dots,a_n$ and $\val^c(b) := D\setminus \val(b)$ for all $b\in A$.
Observe that, already with the abstract definition of Boolean dual given in Definition~\ref{d:bdual1} we can prove the expected relationship between the notions of additivity and multiplicativity.

\begin{proposition}\label{prop:adddualmult}
	A formula $\varphi \in \llang(A)$ is completely additive in $A'\subseteq A$ if and only if $\varphi^\delta$ is completely multiplicative in $A'$.
\end{proposition}
\begin{proof}
We prove the left to right direction:
\begin{align}
& (D,\val) \not\models \varphi^\delta\\
\text{ iff } & (D,\val^c) \models \varphi
\tag{Definition~\ref{d:bdual1}} \\
\text{ iff } & (D,\val^c[A'\resto\vlist{Q}]) \models \varphi 
,\text{ for some quasi-atom $\vlist{Q}$ of $\val^c(A')$.}
\tag{$\varphi$ completely additive} 
%
\end{align}
%
The proof of the other direction is analogous.
\end{proof}

To define a syntactic notion of multiplicativity we first give a concrete definition of the dualization operator of Definition~\ref{d:bdual1} and then show that the one-step language $\ofoe$ is closed under Boolean duals.

\begin{definition}\label{DEF_dual} 
Let $\varphi \in \ofoe(A,\sorts)$. 
The \emph{dual} $\varphi^{\delta} \in \ofoe(A,\sorts)$ of $\varphi$ is defined 
as follows.
\begin{align*}
 (a(x))^{\delta} & :=  a(x) 
\\ (\top)^{\delta} & :=  \bot 
  & (\bot)^{\delta} & :=  \top 
\\  (x \foeq y)^{\delta} & :=  x \not\foeq y 
  & (x \not\foeq y)^{\delta}& :=  x \foeq y 
\\ (\varphi \land \psi)^{\delta} &:=  \varphi^{\delta} \lor \psi^{\delta} 
  &(\varphi \lor \psi)^{\delta}& :=  \varphi^{\delta} \land \psi^{\delta}
\\ (\exists x{:}\asort.\psi)^{\delta} &:=  \forall x{:}\asort.\psi^{\delta} 
  &(\forall x{:}\asort.\psi)^{\delta} &:=  \exists x{:}\asort.\psi^{\delta} 
\end{align*}
\end{definition}

\begin{remark}
	Observe that if $\varphi \in \ofo(A,\sorts)$ then $\varphi^{\delta} \in \ofo(A,\sorts)$ and that the operator preserves positivity of the predicates. That is, if $\varphi \in \ofoe^+(A,\sorts)$ then $\varphi^{\delta} \in \ofoe^+(A,\sorts)$ and the same occurs with $\ofo^+(A,\sorts)$.
\end{remark}

\noindent The proof of the following Proposition is a routine check.

\begin{proposition}\label{prop:duals}
The sentences $\varphi$ and $\varphi^{\delta}$ are Boolean duals, for every $\varphi \in \ofoe(A,\sorts)$.
\end{proposition}

We are now ready to give the syntactic definition of a completely multiplicative fragment for the one-step logics into consideration.

\begin{definition}\label{def:multfrag}
	\fcerror{Ugly!}Let $A$ be a set of names. The syntactic fragments of $\ofoe(A,\sorts)$ and $\ofo(A,\sorts)$ which are \emph{completely multiplicative} in $A'\subseteq A$ are given by
	\begin{align*}
		\mult{\ofo}{A'}(A,\sorts) &:= \{\varphi \mid \varphi^\delta \in \add{\ofo}{A'}(A,\sorts)\} \\
		\mult{\ofoe}{A'}(A,\sorts) &:= \{\varphi \mid \varphi^\delta \in \add{\ofoe}{A'}(A,\sorts)\} .
	\end{align*}
\end{definition}

\begin{proposition}
	A formula $\varphi \in \ofoe(A,\sorts)$ is completely multiplicative in $A'\subseteq A$ if and only if it is equivalent to some ${\varphi' \in \mult{\ofoe}{A'}(A,\sorts)}$.
\end{proposition}

\begin{proof} This is a consequence of Proposition~\ref{prop:adddualmult}, Theorem~\ref{thm:ofoeadd} and Definition~\ref{def:multfrag}.
\end{proof}

%% file: os-decidability.tex

In the following corollary, we briefly show that for any formula of $\ofo$ and $\ofoe$ we can effectively compute its normal form. Moreover, if the formula is monotone we can also compute its monotone normal form (\emph{cf.}~Corollary~\ref{cor:ofopositivenf}). The same holds for complete additivity and complete multiplicativity.

\begin{corollary}\label{cor:osnormalize}
	For every $\psi\in \ofo(A,\sorts)$ and $\varphi\in\ofoe(A,\sorts)$ we can effectively calculate its normal form, monotone normal form, completely additive normal form and completely multiplicative normal form.
\end{corollary}
\begin{proof}
	We only show the corollary for the normal form of arbitrary formulas of $\ofoe$ and, as an example, the normal form for completely additive formulas of $\ofoe$. The other cases are similar left to the reader.

	According to Theorem~\ref{thm:bnfofoe}, every $\varphi \in \ofoe(A,\sorts)$ is equivalent to a formula of the form $\bigvee\bigwedge_\aSort \dbnfofoe{\vlist{T}}{\Pi}_\aSort$ where for each conjunct $\vlist{T} \in \wp(A)^k$ for some $k$ and $\Pi \subseteq \vlist{T}$. We non-deterministically guess the number of disjuncts and parameters $k$, $\Pi$ and $\vlist{T}$ for each conjunct and repeatedly check wether the formulas $\varphi$ and $\bigvee\bigwedge_\aSort \dbnfofoe{\vlist{T}}{\Pi}_\aSort$ are equivalent. This check can be done because $\ofoe$ is decidable: in~\cite{monofoe1,monofoe2} it is proved that \emph{unsorted} $\ofoe$ is decidable. Multi-sorted $\ofoe$ can be reduced to unsorted $\ofoe$ by introducing new predicates for the sorts (a standard trick). 

	Next, we also want the formulas which are completely additive to be in their corresponding normal form given by Corollary~\ref{cor:ofoeadditivenf}. What we do now is to apply the translation $(-)^\tadd$ of Lemma~\ref{lem:ofoetrans} to every formula and keep only those that satisfy that $\varphi \equiv \varphi^\tadd$. The set $A'\subseteq A$ in which the formula $\varphi$ should be completely additive is guessed non-deterministically as well, we keep the biggest set.

	The same procedure can also be performed to get normal forms for completely multiplicative formulas, with additional dualization steps.
\end{proof}

%% file: soaut-def.tex

In this section we formally define the additive-weak automata discussed in the introduction (Section~\ref{sec:intro}). We start by introducing the concept of weak automata, together with some intuitions, and then move towards the definition of additive-weak automata.

\begin{definition}
\label{def:weak}
Let $\llang$ be a one-step language, and let $\aut = \tup{A,\tmap,\pmap,a_I}$
be in $\Aut(\llang,\props)$.
Given $a,b\in A$,
we say that there is a transition from $a$ to $b$ (notation: $a \leadsto b$)
if $b$ occurs in $\tmap(a,c)$ for some $c \in \wp(\props)$.
We let the \emph{reachability} relation $\ord$ denote the reflexive-transitive
closure of the relation $\leadsto$.


We say that $\pmap$ is a \emph{weak} parity condition, and $\aut$ is a
\emph{weak} parity automaton if we have
\begin{description}
\item[(weakness)] if $a \ord b$ and $b \ord a$ then $\pmap(a) = \pmap(b)$.
\end{description}
\end{definition}

The intuition is that every run of a weak automaton $\aut$ stabilizes on some strongly connected component $\mccomp$ after finitely many steps, and therefore the only parity seen infinitely often after that point will be the parity of $\mccomp$. Moreover, as only \emph{one} parity can be repeated infinitely often, the precise number does not matter; only the parity does:

\begin{fact}
Any weak parity automaton $\aut$ is equivalent to a weak parity automaton
$\aut'$ with $\pmap: A' \to \{0,1\}$. From now on we assume such a map for weak parity automata.
\end{fact}

If we think about trees, the leading intuition is that weak parity automata are those unable to register non-trivial properties concerning the `vertical dimension' of input trees. In some sense, they can only describe properties of well-founded (and co-well-founded) subsets of trees. Indeed, in~\cite{MullerSaoudiSchupp92} it is shown that on trees \emph{with bounded branching} weak automata characterize weak \mso (WMSO). However, if the branching of the tree is not bounded, the story is quite different, since an extra `horizontal' constraint is required to capture WMSO. We refer to~\cite{DBLP:conf/lics/FacchiniVZ13,Zanasi:Thesis:2012,LICS14,DBLP:journals/corr/CarreiroFVZ14} for more details.

\medskip
We now turn to the second condition that we will be interested in, viz., additivity. Intuitively, this property expresses a constraint on how much of the `horizontal dimension' of an input tree the automaton is allowed to process.
First we formulate our additivity condition abstractly in the setting of $\Aut(\llang)$.
Given the semantics of the one-step language $\llang$, the (semantic) notion
of additivity/multiplicativity applies to one-step formulas (see for instance
Section~\ref{subsec:one-stepcont}). We can then formulate the following requirement on automata from $\Aut(\llang)$:

\begin{description}
\item[(additivity)] for every \emph{maximal} strongly connected component $\mccomp \subseteq A$, states $a,b \in \mccomp$ and color $c\in \wp(\props)$:
    if $\pmap(b)=1$ then $\tmap(a,c)$ is completely additive in $\mccomp$.
    If $\pmap(b)=0$, then $\tmap(a,c)$ is completely multiplicative in $\mccomp$.
\end{description}

Intuitively, the additivity restriction has the following effect: while a run of an additive automaton stays inside a connected component with parity $1$, we can assume without loss of generality that the nodes of the tree coloured with some state of $\mccomp$ form a \emph{path} in the tree. The reason being that at each step --because of complete additivity-- player \eloise can play a valuation where at most one node is colored with $\mccomp$. Therefore, if \abelard chooses the element coloured by $\mccomp$, a repetition of this step will define a path.

\medskip
For the automata used in this article we need to combine the constraints for the horizontal and vertical dimensions, yielding automata with both the weakness and additivity constraints. The intuition is that we want to define a class of automata that works with \emph{finite paths}.

\begin{definition}
An \emph{additive-weak parity automaton} is an automaton $\aut \in \Aut(\llang)$ additionally satisfying both the \textbf{(weakness)} and \textbf{(additivity)} conditions.
We let $\AutWA(\llang)$ denote the class of such automata.
\end{definition}

Observe that, so far, the additivity condition has been given semantically.
However, given that the one-step languages that we are interested in have
a \emph{syntactic characterization} of complete  additivity (see for example Theorem~\ref{thm:ofoeadd})
we will give concrete definitions of these automata that take advantage of the mentioned characterizations.

\begin{definition}
The class $\AutWA(\ofoe)$ of automata is concretely given by the automata
$\aut = \tup{A,\tmap,\pmap,a_I}$ from $\Aut(\ofoe)$ such that for every \emph{maximal} strongly connected component $\mccomp \subseteq A$ and states $a,b \in \mccomp$ 
the following conditions hold:
\begin{description}
	\itemsep 0 pt
	\item[(weakness)] $\pmap(a)=\pmap(b)$,
	\item[(additivity)] for every color $c\in\wp(\props)$:\\
	If $\pmap(a)$ is odd then $\tmap(a,c) \in \add{\ofoe^+}{\mccomp}(A)$, otherwise\\
	if $\pmap(a)$ is even then $\tmap(a,c) \in \mult{\ofoe^+}{\mccomp}(A)$.
\end{description}
\end{definition}

In the following sections we analyze certain closure properties of $\AutWA(\ofoe)$ in the class of \emph{strict} trees. Namely, closure under Boolean operations and under (weak chain) projection. We start with the latter. As usual, to prove the closure under projection we first prove a simulation theorem.

%% file: soaut-simulation.tex

One of the main technical results for parity automata is the so-called ``Simulation Theorem'':

\begin{theorem}[\cite{Walukiewicz96,Walukiewicz02}]\label{thm:origsimulation}
	Every automaton $\aut\in\Aut(\ofoe)$ is equivalent (over all models) to a \emph{non-deterministic} automaton $\aut'\in\Aut(\ofoe)$.
\end{theorem}

Very informally, an automaton $\aut$ is called non-deterministic when in every acceptance game $\agame(\aut,\model)$, if \abelard can choose to play both $(a,s)$ and $(b,s)$ at a given moment, then $a=b$. That is, \abelard's power boils down to being a \emph{pathfinder} in $\model$. He chooses the elements of $\model$ whereas the state of $\aut$ is `fixed' by the valuation played by \eloise, for every given $s$. For a formal definition we refer the reader to Definition~12 and Lemma~19--20 in~\cite{Walukiewicz96}.

To get a better picture of what non-determinism means, it is good to do the following: first observe that if we fix a strategy $f$ for \eloise for the game $\agame(\aut,\tmodel)$ then the whole game can be represented by a tree, whose nodes are the different admissible moves for \abelard. We assume that the automaton is clear from context and denote such a tree by $\tmodel_{\!f}$. Figure~\ref{fig:strategies1} shows the move-tree for \abelard for some fixed strategy $f$ for \eloise. Each branch of $\tmodel_{\!f}$ represents a possible $f$-guided match.

\input{fig-strategies-1.tex}

Observe that in this figure the chosen strategy for \eloise is \emph{not} non-deterministic. The admissible moves which violate this condition are underlined. On trees, the notion of non-deterministic strategy can be rephrased as ``every element $s\in\tmoddom$ occurs at most once as an admissible move for \abelard.''

\begin{remark}
	The terminology ``non-deterministic'' may seem confusing, given that $\aut'$ is certainly more ``deterministic''  than $\aut$ (from the point of view of \abelard.) However, the terminology is sensible when seen from the following perspective: we say that a finite state automaton (on words) is \emph{deterministic} when the next state is uniquely determined by the current state (and the input); on the other hand, they are called \emph{non-deterministic} when \eloise can choose between different transitions, leading to the next state; finally, \emph{alternating} finite state automata are a generalization where the next state is chosen by a complex interaction of \eloise and \abelard.
	Going back to parity automata, the above theorem then says that every alternating parity automaton is equivalent to a non-deterministic automaton. In light of our brief discussion, it should be clear that non-deterministic automata are ``more deterministic'' than alternating automata.
\end{remark}

Unfortunately, the transformation performed by Theorem~\ref{thm:origsimulation} does not preserve the weakness and additivity conditions (see~\cite[Remark~3.5]{Zanasi:Thesis:2012}), and therefore does not give us non-deterministic automata for the class $\AutWA(\ofoe)$.

In this section we provide, for every automaton $\aut\in\AutWA(\ofoe)$ and $p\in\props$, an automaton $\fwa{\aut}{p}\in\AutWA(\ofoe)$ which, although not being fully non-deterministic, is specially tailored to prove the closure of $\AutWA(\ofoe)$ under finite chain projection.
%
The state space of the construction $\fwa{\aut}{p}\in\AutWA(\ofoe)$ will be the disjoint union of two parts:
\begin{itemize}
	\itemsep 0 pt 
	\item A \emph{non-deterministic} part based on $\wp(A)$; and
	\item An \emph{alternating} part based on $A$.
\end{itemize}
The non-deterministic part will basically be a powerset construction of $\aut$, and contain the initial state of the automaton. It will have very nice properties enforced by construction:
\begin{itemize}
 	\itemsep 0 pt
 	\item It will behave non-deterministically,
 	\item The parity of every element will be $1$ (trivially satisfying the weakness condition); and 
 	\item The transition map of every element will be completely additive in $\wp(A)$.
\end{itemize} 
%
The alternating part will be a copy of $\aut$ modified such that it cannot be used to read nodes colored with the propositional variable $p$. The automaton $\fwa{\aut}{p}$ will therefore be based both on states from $A$ and on ``macro-states'' from $\wp(A)$. Moreover, the transition map of $\fwa{\aut}{p}$ will be defined such that once a match goes from the non-deterministic part to the alternating part, then it cannot come back (see Fig.~\ref{fig:twopart} for an illustration). Successful runs of $\fwa{\aut}{p}$ will have the property of processing only a \emph{finite} amount of the input being in a macro-state and all the rest behaving exactly as $\aut$ (but without reading any~$p$).

\input{fig-twopart.tex}

The key property of $\fwa{\aut}{p}$, which we will use to prove the closure under finite chain projection, is that for every tree $\tmodel$ and proposition $p\in\props$, the following holds:
\[
\tmodel \models \fwa{\aut}{p} \quad\text{iff}\quad \tmodel[p\resto X_p] \models \aut \text{ for some finite chain $X_p \subseteq \tmoddom$.}
\]

This finishes the intuitive explanations and we now turn to the necessary definitions to prove the results. The first step is to define a translation on the sentences associated with the
transition map of the original additive-weak automaton, which will aid us to define the transition map of the non-deterministic part. Henceforth, we use the notation $\shA$ to denote the set $\wp(A)$.

\begin{definition}\label{def:nbl}
Let $\alpha\in \ofoe^+(A,\sorts)$ be of the shape $\posdbnfofoe{\vlist{T}}{\Pi}_\aSort$ for some $\vlist{T} \in \wp(A)^k$ and $\Pi \subseteq \vlist{T}$. We say that $\alpha'\in\ofoe^+(\shA \cup A,\sorts)$ is a \emph{non-branching lifting} of $\alpha$ if
\begin{enumerate}[(i)]
	\itemsep 0 pt
	\item $\alpha' = \posdbnfofoe{\vlist{R}}{\Pi}_\aSort$ for some $\vlist{R} \in \wp(\shA \cup A)^k$,
	\item For every $i$, either:
			(a) $R_i = T_i$, or 
			(b) $T_i \neq \nada$ and $R_i = \{T_i\}$.
	\item Case~(ii.b) occurs at most once.
\end{enumerate}
%
Consider now $\psi\in \ofoe^+(A,\sorts)$ of the shape $\bigwedge_\aSort \alpha_\aSort$. We say that $\psi'\in\ofoe^+(\shA \cup A,\sorts)$ is a \emph{non-branching lifting} of $\psi$ if $\psi' = \bigwedge_\aSort \alpha'_\aSort$ and 
\begin{enumerate}[(i)]
	\itemsep 0 pt
	\item For every $\aSort \subseteq \sorts$, either:
			(a) $\alpha'_\aSort = \alpha_\aSort$, or 
			(b) $\alpha'_\aSort$ is a non-branching lifting of $\alpha_\aSort$.
	\item Case~(i.b) occurs at most once.
\end{enumerate}
Observe that every such $\alpha'$ is completely additive in $\shA$.
\end{definition}

\begin{definition}\label{def:bigpsi}
Let $\aut = \tup{A,\tmap,\pmap,a_I} \in \AutWA(\ofoe)$. Let $c \in \wp(\props)$ be a color and $Q \in \shA$ be a macro-state. First consider the formula
$
	\bigwedge_{a \in Q} \tmap(a,c). 
$
By Corollary~\ref{cor:ofoepositivenf} there is a formula $\Phi_{Q,c} \in \ofoe^+(A)$ such that $\Phi_{Q,c}\equiv\bigwedge_{a \in Q} \tmap(a,c)$ and $\Phi_{Q,c}$ is in the basic form $\bigvee_j \varphi_j$ where
$\varphi_j = \bigwedge_\aSort \posdbnfofoe{\vlist{T}}{\Pi}_\aSort$.
%
We define
\[
\Psi_{Q,c} := \bigvee_j\bigvee\{ \psi \mid \psi \text{ is a non-branching lifting of } \varphi_j\}.
\]
Observe that $\Psi_{Q,c} \in \ofoe^+(\shA\cup A)$ and $\Psi_{Q,c}$ is completely additive in $\shA$. The latter is because the non-branching liftings have this property, which is preserved by disjunction.
\end{definition}

\noindent
We are finally ready to define the two-part construct.

\begin{definition}\label{def:fc2}
Let $\aut = \tup{A,\tmap,\pmap,a_I}$ belong to $\AutWA(\ofoe,\oprops)$ and let $p$ be a propositional variable. We define the \emph{two-part construct of $\aut$ with respect to $p$} as the automaton $\fwa{\aut}{p} = \tup{A^{\fconst},\tmap^{\fconst},\pmap^{\fconst},a_I^{\fconst}} \in \AutWA(\ofoe,\oprops)$ given by:
\[
\begin{array}{lcl}
	A^{\fconst} &:=& A \cup \shA \\
	a_I^{\fconst} &:=& \{a_I\}\\ 
	\phantom{m}\\
	\pmap^{\fconst}(Q) &:=& 1 \\
	\pmap^{\fconst}(a) &:=& \pmap(a)
\end{array}
\hspace*{1cm}
\begin{array}{lcl}
	\tmap^{\fconst}(Q,c) &:=& \Psi_{Q,c}\\
	\phantom{m}\\
	\tmap^{\fconst}(a,c) &:=& 
	\begin{cases}
	\bot & \text{if $p\in c$,}\\
	\tmap(a,c) & \text{otherwise.}
	\end{cases}
\end{array}
\]
\end{definition}

The next definition introduces some notions of strategies that are closely related to our desiderata on the two-part construct.

\begin{definition}\label{def:StratfunctionalFinitary}
Given an automaton $\aut\in \Aut(\llang)$, a subset $B\subseteq A$ of the states of $\aut$, and a tree~$\tmodel$; a strategy $f$ for \eloise in $\agame(\aut,\tmodel)$ is called:
\begin{itemize}
	\itemsep 0pt
	\item \emph{Functional in $B$} (denoted `F in $B$', for short) if for each node $s\in\tmodel$ there is at most one $b \in B$ such that $(b,s)$ belongs to $\tmodel_{\!f}$.
	\item \emph{Non-branching in $B$} (denoted `NB in $B$', for short) if all the nodes of $\tmodel_{\!f}$ with a state from $B$ belong to the same branch of $\tmodel_{\!f}$.
	\item \emph{Well-founded in $B$} (denoted `WF in $B$', for short) if the set of nodes of $\tmodel_{\!f}$ with a state from $B$ are all contained in a well-founded subtree of $\tmodel_{\!f}$.
\end{itemize}
\end{definition}


\input{fig-strategies-2.tex}

Before proving that the two-part construct satisfies nice properties we need a few propositions. The following two lemmas show how to go from admissible moves in the two-part construct to the original automaton and vice-versa.

\begin{lemma}[functional to alternating]\label{lem:os-nd-to-alt}
	Given an automaton $\aut\in\AutWA(\ofoe)$, a macro-state $Q\in\shA$ and a color ${c\in \wp(\props)}$ such that ${(D,\val_{Q,c}) \models \Psi_{Q,c}}$ for some $\val_{Q,c}:\shA\cup A\to\wp(D)$; there is $\uval:A\to\wp(D)$ such that
		\begin{enumerate}[(i)]
			\itemsep 0 pt
			\item $(D,\uval) \models \tmap(a,c)$ for all $a\in Q$,
			\item If $d \in \uval(b)$ then
				$d \in \val_{Q,c}(b)$, or
				$d \in \val_{Q,c}(Q')$ for some $Q'\in \shA$ such that $b \in Q'$.
		\end{enumerate}
\end{lemma}
\begin{proof}
	Define $\uval:A\to\wp(D)$ as
	\[
	\uval(b) := V_{Q,c}(b) \cup \bigcup_{b\in Q'} V_{Q,c}(Q').
	\]
	Recall that
	\[
		\Psi_{Q,c} := \bigvee_j\bigvee\{ \psi \mid \psi \text{ is a non-branching lifting of } \varphi_j\}.
	\]
	%
	As a first step, let $(D,\val_{Q,c}) \models \psi$ where $\psi$ is a non-branching lifting of some $\varphi_j$.

	\begin{claimfirst}
		$(D,\uval) \models \bigwedge_{a\in Q}\tmap(a,c)$.
	\end{claimfirst}
	\begin{pfclaim}
		We show that $(D,\uval) \models \varphi_j$. This is enough because $\bigwedge_{a\in Q}\tmap(a,c) \equiv \bigvee_j \varphi_j$. The only interesting case is that of the predicates in $\shA$. Observe that in $\psi$ there can be at most one $Q'\in\shA$. If there are none then $\psi\equiv\varphi_j$ and we are done. Suppose that $Q'(x)$ occurs in $\psi$, then $\bigwedge\{b(x) \mid b\in Q'\}$ occurs in $\varphi_j$ in the same place. It is clear from the definition of $\val$ that if $d\in \val_{Q,c}(Q')$ then $d\in \uval(b)$ for every $b\in Q'$.
	\end{pfclaim}
	It is direct from the claim that $(D,\uval) \models \tmap(a,c)$ for all $a\in Q$.
\end{proof}

\begin{lemma}[alternating to functional]\label{lem:os-alt-to-nd}
	Let $\aut$ belong to $\AutWA(\ofoe)$, $Q\in\shA$ be a macro-state, and $c\in \wp(\props)$ be a color. Let $\{\val_{a,c}:A\to\wp(D) \mid a\in Q\}$ be a family of valuations such that ${(D,\val_{a,c}) \models \tmap(a,c)}$ for each $a\in Q$.
	Then, for every $P\subseteq D$ with $|P| \leq 1$ there is a valuation $\val_{Q,c}:A\cup\shA\to\wp(D)$ such that
		\begin{enumerate}[(i)]
			\itemsep 0 pt
			\item $(D,\val_{Q,c}) \models \Psi_{Q,c}$,
			\item If $d \in \val_{Q,c}(b)$ then
				$d \in \val_{a,c}(b)$ for some $a \in Q$.
			\item If $d \in \val_{Q,c}(Q')$ then
				$d \in \val_{a,c}(b)$ for some $a\in Q$, $b\in Q'$.
			\item For every $(a,s) \in Z_{\val_{Q,c}}$ we have that $a\in\shA$ iff $s\in P$.
		\end{enumerate}
\end{lemma}
\begin{proof}
	We use an auxiliary valuation $\val_t:A\to\wp(D)$ defined as $\val_t(b) := \bigcup_{a\in Q} \val_{a,c}(b)$.

	\begin{claimfirst}
		$(D,\val_t) \models \bigwedge_{a\in Q}\tmap(a,c)$.
	\end{claimfirst}
	\begin{pfclaim}
		Observe that for every $a\in Q$, $b\in A$ we have $V_{a,c}(b) \subseteq V_t(b)$ then by monotonicity we get that $(D,\val_t) \models \tmap(a,c)$ for every $a\in Q$.
	\end{pfclaim}
	Define the valuation $\val_{Q,c}:A\cup\shA\to\wp(D)$, using the alternative marking representation $\val^\natural_{Q,c}:D\to\wp(A\cup\shA)$, as follows:
	\[
	\val^\natural_{Q,c}(d) :=
	\begin{cases}
		\val^\natural_t(d) & \text{if $d\notin P$},\\
		\{\val^\natural_t(d)\} & \text{if $d\in P$}.
	\end{cases}
	\]
	and recall that $\bigwedge_{a\in Q}\tmap(a,c) \equiv \bigvee_i \varphi_i$ and
	\[
		\Psi_{Q,c} := \bigvee_j\bigvee\{ \psi \mid \psi \text{ is a non-branching lifting of } \varphi_j\}.
	\]
	Assume that $(D,\val_t) \models \varphi_j$, we show that $(D,\val_{Q,c}) \models \psi$ for some non-branching lifting of $\varphi_j$. If $P$ is empty then $\val_{Q,c} = V_t$ and as $\varphi_j$ is itself a non-branching lifting of $\varphi_j$ and $(D,\val_t) \models \varphi_j$, we can conclude that $(D,\val_{Q,c}) \models \varphi_j$ and we are done.

	If $P = \{d\}$ we proceed as follows: first recall that the shape of $\varphi_j$ is $\varphi_j = \bigwedge_\aSort \posdbnfofoe{\vlist{T}}{\Pi}_\aSort$. Let $\aSort_d\subseteq \sorts$ be the set of sorts to which $d$ belongs. We show that $(D,\val_{Q,c}) \models \psi_{\aSort_d}$ for some non-branching lifting $\psi_{\aSort_d}$ of $\posdbnfofoe{\vlist{T}}{\Pi}_{\aSort_d}$. This will be enough, since it is easy to show that in that case $\psi := \psi_{\aSort_d}\land \bigwedge_{\aSort\neq \aSort_d} \posdbnfofoe{\vlist{T}}{\Pi}_\aSort$ is a non-branching lifting of $\varphi_j$ (by Definition~\ref{def:nbl}) and $(D,\val_{Q,c}) \models \psi$.

	Our hypothesis is that $(D,\val_t) \models \posdbnfofoe{\vlist{T}}{\Pi}_{\aSort_d}$, which gives a full description of the elements of $D$ with sorts $\aSort_d$. Namely, if we restrict to the elements of sorts $\aSort_d$, then:
	\begin{itemize}
		\itemsep 0pt
		\item There are $d_1,\dots,d_k\in D$ such that $d_i$ has type $T_i$,
		\item Every $d'\in D$ which is not among $d_1,\dots,d_k$ satisfies some type in $\Pi$.
	\end{itemize}
	We consider the following two cases:
	\begin{enumerate}[(1)]
		\item Suppose that $d=d_i$ for some $i$; without loss of generality assume that $i=1$. In this case it is easy to see that $(D,\val_{Q,c}) \models \posdbnfofoe{\{T_1\}{\cdot}T_2\cdots T_k}{\Pi}_{\aSort_d}$, which is a non-branching lifting of $\posdbnfofoe{\vlist{T}}{\Pi}_{\aSort_d}$.
		\item Suppose that $d \neq d_i$ for all $i$. Then, $d$ must have some type $S_d\in \Pi$. The key observation is that $\Pi \subseteq \vlist{T}$. Hence, there is some $T_i$ such that $T_i = S_d$. Observe now that if we `switch' the elements $d$ and $d_i$ we end up in case (1).\qedhere
	\end{enumerate}
\end{proof}

\begin{remark}
	Observe that, without loss of generality, \eloise can always choose to play \emph{minimal} valuations. That is, in a basic position $(a,s)$ she plays a valuation $\val:A\to\wp(R[s])$ such that for every $t\in R[s]$ and $a\in A$, the element $t$ belongs to $\val(b)$ only if it is strictly needed to make $\tmap(a,\tscolors(s))$ true. That is, she plays valuations $\val$ such that
	\[
		\text{If}\quad (D,\val) \models \tmap(a,\tscolors(s)) \quad\text{then}\quad (D,\val[b \mapsto \val(b)\setminus\{t\}]) \not\models \tmap(a,\tscolors(s))
	\]
	for all $a,b\in A$ and $t\in \val(b)$.
	In what follows we assume that \eloise plays minimal valuations and we call such strategies minimal. For more detail we refer the reader to~\cite[Proposition~2.13]{Zanasi:Thesis:2012}.
\end{remark}

\noindent
Finally we can state and prove the properties of the two-part construct.

\begin{theorem}\label{thm:fc2}
Let $\aut\in\AutWA(\ofoe,\oprops)$, $p\in\props$ be a propositional variable, and $\tmodel$ be a $\oprops$-tree. The following holds:
\begin{enumerate}
  \itemsep 0 pt
  \item\label{point:fc2funConstrAut}
  $\fwa{\aut}{p}\in\AutWA(\ofoe,\oprops)$.
  \item\label{point:fc2funConstrStrategy}
  Every winning strategy for $\eloise$ in
  the game $\agame(\fwa{\aut}{p},\tmodel)@(a_I^{\fconst},s_I)$ can be assumed to be functional, non-branching and well-founded in $\shA$.
  %
  \item\label{point:fc2funConstrEquiv}
  $\fwa{\aut}{p}$ accepts $\tmodel$ iff 
  $\aut$ accepts $\tmodel[p\resto X_p]$ for some finite chain $X_p\subseteq \tmoddom$.
\end{enumerate}
\end{theorem}
\begin{proof}
\eqref{point:fc2funConstrAut} The key observation is that $\Psi_{Q,c}$ is completely additive in $\shA$.

\noindent
\eqref{point:fc2funConstrStrategy} We treat the properties separately:
\begin{itemize}
	\itemsep 0 pt
	\item \textit{Functional in $\shA$}: Suppose that $(a,s)$ is a position of an $f$-guided match where the proposed valuation $\val:A\to\wp(R[s])$ is such that $t \in \val(Q)$ and $t \in \val(Q')$ for distinct $Q,Q'\in\shA$ and some $t\in R[s]$. Let $\psi$ 
	be a disjunct of $\Psi_{Q,c}$ witnessing $(R_{\aact_1}[s],\dots,R_{\aact_n}[s],\val) \models \Psi_{Q,c}$. As $\psi$ is a non-branching lifting, the element $t$ has to be witness for exactly one type $T_i = \{Q''\}$ with $Q''\in\shA$. As we assume that \eloise plays minimal strategies then we can assume that $t\in\val(Q'')$ only, among $\shA$. Therefore, $t$ cannot be required to be a witness for both $Q$ and $Q'$ at the same time.
	\item \textit{Non-branching in $\shA$}: This is direct from the syntactical form of $\Psi_{Q,c}$. Observe that in each disjunct, at most one element of $\shA$ can occur. Assuming that $\eloise$ plays minimal strategies then she always proposes a valuation $\val$ where $\val(\shA)$ is a quasi-atom.
	\item \textit{Well-founded in $\shA$}: The game starts in $\shA$ and, as the parity of $\shA$ is $1$, it can only stay there for finitely many rounds. This means that, as $f$ is winning, every branch of $\tmodel_{\!f}$ has to leave $\shA$ at some finite stage. 
\end{itemize}

\noindent
\eqref{point:fc2funConstrEquiv}
\fbox{$\Leftarrow$}
Let $\tmodel' := \tmodel[p\resto X_p]$ and therefore $\tscolors' := \tscolors[p\resto X_p]$. Given a winning strategy $f$ for \eloise in $\game = \agame(\aut,\tmodel')@(a_I,s_I)$ we construct a winning strategy $f^\fconst$ for \eloise in $\game^\fconst = \agame(\fwa{\aut}{p},\tmodel)@(a_I^{\fconst},s_I)$. We define it inductively for a match $\match^\fconst$ of $\game^\fconst$. While playing $\match^\fconst$ we maintain a bundle (set) $\matches$ of $f$-guided shadow matches. We use $\matches_i$ to denote the bundle at round $i$. We maintain the following condition ($\ddag$) for every round along the play:
\begin{enumerate}[$\ddag$1.]
  \item
  If the current basic position in $\match^\fconst$ is of the form $(Q,s) \in \shA \times \tmoddom$, then (a) for every $a\in Q$ there is an $f$-guided shadow match $\match_a \in \matches$ such that the current basic position is $(a,s) \in A\times \tmoddom$; moreover, (b) $\tmodel'.s$ is not $p$-free.
  \item
  Otherwise, (a) $\matches = \{\match\}$ and the position in both $\match^\fconst$ and $\match$ is of the form $(a,s) \in A \times \tmoddom$; and moreover, (b) $\tmodel'.s$ is $p$-free.
\end{enumerate}
Intuitively, in order to simulate $\aut$ with $\fwa{\aut}{p}$ we have to keep two things in mind: (1) $\fwa{\aut}{p}$ can only read $p$ while it is in the non-deterministic part; and (2) every choice of \abelard in $\game$ corresponds to a match that has to be won by \eloise; these parallel matches are kept track of in $\matches$ and represented as a macro state in $\fwa{\aut}{p}$. Therefore, we want the simulation to stay in the non-deterministic part while we could potentially read some $p$ in $\tmodel'$ (condition $\ddag$1). Whenever there are no more $p$'s to be read, we can relax and behave exactly as $\aut$ (condition $\ddag$2).

We only consider the case where $\tmodel'.s_I$ is not $p$-free, and hence, $X_p$ is non-empty. Otherwise it can be easily seen that \eloise can win $\game^\fconst$ using the same strategy $f$. Assume then that $\tmodel'.s_I$ is not $p$-free. At round $0$ we initialize the bundle $\matches = \{\match_{a_I}\}$ with the $f$-guided match $\match_{a_I}$ at basic position $(a_I,s_I)$. It is clear that~($\ddag$1) holds. For the inductive step 
we divide in cases:

\begin{itemize}
	\item If ($\ddag$2) holds we are given a bundle $\matches=\{\match\}$ such that both $\match^\fconst$ and $\match$ are in position $(a,s) \in A\times \tmoddom$. We define $f^\fconst$ as $f$ for this position. To see that this is an admissible move in $\match^\fconst$ observe that $\tmodel'.s$ is $p$-free and therefore $\tmap^\fconst(a,\tscolors(s)) = \tmap(a,\tscolors'(s))$. Now it is \abelard's turn to make a move in $\match^\fconst$. By definition of $\tmap^\fconst$, the formula $\tmap^\fconst(a,\tscolors'(s))$ belongs to $\ofoe^+(A)$ and hence the next position in $\match^\fconst$ will be of the form $(a',s')\in A \times \tmoddom$ with $\tmodel'.s'$ $p$-free. We replicate the move in the shadow match $\match$ and hence ($\ddag$2) is preserved.
	\item If ($\ddag$1) holds at round $i$ of $\match^\fconst$ we are given $f$-guided matches $\matches_i = \{\match_{a_1},\dots,\match_{a_k}\}$ such that for the current position $(Q,s) \in \shA\times \tmoddom$ and for each $a\in Q$ we have $\match_a \in \matches_i$. 
	For every match $\match_a$, the strategy $f$ provides a valuation $\val_a$ which is an admissible move in this match. Define $P := \{t \in R[s] \mid \tmodel'.t \text{ is not $p$-free}\}$ and observe that, as $X_p$ is a chain, $P$ will be either a singleton or empty. Using Lemma~\ref{lem:os-alt-to-nd} with $P$ and $\{V_a\}_{a\in Q}$ we can combine these valuations into an admissible move $\val^\fconst$ in $\match^\fconst$.

	To prove that~($\ddag$) is preserved we distinguish cases as to \abelard's move: first suppose that \abelard chooses a position of the form $(b,t) \in A\times\tmoddom$. Because of Lemma~\ref{lem:os-alt-to-nd}(ii) we know that $t \in \val_{a_i}(b)$ for some $a_j \in Q$. That is, we can replicate this move in one of the shadow matches $\match_{a_j}$. We do that and set $\matches=\{\match_{a_j}\}$ hence validating ($\ddag$2a). To see that ($\ddag$2b) is also satisfied observe that Lemma~\ref{lem:os-alt-to-nd}(iv) implies that $\tmodel'.t$ is $p$-free.

	For the other case, suppose that \abelard chooses a position of the form $(Q',t) \in \shA\times\tmoddom$. Similar to the last case, this time using Lemma~\ref{lem:os-alt-to-nd}(iii), we can trace every $b\in Q'$ back to some match $\match_{b} \in \matches_i$. We define $\match_{b}{\cdot}b$ as the match $\match_{b}$ extended with \abelard's move $(b,t)$. Finally we let $\matches_{i+1}:=\{\match_{b}{\cdot}b \mid b\in Q'\}$, which validates ($\ddag$1a). To see that ($\ddag$1b) is also satisfied observe that Lemma~\ref{lem:os-alt-to-nd}(iv) implies that $\tmodel'.t$ is not $p$-free.
\end{itemize}
Now we prove that $f^\fconst$ is actually winning. It is clear that \eloise wins every finite full $f^\fconst$-guided match (because the moves are admissible). Now suppose that an $f^\fconst$-guided match is infinite. By hypothesis the extension of $p$ in $\tmodel'$ is a finite chain, so after a finite amount of rounds we arrive to an element $s$ such that $\tmodel'.s$ is $p$-free. This means --because of ($\ddag$)-- that the automaton stays in $\shA$ only for a finite amount of steps and then moves to $A$, at a position $(a,s)$ which is \emph{winning} for \eloise. From there on the match $\match^\fconst$ and $\match$ are exactly the same and, as \eloise wins $\match$ (which is $f$-guided for a winning strategy $f$), she also wins $\match^\fconst$.

\input{fig-xp.tex}

\medskip
\noindent
\fbox{$\Rightarrow$}
Given a winning strategy $f^\fconst$ for \eloise in $\game^\fconst = \agame(\fwa{\aut}{p},\tmodel)@(a_I^{\fconst},s_I)$ we construct a winning strategy $f$ for \eloise in $\game = \agame(\aut,\tmodel')@(a_I,s_I)$ where $\tmodel' := \tmodel[p\resto X_p]$ for some finite chain $X_p\subseteq \tmoddom$, which we promptly define.
Using Theorem~\ref{thm:fc2}(\ref{point:fc2funConstrStrategy}) we assume that $f^\fconst$ is functional, non-branching and well-founded in $\shA$ and define the set $X_p$ as follows:
\[
	X_p := \{s\in\tmoddom \mid (Q,s) \in \tmodel_{\!f^\fconst} \text{ for some $Q\in\shA$}\}.
\]
The fact that $f^\fconst$ is non-branching (in $\shA$) makes $X_p$ a chain, and well-foundedness makes it finite. As an illustration, Fig.~\ref{fig:xp} represents a possible tree $\tmodel_{\!f^\fconst}$ where the path induced by positions of the form $(Q,s)$ is drawn with a thicker stoke.

Next, we define the strategy inductively for a match $\match$ of $\game$. While playing $\match$ we maintain an $f^\fconst$-guided shadow match $\match^\fconst$. We maintain the following condition ($\ddag$) for every round along the play: let $(a,s) \in A \times \tmoddom$ be the current position in $\match$, then one of the following conditions holds:
\begin{enumerate}[$\ddag$1.]
	\itemsep 0pt
	\item The current basic position in $\match^\fconst$ is of the form $(Q,s) \in \shA \times \tmoddom$ with $a\in Q$,
	\item The current basic position in $\match^\fconst$ is also $(a,s) \in A \times \tmoddom$. 
\end{enumerate}
At round $0$ the matches $\match$ and $\match^\fconst$ are in position $(a_I,s_I)$ and $(\{a_I\},s_I)$ respectively, therefore ($\ddag$1) holds. For the inductive step we divide in cases:

\begin{itemize}
	\item If ($\ddag$2) holds, the match $\match$ is in position $(a,s)$. For this position, we let $f$ be defined as $f^\fconst$. Observe that it must be the case that $p\notin\tscolors(s)$, otherwise \eloise wouldn't have an admissible move $\val^\fconst$ in $\match^\fconst$. Given this, and assuming that \eloise plays minimal strategies, \eloise can use the same $\val^\fconst$ in $\match$. It is easy to see that we can replicate \abelard's next move in the shadow match.
	\item If ($\ddag$1) holds, the matches $\match$ and $\match^\fconst$ are respectively in position $(a,s)$ and $(Q,s)$ with $a\in Q$. The strategy $f^\fconst$ provides a valuation $\val^\fconst$ which is admissible in $\match^\fconst$. Using Lemma~\ref{lem:os-nd-to-alt} we can get a valuation $\uval$ which is admissible in $\match$ --see item (i). Suppose now that \abelard chooses $(b,t)$ as a next position in $\match$. Using Lemma~\ref{lem:os-nd-to-alt}(ii) we know that either (a) $t\in \val^\fconst(b)$ or, (b) there is some $Q' \in \shA$ with $b\in Q'$ and $t \in \val^\fconst(Q')$. In both cases we have a way to replicate \abelard's move in $\match^\fconst$ and preserve ($\ddag$).
\end{itemize}
To see that $f$ is winning we proceed similar to the other direction.
\end{proof}

\paragraph{Historical remarks and related results.}
The idea of a Simulation Theorem goes back to (at least) Safra~\cite{Safra:1988} and Muller and Schupp~\cite{MullerSimulation}. In the first case, Safra used an augmented state space to convert non-deterministic B\"uchi automata into deterministic automata. In the latter, Muller and Schupp also use an augmented state space to convert alternating tree automata to non-deterministic tree automata.

The idea to use a two-part automata to preserve the weakness condition was introduced in~\cite{Zanasi:Thesis:2012,DBLP:conf/lics/FacchiniVZ13}, although the authors claim that some concepts were already present in~\cite{MullerSaoudiSchupp92}. In~\cite{Zanasi:Thesis:2012,DBLP:conf/lics/FacchiniVZ13} the authors use an automaton based on $\wp(A\times A)$ and $A$ with a non-parity acceptance condition. This automaton is then converted to a parity automaton with a standard trick. Using $\wp(A\times A)$ as the state space instead of $\wp(A)$ is necessary to correctly keep track of \emph{infinite} runs of the automata. The first explicit use of $\wp(A\times A)$ as the state space of such automata seems to be in~\cite[Section~9.6.2]{ArnoldN01}.

Observe, however, that in the non-deterministic part of our constructions the parity is uniformly $1$ and therefore therefore any infinite run which stays in that part will be a rejecting run. Using this observation, we give a slightly simpler construction based on $\wp(A)$ and $A$. In this respect, the proofs are cleaner and we avoid a non-parity acceptance condition.

The second critical element of this section is the enforcing of the additivity condition on the non-deterministic component. The core of this idea was developed in~\cite{LICS14,DBLP:journals/corr/CarreiroFVZ14} where it is applied for another notion called ``continuity.''

%% file: fig-strategies-1.tex

\begin{figure}[h]
\centering
\begin{tikzpicture}
\begin{scope}[
->,>=latex,
level/.style={level distance=1.5cm, sibling distance=2cm/#1}
]
\node {$s_I$}
    child {
        node {$s_0$}
            child {
                node {$s_{00}$}
                edge from parent
            }
            child {
                node {$s_{01}$}
                edge from parent
            }
            edge from parent 
    }
    child {
        node {$s_{1}$}
        child {
                node {$s_{10}$}
                edge from parent
            }
            child {
                node {$s_{11}$}
                edge from parent
            }
        edge from parent         
    };
\end{scope}

\begin{scope}[
xshift=7cm,
level/.style={level distance=1.5cm, sibling distance=3cm/#1},
level 1/.style={sibling distance=2cm}
]
\node {$(a_I,s_I)$}
    child {
        node {$(a_3,\underline{s_0})$}
            child {
                node {$(a_5,\underline{s_{00}})$}
                edge from parent
            }
            child {
                node {$(a_1,\underline{s_{00}})$}
                edge from parent
            }
            edge from parent 
    }
    child {
        node {$(a_7,\underline{s_0})$}
        edge from parent 
    }
    child {
        node {$(a_1,s_{1})$}
            child {
                node {$(a_1,\underline{s_{10}})$}
                edge from parent
            }
            child {
                node {$(a_3,\underline{s_{10}})$}
                edge from parent
            }
            child {
                node {$(a_1,s_{11})$}
                edge from parent
            }
        edge from parent         
    };
\end{scope}
\end{tikzpicture}
\caption{A tree $\tmodel$ and $\tmodel_{\!f}$ for a fixed strategy for \eloise.}
\label{fig:strategies1}
\end{figure}

%% file: fig-twopart.tex

\begin{figure}[h]
\centering
\begin{tikzpicture}[node distance=1.3cm,>=stealth',bend angle=45,auto]
  \tikzstyle{node}=[circle,fill=black,minimum size=5pt,inner sep=1pt]
  \tikzstyle{place}=[circle,thick,draw=blue!75,fill=blue!20,minimum size=6mm]
  \tikzstyle{iplace}=[place,draw=red!75,fill=red!20]
  \tikzstyle{inter}=[dashed,draw=black!50,->]

  \begin{scope}
    \node [node] (w1)                                    {};
    \node [node] (c1) [below of=w1]                      {};
    \node [node] (s)  [below of=c1] {};
    \node [node] (c2) [below of=s]                       {};
    \node [node] (w2) [below of=c2]                      {};

    \node [node] (e1) [left of=c1] {}
      edge [pre,bend left]                  (w1)
      edge [post,bend right]                (s)
      edge [post]                           (c1);

    \node [node] (e2) [left of=c2] {}
      edge [pre,bend right]                 (w2)
      edge [post,bend left]                 (s)
      edge [post]                           (c2);

    \node [node] (l1) [right of=c1] {}
      edge [pre]                            (c1)
      edge [pre,bend left]                  (s)
      edge [post,bend right] (w1);

    \node [node] (l2) [right of=c2] {}
      edge [pre]                            (c2)
      edge [pre,bend right]                 (s)
      edge [post,bend left]         (w2);
  \end{scope}

  \begin{scope}[xshift=-6cm]
    \node [iplace,tokens=1] (w1')   {}
      edge [inter,bend angle=20, bend left] (e1);

    \node [place,tokens=2] (c1') [below of=w1'] {}
    	edge [inter,bend angle=30, bend right] (l1);

    \node [place,tokens=2] (s1') [below of=c1',xshift=-5mm]      {};
    \node [place,tokens=3]
                      (s2') [below of=c1',xshift=5mm] {};
    \node [place,tokens=4]     (c2') [below of=s1',xshift=5mm]                      {}
    	edge [inter,bend angle=20, bend right] (l2);

    \node [place,tokens=3]
                      (w2') [below of=c2']                                 {};

    \node [place,tokens=3] (e1') [left of=c1'] {}
      edge [pre,bend left]                  (w1')
      edge [post]                           (s1')
      edge [pre]                            (s2')
      edge [post]                           (c1')
      edge [inter,bend angle=20, bend right] (c1);

    \node [place,tokens=3] (e2') [left of=c2'] {}
      edge [pre,bend right]                 (w2')
      edge [post]                           (s1')
      edge [pre]                            (s2')
      edge [post]                           (c2')
      edge [inter,bend angle=20, bend right] (c2);

    \node [place,tokens=3] (l1') [right of=c1'] {}
      edge [pre]                            (c1')
      edge [pre]                            (s1')
      edge [post]                           (s2')
      edge [post,bend right]  (w1');

    \node [place,tokens=3] (l2') [right of=c2'] {}
      edge [pre]                            (c2')
      edge [pre]                            (s1')
      edge [post]                           (s2')
      edge [post,bend left]        (w2')
      edge [inter,bend angle=20, bend right] (w2);
  \end{scope}

  \begin{pgfonlayer}{background}
    \filldraw [line width=4mm,join=round,black!10]
      (w1'.north  -| l1.east) rectangle (w2'.south  -| e1.west)
      (w1'.north -| l1'.east) rectangle (w2'.south -| e1'.west);
  \end{pgfonlayer}
\end{tikzpicture}
\caption{Two-part construction, initial state in red (illustrative)}
\label{fig:twopart}
\end{figure}
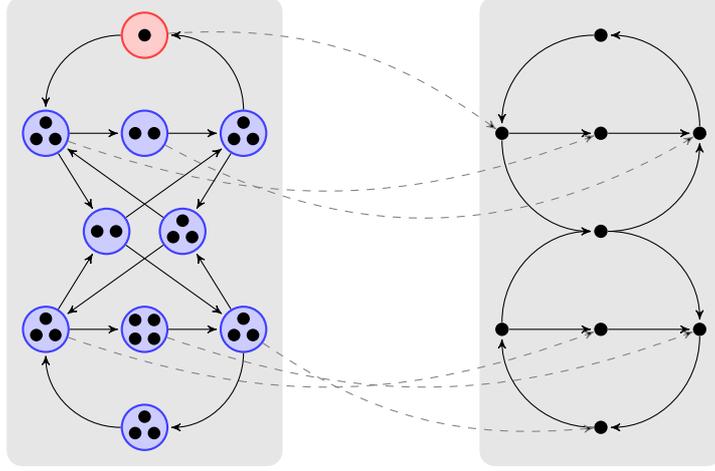

%% file: fig-strategies-2.tex

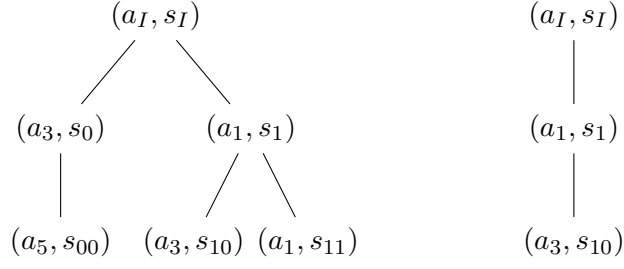
\begin{figure}
\centering
\begin{tikzpicture}
\begin{scope}[
level/.style={level distance=1.5cm, sibling distance=3cm/#1},
level 1/.style={sibling distance=2.5cm}
]
\node {$(a_I,s_I)$}
    child {
        node {$(a_3,s_0)$}
            child {
                node {$(a_5,s_{00})$}
                edge from parent
            }
            edge from parent 
    }
    child {
        node {$(a_1,s_{1})$}
            child {
                node {$(a_3,s_{10})$}
                edge from parent
            }
            child {
                node {$(a_1,s_{11})$}
                edge from parent
            }
        edge from parent         
    };
\end{scope}


\begin{scope}[
xshift=5.5cm,
level/.style={level distance=1.5cm, sibling distance=3cm/#1},
level 1/.style={sibling distance=2.5cm}
]
\node {$(a_I,s_I)$}
    child {
        node {$(a_1,s_{1})$}
            child {
                node {$(a_3,s_{10})$}
                edge from parent
            }
        edge from parent         
    };
\end{scope}
\end{tikzpicture}
\caption{Functional and functional+non-branching strategies.}
\end{figure}

%% file: fig-xp.tex

\begin{figure}[h]
\centering
\begin{tikzpicture}[
level/.style={sibling distance=70mm/#1},
level 1/.style={sibling distance=40mm},
level 2/.style={sibling distance=25mm},
level 3/.style={sibling distance=20mm},
every node/.style={draw=black!50,circle,fill=black!50,inner sep=1pt},
every edge/.style={draw=black!50},
aedge/.style={draw=black!50,thick},
ndnode/.style={draw=blue!75,circle,fill=blue!75,inner sep=2pt},
ndedge/.style={draw=blue!75,very thick},
end/.style={isosceles triangle,fill=black!20,rotate=+90,minimum width=1cm,line width=0pt}
]
\node [ndnode] {}
  child[ndedge] {
    node[ndnode,fill=yellow] {}
    child {
      node[ndnode] {}
      child[aedge] {
        node {}
      	child {node[end] {}}
      } 
      child {
        node[ndnode,fill=yellow] {}
        child[draw=red!75] {
            node[end,fill=red!50,draw=red!75] {}
        }
      }
    }
    child[aedge] {
      node[fill=yellow] {}
      child {
        node {}
          child {
            node[end] {}
          }
      }
    }
  }
	child[aedge] {
	  node {}
	  child {
      node {}
		  child {node[end] {}}
	  } 
	  child {
      node[fill=yellow] {}
	  	child {node[end] {}}
	  }
	};
\end{tikzpicture}
\caption{F+NB mode (blue), alternating mode (red) and $p$ (yellow).}
\label{fig:xp}
\end{figure}
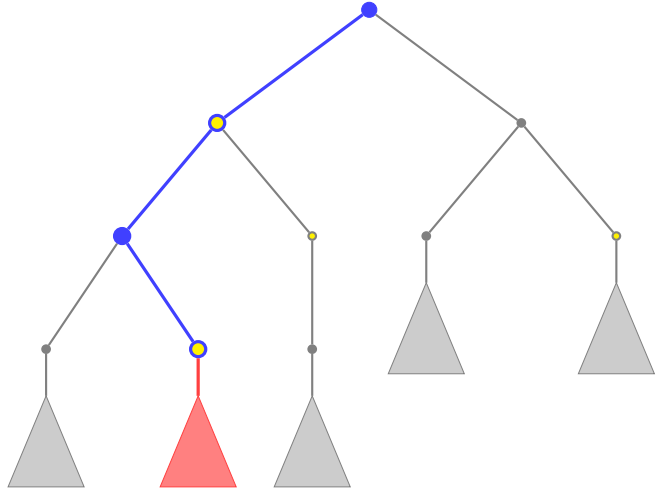

%% file: soaut-clintro.tex

Given an automaton $\aut\in\Aut(\llang,\props)$, we define the \emph{tree language recognized by $\aut$} as the set of $\props$-labelled trees $\trees(\aut)$ given by
\[
	\trees(\aut) := \{\tmodel \mid \aut \text{ accepts } \tmodel\}.
\]

In this subsection we prove that the collection of tree languages recognized by the automata of $\AutWA(\ofoe)$ is closed under the operations corresponding to the connectives of $\wcl$, that is: union, complementation and projection with respect to finite chains. We start with the latter.

%% file: soaut-projection.tex

In the following definition we give, for every automaton $\aut\in\AutWA(\ofoe,\oprops \dcup \{p\})$, the weak chain projection over $p$, denoted ${\existswc p}.\aut$ and belonging to $\AutWA(\ofoe,\oprops)$. The domain and transition function of the automaton ${\existswc p}.\aut$ will be based on $\fwa{\aut}{p}$.

\begin{definition}\label{def:autproj}
Let $\aut$ belong to $\AutWA(\ofoe,\oprops \dcup \{p\})$. 
We define the \emph{chain projection of $\aut$ over $p$} as the automaton ${\existswc p}.\aut := \tup{A\cup\shA, \tmapProj, \pmapProj, \{a_I\}}\in\AutWA(\ofoe,\oprops)$ given as follows, for every $c\in \wp(\oprops)$:
%
\[
\begin{array}{lcl}
  \pmapProj(a) &:=& \pmap(a) \\
  \pmapProj(Q) &:=& 1
\end{array}
\hspace*{1cm}
\begin{array}{lcl}
  \tmapProj(a,c) &:=& \tmap(a,c) \\
  \tmapProj(Q,c) &:=& \Psi_{Q,c} \lor \Psi_{Q,c\cup\{p\}}
\end{array}
\]
where $\Psi_{Q,c}$ is as in Definition~\ref{def:bigpsi}.
\end{definition}

The key observation to be made about the above definition is that ${\existswc p}.\aut$ is actually defined based on the two-part construction $\fwa{\aut}{p}$ (see Definition~\ref{def:fc2}). The main change is that the non-deterministic part ($\shA$) has been projected with respect to $p$. This can be observed in the definition of $\tmapProj(Q,c)$.

\begin{definition} Let $p\notin\oprops$ and $L$ be a tree language of $(\oprops\dcup\{p\})$-labeled trees. The \emph{finite chain projection} of $L$ over $p$ is the language $\existswc p.L$ of $\oprops$-labeled trees defined as
\[
  \existswc p.L := \{\tmodel \mid \tmodel[p\mapsto X_p] \in L \text{ for some finite chain } X_p\subseteq\tmoddom \}.
\]
\end{definition}

\begin{lemma}\label{lem:projection}
For each $\aut\in\AutWA(\ofoe,\oprops \dcup \{p\})$
we have 
  $\trees({\existswc p}.\aut) = {\existswc p}.\trees(\aut)$.
\end{lemma}

\begin{proof}
What we need to show is that for any tree $\tmodel$ over $\wp (\oprops)$:
\begin{eqnarray*}
  {\existswc p}.\aut \text{ accepts } \tmodel & \text{ iff }& \text{there is a finite chain $Y_p \subseteq \tmoddom$ } \\
   & & \text{such that } \aut \text{ accepts }\tmodel[p\mapsto Y_p].
\end{eqnarray*}
However, we will show that the following statement ($\S$) holds:
\begin{eqnarray*}
  {\existswc p}.\aut \text{ accepts } \tmodel & \text{ iff }& \text{there is a finite chain $X_p \subseteq \tmoddom$ }\label{proj:2} \\
   & & \text{such that } \fwa{\aut}{p} \text{ accepts }\tmodel[p\mapsto X_p]
\end{eqnarray*}
and then recover the result we want using the following claim.
\begin{claimfirst}
  There is a finite chain $Y_p \subseteq \tmoddom$ such that $\aut$ accepts $\tmodel[p\mapsto Y_p]$ iff there is a finite chain $X_p \subseteq \tmoddom$ such that $\fwa{\aut}{p}$ accepts $\tmodel[p\mapsto X_p]$.
\end{claimfirst}
\begin{pfclaim}
  Recall that Theorem~\ref{thm:fc2}(\ref{point:fc2funConstrEquiv}) states that
  \[
  \aut \text{ accepts } \tmodel[p\resto Z] \text{ for some finite chain } Z\subseteq \tmoddom
  \quad\text{iff}\quad
  \fwa{\aut}{p} \text{ accepts } \tmodel.\tag{$*$}\label{eq:recall:th3}
  \]
  \fbox{$\Rightarrow$}
  Suppose that there is a finite chain $Y_p \subseteq \tmoddom$ such that $\aut$ accepts $\tmodel[p\mapsto Y_p]$ and define $\tmodel' := \tmodel[p\mapsto Y_p]$. Observe now that $\tmodel'[p\resto Y_p] = \tmodel'$. Using the left-to-right direction of~\eqref{eq:recall:th3} with $\tmodel = \tmodel'$ and $Z = Y_p$, we get that $\fwa{\aut}{p}$ accepts $\tmodel' = \tmodel[p\mapsto Y_p]$.
  \fbox{$\Leftarrow$} 
  Suppose that there is a finite chain $X_p \subseteq \tmoddom$ such that $\fwa{\aut}{p}$ accepts $\tmodel[p\mapsto X_p]$ and define $\tmodel' := \tmodel[p\mapsto X_p]$. Using the right-to-left direction of~\eqref{eq:recall:th3} with $\tmodel = \tmodel'$, we get that $\aut$ accepts $\tmodel'[p\resto Z]$ for some finite chain $Z\subseteq\tmoddom$. That is, there exists a finite chain $Z\subseteq\tmoddom$ such that $\aut$ accepts $\tmodel[p\mapsto X_p \cap Z]$. As the intersection of two finite chains is again a finite chain, we can take $Y_p := X_p \cap Z$ and conclude that there exists a finite chain $Y_p$ such that $\aut$ accepts $\tmodel[p\mapsto Y_p]$.
\end{pfclaim}

\noindent
We now turn to the proof of~($\S$).

\medskip\noindent
\fbox{$\Rightarrow$}
It is not difficult to prove that properties~(\ref{point:fc2funConstrAut},\ref{point:fc2funConstrStrategy}) in Theorem~\ref{thm:fc2} hold for ${\existswc p}.\aut$ as well, since the latter is defined in terms of $\fwa{\aut}{p}$. Therefore we can assume that the given winning strategy $f_\exists$ for \eloise in $\game_\exists = \agame({\existswc p}.\aut,\tmodel)@(a_I^{\fconst},s_I)$ is functional, non-branching and well-founded in~$\shA$. Functionality allows us to associate with each node $s$ either none or a unique state $Q_s \in \shA$ (\emph{cf.} \cite[Prop.~3.12]{Zanasi:Thesis:2012}). We now want to isolate the nodes that $f_\exists$ treats ``as if they were labeled with~$p$''. For this purpose, let $\val_{s}$ be the valuation suggested by $f_\exists$ at a position $(Q_s,s) \in \shA \times \tmoddom$. As $f_\exists$ is winning, $\val_{s}$ makes $\tmapProj(Q_s,\tscolors(s)) = \Psi_{\tscolors(s)} \lor \Psi_{\tscolors(s)\cup\{p\}}$ true in $R[s]$. We define
\[
  X_p\ :=\ \{s \in \tmoddom\mid Q_s \text{ is defined and } (R[s],\val_{s}) \models \Psi_{\tscolors(s)\cup\{p\}}\}.
\]
The fact that $f_\exists$ is functional in $\shA$ guarantees that $X_p$ is well-defined; as the strategy is non-branching in $\shA$ we get that $X_p$ is a chain; finally, well-foundedness makes it finite. Let ${\tmodel' := \tmodel[p\mapsto X_p]}$, we show that we can give a winning strategy $f_\div$ for \eloise in the game $\game_\div = \agame(\fwa{\aut}{p},\tmodel')@(a_I^{\fconst},s_I)$. Actually, we show that $f_\div :=f_\exists$ works, we do it by induction for a match $\match_\div$ of $\game_\div$. We keep a shadow match $\match_\exists$ in $\game_\exists$ such that the following condition holds at each round:
\begin{equation}
\text{Both matches $\match_\div$ and $\match_\exists$ are in the same position $(q,s)\in A\cup\shA\times \tmoddom$.}\tag{\ddag}
\end{equation}
This condition obviously holds at the beginning of the games. For the inductive step let $\tscolors' = \tscolors[p\mapsto X_p]$ be an abbreviation for the coloring of $\tmodel'$ and consider the following cases:
\begin{itemize}
  \item If the current basic position in $\match_\div$ is of the form $(a,s) \in A\times \tmoddom$: by definition of $X_p$ we know that $s\notin X_p$, so $p\notin \tscolors'(s)$ and hence $\tscolors'(s) = \tscolors(s)$. As $f_\exists$ is winning in $\game_\exists$ we know that the suggested valuation $\val_{a,s}$ is admissible in $\match_\exists$, that is, $(R[s],\val_{a,s}) \models \tmap(a,\tscolors(s))$. As $\tscolors'(s) = \tscolors(s)$, we can conclude that $(R[s],\val_{a,s}) \models \tmap(a,\tscolors'(s))$ and thus is also an admissible move in $\match_\div$. 
  \item If the current basic position in $\match_\div$ is of the form $(Q,s) \in \shA\times \tmoddom$ we let $\val_{Q,s}$ be the valuation suggested by $f_\exists$ and consider the following cases:
    \begin{enumerate}
      \item If $p \in \tscolors'(s)$: then by definition of $X_p$ we have that
      $(R[s],\val_{Q,s}) \models \Psi_{\tscolors(s)\cup\{p\}}$. As $\tscolors'(s) = \tscolors(s) \dcup \{p\}$ we have that $(R[s],\val_{Q,s}) \models \Psi_{\tscolors'(s)}$. This is, by definition of $\aut_\div$, equivalent to $(R[s],\val_{Q,s}) \models \tmap^F(Q,\tscolors'(s))$ and therefore $\val_{Q,s}$ is admissible in $\match_\div$.
      \item If $p \notin \tscolors'(s)$: then $(R[s],\val_{Q,s}) \models \Psi_{\tscolors(s)} \lor \Psi_{\tscolors(s)\cup\{p\}}$ but $(R[s],\val_{Q,s}) \not\models \Psi_{\tscolors(s)\cup\{p\}}$ hence it must be the case that $(R[s],\val_{Q,s})\models \Psi_{\tscolors(s)}$. As $\tscolors'(s) = \tscolors(s)$, then $(R[s],\val_{Q,s})\models \Psi_{\tscolors'(s)} = \tmap^F(Q,\tscolors'(s))$ and therefore $\val_{Q,s}$ is admissible in $\match_\div$.
    \end{enumerate}
\end{itemize}
As the move by \eloise is the same in both matches it is clear that we can mimic in the shadow match $\match_\exists$ the choice of \abelard in $\match_\div$, therefore preserving ($\ddag$).

It is only left to show that this strategy is winning for \eloise. It is enough to observe that $\match_\div$ and $\match_\exists$ go through the same basic positions and, as \eloise wins $\match_\eloise$, she also wins $\match_\div$.

\medskip\noindent
\fbox{$\Leftarrow$}
Given a winning strategy $f_\div$ for \eloise in $\game_\div:=\agame(\fwa{\aut}{p},\tmodel')@(a_I^{\fconst},s_I)$ it is not difficult to see that the same strategy is winning for \eloise in $\game_\exists:=\agame({\existswc p}.\aut,\tmodel)@(a_I^{\fconst},s_I)$. As before, we can maintain the following invariant between a match $\match_\exists$ of $\game_\exists$ and a shadow match $\match_\div$ of $\game_\div$:
\begin{equation*}
\text{The matches $\match_\div$ and $\match_\exists$ are in the same position $(q,s)\in A\cup\shA\times \tmoddom$.}
\end{equation*}
The key observation in this case is that whenever the match $\match_\div$ is in a position $(a,s)$ then $p \notin \tscolors'(s)$. This is because $\tmap^\fconst(a,c) = \bot$ if $p\in c$ and that would contradict that $f_\div$ is winning. As a consequence, $\tmap^\exists(a,\tscolors(s)) = \tmap^\fconst(a,\tscolors'(s))$ and therefore the move suggested by $f_\div$ in $\game_\div$ will also be admissible in $\game_\exists$.
\end{proof}

%% file: soaut-booleans.tex

In this section we show that the class of tree languages recognized by the automata of $\AutWA(\ofoe)$ is closed under the Boolean operations.
Starting with the closure under union, we just mention the following result, without
providing the (completely routine) proof.

\begin{theorem}
\label{t:cl-dis}
Let $\aut$ and $\aut'$ belong to $\AutWA(\ofoe)$.
There is $\mathbb{U}\in\AutWA(\ofoe)$ such that $\trees(\mathbb{U}) = \trees(\aut) \cup \trees(\aut')$.
\end{theorem}

In order to prove closure under complementation, we crucially use that the 
one-step language $\ofoe$ is closed under Boolean duals 
(cf.~Proposition~\ref{prop:duals}).

\begin{theorem}
\label{t:cl-cmp}
If $\aut$ belongs to $\AutWA(\ofoe)$ then the automaton $\dual{\aut}$ defined in Definition~\ref{d:caut} also belongs to $\AutWA(\ofoe)$ and recognizes the complement of $\trees(\aut)$.
\end{theorem}

\begin{proof}
Since we already know that $\dual{\aut}$ accepts exactly the transition
systems that are rejected by $\aut$, we only need to check that 
$\dual{\aut}$ indeed belongs to $\AutWA(\ofoe)$.
But this is straightforward: for instance, the additivity and multiplicativity
constraints can be checked by observing the dual nature of these properties,
as shown in Proposition~\ref{prop:adddualmult}.
\end{proof}

%% file: soaut-logic-intro.tex

In this section we show that, on trees, the following formalisms are equivalent:
\begin{enumerate}[(i)]
	\itemsep 0 pt
	\item $\wcl$: Weak chain logic,
	\item $\AutWA(\ofoe)$: Additive-weak automata based on $\ofoe$,
	\item $\mucaffoe$: Forward-looking fragment of $\mucafoe$,
	\item $\mucafoe$: Completely additive restriction of $\mufoe$.
\end{enumerate}
Together with the fact that $\mucafoe \equiv \binfotc$ which was proved in Theorem~\ref{thm:fotcmucafoe}, this implies Theorem~\ref{thm:allthesame}.
The spirit of this section (but not the techniques) is similar to~\cite{Walukiewicz96,Walukiewicz02} where Walukiewicz shows that $\Aut(\ofoe) \equiv \mso \equiv \mufoe$ on trees.

\medskip
Recall that in Remark~\ref{rem:parameters} we observed that, in general, the parameters (free variables) of a fixpoint cannot be avoided unless we go up in the arity hierarchy. From the equivalence of (iii) and (iv) we get, in particular, the following corollary which says that, on trees, it is actually possible to get rid of the parameters without increasing the arity of the fixpoints.

\begin{corollary}
	On trees, every $\varphi\in\mucafoe$ is equivalent to a parameter-free $\varphi'\in\mucafoe$.
\end{corollary}
\begin{proof}
	Direct from the equivalence of (iii) and (iv) and that the fixpoints of $\mucaffoe$ do not have parameters.
\end{proof}

In order to develop the results we first perform an analysis of the fixpoints of a special class of maps that `restrict to descendants.' The intuition behind these maps, which we will introduce shortly, is that they are induced by formulas which are invariant under generated submodels. That is, formulas which when evaluated at a certain point, can only talk about the descendants of that point.

%% file: fixpoint-rdes.tex

\begin{definition}\label{def:restodes}
	A map $G:\wp(\npmoddom)^n\to\wp(\npmoddom)$ on a model $\npmodel$ is said to \emph{restrict to descendants} if for every $s \in \npmoddom$ and $\vlist{X} \in \wp(\npmoddom)^n$ we have that $s \in G(\vlist{X})$ iff $s \in G(\vlist{X} \cap R^*[s])$.
\end{definition}

\noindent Our main interest in this subsection is to prove the following theorem.

\begin{theorem}\label{thm:propssucmap}~
		If $G(X,\vlist{Y})$ is monotone and restricts to descendants then
		\[H(\vlist{Y}) := \lfp_X.G(X,\vlist{Y})\]
		also restricts to descendants.
		%
\end{theorem}
\begin{proof}
	Define the abbreviations $F(X) := G(X, \vlist{Y})$ and $F_s(X) := G(X, \vlist{Y}\cap R^*[s]) \cap R^*[s]$. We first prove the following claim linking $F$ and $F_s$.
	\begin{claimfirst}
		For every $t\in R^*[s]$ we have that $t\in F(X)$ iff $t \in F_s(X)$.
	\end{claimfirst}
	\begin{pfclaim}
		Direct using restriction to descendants and monotonicity, together with the observation that $R^*[t] \subseteq R^*[s]$.
	\end{pfclaim}
	Moreover, this connection lifts to the approximants of the least fixpoints of $F$ and $F_s$.
	\begin{claim}
		For every $t\in R^*[s]$ we have that $t\in F^\alpha(\nada)$ iff $t \in F_s^\alpha(\nada)$.
	\end{claim}
	\begin{pfclaim}
		We prove it by transfinite induction. It is clear for $F^0(\nada) = \nada = F^0_s(\nada)$. For the inductive case of a successor ordinal $\alpha+1$ let $t$ belong to $R^*[s]$. We have
		\begin{align*}
			t\in F^{\alpha+1}(\nada)
			& \quad\text{iff}\quad t\in F(F^{\alpha}(\nada))
			\tag{by definition}\\
			& \quad\text{iff}\quad t\in F_s(F^{\alpha}(\nada))
			\tag{by Claim~1}\\
			& \quad\text{iff}\quad t\in F_s(F_s^{\alpha}(\nada))
			\tag{by IH}\\
			& \quad\text{iff}\quad t \in F_s^{\alpha+1}(\nada).
			\tag{by definition}
		\end{align*}
		The case of limit ordinals is left to the reader.
	\end{pfclaim}
	The following claim is direct by the definition of $F_s$ as $G(X, \vlist{Y}\cap R^*[s]) \cap R^*[s]$.
	\begin{claim}
		$\lfp_X.F_s(X) \subseteq \lfp_X.G(X,\vlist{Y} \cap R^*[s])$.
	\end{claim}
	Finally, we use the claims and prove that $s\in \lfp_X.G(X,\vlist{Y})$ iff $s\in \lfp_X.G(X,\vlist{Y} \cap R^*[s])$ which means that $H(\vlist{Y})$ restricts to descendants.

	\medskip\noindent
	\fbox{$\Leftarrow$}
	The key observation for this direction is that $G(X,\vlist{Y} \cap R^*[s]) \subseteq G(X,\vlist{Y})$ by monotonicity of $G$. Therefore $\lfp_X.G(X,\vlist{Y} \cap R^*[s]) \subseteq \lfp_X.G(X,\vlist{Y})$.

	\medskip\noindent
	\fbox{$\Rightarrow$}
	If $s\in \lfp_X.G(X,\vlist{Y})$ then there is an ordinal $\beta$ such that $s\in F^\beta(\nada)$. By Claim~2, we then have that $s\in F_s^\beta(\nada)$ and hence $s\in \lfp(F_s)$. Using Claim~3 we can conclude that $s\in \lfp_X.G(X,\vlist{Y} \cap R^*[s])$.
\end{proof}

%% file: soaut-logic-forward.tex

It is easy to see that parity automata `restrict to descendants.' That is, whenever the game $\agame(\aut,\tmodel)$ is at some basic position $(a,s)$, the match can only continue to positions of the form $(b,t)$ where $t \in R^*[s]$. Moreover, the game can never go back towards the root of the tree. Therefore, it is to be expected that formulas that correspond to parity automata also `restrict to descendants.' This concept is formalized as follows.

\begin{definition}
The \emph{forward-looking fragment} $\muffoe$ of $\mufoe$ is the smallest collection of formulas containing all atomic formulas, closed under Boolean connectives, and such that:
\begin{itemize}
	\itemsep 0pt
	\item If $\vlist{x}=(x_1, \dots, x_m)$ are individual variables and $\varphi(\vlist{x},y)$ is a $\muffoe$-formula whose free variables are among $\vlist{x},y$ then the formulas
	\[
	\exists y. (R_{d}(x_j,y) \land \varphi(\vlist{x},y))
	\quad \text{and} \quad
	\forall y. (R_{d}(x_j,y) \to \varphi(\vlist{x},y))
	\]
	are in $\muffoe$ for all $1 \leq j \leq m$ and $\aact\in\acts$.
	\item If $\varphi(q,y)$ is a $\muffoe$-formula which is positive in $q$ and whose only free individual variable is $y$ then $[\lfp_{q{:}y}.\varphi(q,y)](x)$ is in $\muffoe$ for all $q\in\props$.
\end{itemize}
The forward-looking fragment of $\mucafoe$ is defined as $\mucaffoe := \mucafoe \cap \muffoe$.
\end{definition}

\begin{remark}
	Fixpoints of $\muffoe$ are parameter-free.
\end{remark}

\begin{definition}
	Let $\varphi\in\muffoe$ be such that $FV(\varphi) \subseteq \{\vlist{z}\}$. We say that $\varphi$ \emph{restricts} to descendants if for every model $\npmodel$, assignment $\ass$ and $p\in\props$ the following holds:
	\[
	\npmodel,\ass \models \varphi 
	\quad\text{iff}\quad
	\npmodel[p\resto R^*[\vlist{z}]],\ass \models \varphi 
	\]
	where $R^*[\vlist{z}] := \bigcup_i R^*[\ass(z_i)]$.
\end{definition}

\begin{remark}
	The reader might have expected an alternative definition which requires that 
	$
	\npmodel,\ass \models \varphi
	\text{ iff }
	\npmodel[\props\resto R^*[\vlist{z}]],\ass \models \varphi
	$
	or even that 
	$
	\npmodel,\ass \models \varphi
	\text{ iff }
	\npmodel[\props\resto R^*[FV(\varphi)]],\ass \models \varphi
	$. All these definitions can be proved to be equivalent, and we keep the above version because it will simplify our inductive proofs.
\end{remark}

\begin{remark}
	Restriction to descendants is a weak kind of invariance under generated submodels. Suppose that for formulas $\varphi\in\mufoe$ whose free variables are among $\vlist{x}$ we say that $\varphi$ is \emph{invariant under generated submodels} if for every model $\npmodel$ and assignment $\ass$ we have:
	\[
	\npmodel,\ass \models \varphi
	\quad\text{iff}\quad
	\npmodel\tup{\vlist{x}},\ass \models \varphi
	\]
	where $\npmodel\tup{\vlist{x}}$ is the submodel of $\npmodel$ generated by $\ass(x_1),\dots,\ass(x_m)$.
	As an example, the formula $\varphi(x) := \exists y. \lnot R(x,y)$ is not invariant under generated submodels but, as no $p$ occurs in it, it trivially restricts to descendants of $x$. The fragment $\muffoe$ can be proved to be invariant under generated submodels, but we don't do it in this paper because we will not need it. 
\end{remark}

\begin{proposition}\label{prop:rsformrsmap}
Let $\varphi\in\mufoe$ restrict to descendants and be such that $FV(\varphi) \subseteq \{x\}$.
For every model $\npmodel$, assignment $\ass$, predicates $\qprops\subseteq\props$ and variable $x\in\fovar$, the map $G_{x}:\wp(\npmoddom)^n\to\wp(\npmoddom)$ given by
\[
G_{x}(\vlist{Z}) := \{ t\in\npmoddom \mid \npmodel[\qprops\mapsto\vlist{Z}],\ass[x\mapsto t] \models \varphi\}
\]
restricts to descendants.
\end{proposition}
\begin{proof}
	An element $t$ belongs to $G_x(\vlist{Z})$ iff $\npmodel[\qprops\mapsto\vlist{Z}],\ass[x\mapsto t] \models \varphi$. As $\varphi$ restricts to descendants, this occurs iff $\npmodel[\qprops\mapsto\vlist{Z}\cap R^*[t]],\ass[x\mapsto t] \models \varphi$. By definition of $G_x$, this is equivalent to saying that $t \in G_x(\vlist{Z}\cap R^*[t])$. That is, the map $G_x$ restricts to descendants.
\end{proof}

\begin{lemma}\label{lem:restosuc}
	Every $\varphi\in\muffoe$ restricts to descendants.
\end{lemma}
\begin{proof}
	%
	Fix $p\in \props$, we prove the statement by induction on $\varphi$.
	%
	\begin{itemize}
		\itemsep 0pt 
		\item If $\varphi$ does not include $p$ or $\varphi = p(x)$ the statement is clear.
		\item Let $\varphi(p,\vlist{x},\vlist{y}) = \psi_1(p,\vlist{x}) \lor \psi_2(p,\vlist{y})$; and consider $\vlist{z}$ such that $\vlist{x},\vlist{y} \subseteq \vlist{z}$.

		\fbox{$\Rightarrow$} Without loss of generality suppose $\npmodel,g \models \psi_1$, then by inductive hypothesis we know that $\npmodel[p\resto R^*[\vlist{z}]],g \models \psi_1$. From this we can conclude $\npmodel[p\resto R^*[\vlist{z}]],g \models \varphi$.

		\fbox{$\Leftarrow$} Without loss of generality suppose $\npmodel[p\resto R^*[\vlist{z}],g \models \psi_1$. By inductive hypothesis we get $\npmodel,g \models \psi_1$ which clearly implies $\npmodel,g \models \varphi$.
		%
		%
		\item Negation is handled by the inductive hypothesis.
		\item Let $\varphi(p,\vlist{x}) = \exists y. (R_{d}(x_j,y) \land \psi(\vlist{x},y))$; and consider $\vlist{z}$ such that $\vlist{x}\subseteq\vlist{z}$.

		\fbox{$\Rightarrow$} Suppose $\npmodel,\ass \models \varphi$. Then there is $s_y \in R_{d}[g(x_j)]$ such that $\npmodel,\ass[y\mapsto s_y] \models \psi(\vlist{x},y)$. By inductive hypothesis we get $\npmodel[p\resto R^*[\vlist{z},y]],\ass[y\mapsto s_y] \models \psi(\vlist{x},y)$ and as $s_y \in R_{d}[g(x_j)]$ and $x_j\in\vlist{z}$ we get that $\npmodel[p\resto R^*[\vlist{z}]],\ass[y\mapsto s_y] \models \psi(\vlist{x},y)$. From this, we can conclude that $\npmodel[p\resto R^*[\vlist{z}]],\ass \models \exists y. (R_{d}(x_j,y) \land \psi(\vlist{x},y))$.

		\fbox{$\Leftarrow$} Suppose $\npmodel[p\resto R^*[\vlist{z}]],\ass \models \varphi$. Then there exists an element $s_y \in R_{d}[g(x_j)]$ such that $\npmodel[p\resto R^*[\vlist{z}]],\ass[y\mapsto s_y] \models \psi(\vlist{x},y)$. As $s_y \in R_{d}[g(x_j)]$ and $x_j\in\vlist{z}$ we know that $R^*[\vlist{z}] = R^*[\vlist{z},y]$. Therefore we also have that $\npmodel[p\resto R^*[\vlist{z},y]],\ass[y\mapsto s_y] \models \psi(\vlist{x},y)$.
		By inductive hypothesis we get $\npmodel,\ass[y\mapsto s_y] \models \psi(\vlist{x},y)$. From this, we can conclude $\npmodel,\ass \models \exists y. (R_{d}(x_j,y) \land \psi(\vlist{x},y))$.
		\item Let $\varphi = [\lfp_{q{:}y}.\psi(q,y)](z)$. Observe that by definition of the fragment, we have $FV(\varphi) = \{z\}$, $q$ is positive in $\psi$ and $FV(\psi) \subseteq \{y\}$. Consider $\vlist{z}$ such that $z\in\vlist{z}$, we have to prove that 
		\[
		\npmodel,\ass \models \varphi \quad\text{iff}\quad \npmodel[p\resto R^*[\vlist{z}]],\ass \models \varphi.
		\]
		By semantics of the fixpoint operator $\npmodel,\ass \models \varphi$ iff $\ass(z) \in \lfp(F^\npmodel_{q:x})$ where
		\[
			F^\npmodel_{q:x}(Q) := \{ t\in\npmoddom \mid \npmodel[q\mapsto Q],\ass[x\mapsto t] \models \psi\}.
		\]
		It will be useful to take a slightly more general definition: consider the map
		\[
		G_{q:x}^\psi(Q,P) := \{ t\in\npmoddom \mid \npmodel[q\mapsto Q; p\mapsto P],\ass[x\mapsto t] \models \psi\}
		\]
		and observe that $F^\npmodel_{q:x}(Q) = G_{q:x}^\psi(Q,\tsval(p))$ and therefore their least fixpoints will be the same. By inductive hypothesis and Proposition~\ref{prop:rsformrsmap}, we know that $G_{q:x}^\psi(Q,P)$ restricts to descendants. Using Theorem~\ref{thm:propssucmap} we get that $\lfp_Q.G_{q:x}^\psi(Q,\tsval(p))$ restricts to descendants as well. That is,
		\[
		\ass(z) \in \lfp_Q.G_{q:x}^\psi(Q,\tsval(p))
		\quad\text{iff}\quad
		\ass(z) \in \lfp_Q.G_{q:x}^\psi(Q,\tsval(p)\cap R^*[\ass(z)]).
		\]
		Because $z\in\vlist{z}$ and the monotonicity of $G_{q:x}^\psi$, we also get that
		\[
		\ass(z) \in \lfp_Q.G_{q:x}^\psi(Q,\tsval(p)\cap R^*[\ass(z)])
		\quad\text{iff}\quad
		\ass(z) \in \lfp_Q.G_{q:x}^\psi(Q,\tsval(p)\cap R^*[\vlist{z}]).
		\]
		Using the definition of $F^\npmodel_{q:x}$ and the above equations we can conclude that
		\[
		\ass(z) \in \lfp(F^\npmodel_{q:x})
		\quad\text{iff}\quad
		\ass(z) \in \lfp(F^{\npmodel[p\resto R^*[\vlist{z}]]}_{q:x}).
		\]
		From this, we finally get $\npmodel,\ass \models \varphi$ iff $\npmodel[p\resto R^*[\vlist{z}]],\ass \models \varphi$. \qedhere
	\end{itemize}
\end{proof}

\paragraph{Historical remarks and related results.}
The fragment $\muffoe$ defined here is similar in spirit to the bounded and guarded fragments defined in~\cite{arec:hybr99a,GuardedFragment,GradelW99}. The most natural perspective is to see $\muffoe$ as an extension of the bounded fragment of first-order logic given in~\cite{arec:hybr99a} to first-order logic with fixpoints. In~\cite{GradelW99} the authors introduce a guarded fragment of $\mufoe$, however, they aim to make it as big as possible. For example, their formalism can define the mu-calculus with backward-looking modalities, and therefore is not invariant under generated submodels.

%% file: soaut-logic-wcl2aut.tex

\subsubsection{From $\wcl$ to $\AutWA(\ofoe)$, on trees}\label{ssec:wcl2aut}

We will first prove the following auxiliary result.

\begin{proposition}
  For every formula $\varphi\in\wcl(\props,\acts)$ with free variables $\mathsf{F} \subseteq \props$ there is an automaton $\aut_\varphi\in\AutWA(\ofoe,\mathsf{F})$ such that for every $\mathsf{F}$-tree $\tmodel$ we have $\tmodel \models \varphi$ iff $\tmodel \models \aut_\varphi$.
%
\end{proposition}
\begin{proof} The proof is by induction on $\varphi$.
\begin{itemize}
  \item For the base cases $\varphi = p \inc q$ and $\varphi = R_\aact(p,q)$,
  we give the following automata. 

  $\aut_{p\inc q} := \tup{\{a_0\}, \tmap, \pmap, a_0}$ where $\pmap(a_0) = 0$ and
  \[
    \tmap(a_0, c) :=
    \begin{cases}
      \bigwedge_\asort \forall x{:}\asort.a_0(x) & \text{if $q \in c$ or $p \notin c$,}\\
      \bot & \text{otherwise}
    \end{cases}
  \]

  $\aut_{R_\aact(p,q)} := \tup{\{a_0,a_1\}, \tmap, \pmap, a_0}$ where $\pmap(a_0) = \pmap(a_1) = 0$ and
  \begin{align*}
    \tmap(a_0, c) & :=
    \begin{cases}
      \exists x{:}\aact.a_1(x) \land \bigwedge_\asort(\forall y{:}\asort.a_0(y)) & \text{if $p \in c$,}\\
      \bigwedge_\asort \forall x{:}\asort.a_0(x) & \text{otherwise.}
    \end{cases} \\
    \tmap(a_1, c) & :=
    \begin{cases}
      \top & \text{if $q\in c$,}\\
      \bot & \text{if $q\notin c$.}
    \end{cases}
  \end{align*}
  It is easy to syntactically check that these automata are additive-weak and also it is not too difficult to see that they do what they should.

  \noindent\emph{Remark.} A nice observation is that, modally, these automata correspond to $\Box^*(p\to q)$ and $\Box^*(p \to \tup{\aact}q)$ respectively. Also, none of the following automata constructions (i.e., Booleans and projection) creat cycles on the automata. This shows that all the ``recursive power'' of these automata boils down to the $\Box^*$ construction.

  \item For the Boolean cases, where $\varphi = \psi_1 \vee \psi_2$ or $\varphi = \neg\psi$
  we refer to the closure properties of recognizable tree languages, see 
  Theorem~\ref{t:cl-dis} and Theorem~\ref{t:cl-cmp}, 
  respectively.
  
\item 
For the case $\varphi = \exists p. \psi$ let $\mathsf{F}$ be the set of free variables of $\varphi$.
We only consider the case where $p$ is free in $\psi$ as otherwise $\varphi \equiv \psi$ and by
induction hypothesis we already have an automaton $\aut_\psi$ which we can use as $\aut_\varphi$.

Let $\aut_{\psi}\in\AutWA(\ofoe,\mathsf{F}\dcup\{p\})$ be given by the
inductive hypothesis.
We define $\aut_\varphi := {\existswc p}.\aut_{\psi}$
using the construction given in Definition~\ref{def:autproj}.
Observe that $\aut_\varphi$ is an automaton over $\wp(\mathsf{F})$
and that:
\begin{align*}
\tmodel \models {\existswc p}.\aut_{\psi}
   & \quad\text{iff}\quad
     \tmodel[p\mapsto X_p] \models \aut_{\psi} \text{ for a finite chain $X_p \subseteq \tmoddom$}
   \tag{Lemma~\ref{lem:projection}}
\\ & \quad\text{iff}\quad
     \tmodel[p\mapsto X_p] \models \psi \text{ for a finite chain $X_p \subseteq \tmoddom$}
   \tag{induction hypothesis}
\\ & \quad\text{iff}\quad
    \tmodel \models \existswc p. \psi
   \tag{semantics of $\wcl$}
\end{align*}
\end{itemize}
This finishes the proof of the auxiliary result.
\end{proof}

\noindent For the general case, we show the following proposition.

\begin{proposition}
For every $\varphi\in\wcl(\props,\acts)$
there is an automaton $\aut_\varphi\in\AutWA(\ofoe,\props)$ such that
for every $\props$-tree $\tmodel$ we have
$\tmodel \models \varphi$ iff $\tmodel \models \aut_\varphi$.
%
%
\end{proposition}
\begin{proof}
The only observation that we need is that we can transform every automaton $\aut\in\Aut(\llang,\mathsf{F})$ 
which runs on $\mathsf{F}$-trees to an automaton $\aut^\props := \tup{A,\tmap^\props,\pmap,a_I}\in\Aut(\llang,\props)$ which runs on $\props$-trees by defining 
\[
\tmap^\props(a,c) := \tmap(a,c \cap \mathsf{F})
\]
for every $a\in A$ and $c\in \wp(\props)$. The intuition behind this construction is that $\aut^\props$ ignores the $(\props\setminus\mathsf{F})$ part of the colors of the nodes.
\end{proof}

%% file: soaut-logic-aut2caf.tex

\subsubsection{From $\AutWA(\ofoe)$ to $\mucaffoe$, on all models}\label{sec:aut2caf}

With each initialized automaton $\aut\in\AutWA(\ofoe)$ we associate a formula $\varphi_\aut\in\mucaffoe$ such that $\aut\equiv\varphi_\aut$ for all transition systems (that is, not necessarily trees). To do it, we first show that every parity automaton can be transformed into an automaton whose induced graph is a tree with back edges. This kind of structure has a natural counterpart as a formula.

\input{fig-backedges}

Intuitively, the tree part of the automaton is used to define the scaffolding of the corresponding formulas. On top of that, the nodes which are the target of back-edges will correspond to binding definitions of fixpoint variables. Fig.~\ref{fig:autunravelling} shows an illustration of this intuition where the target formula is taken to be in the $\mu$-calculus. This is done for illustrative reasons, in our case we will actually have binding definitions given by the first-order fixpoint operator $[\lfp_{a_i{:}y}.\varphi_i(y)](x)$.

\begin{definition}
A directed graph $(G,R \subseteq G^2)$ is a \emph{tree with back edges} if there is a partition $R = E \uplus B$ of the edges into tree edges and back edges such that $(G, E)$ is indeed a directed tree, and whenever $(u, v) \in B$, then $(v,u) \in E^*$.
\end{definition}

Berwanger~\cite{Berwanger-Thesis} shows that every finite model can be transformed, via partial unravelling, into a bisimilar finite model which is a tree with back edges. An unravelling technique is also present in Janin's habilitation thesis~\cite[Section~3.2.3]{Janin2006}, where he puts modal parity automata into the shape of trees with back edges. We adapt these ideas to our setting by defining a similar transformation on parity automata of an arbitrary one-step language~$\llang$. 

\begin{definition}\label{def:parunral}
The \emph{finite unravelling} of a parity automaton $\aut = \tup{A,\tmap,\pmap,a_I}$ is the parity automaton $\aut^u = \tup{A^u,\tmap^u,\pmap^u,a_I^u}$ such that
\begin{enumerate}[1.]
	\itemsep 0 pt
	%
	\item $A^u$ is made of non-empty finite sequences $\vlist{a}\in A^+$ such that $a_0 = a_I$ and $a_i \reach_\aut a_{i+1}$,
	\item $a_I^u$ is the one-element sequence containing only $a_I$,
	\item Every element of $A^u$ is reachable from $a_I^u$,
	\item $\pmap^u(\vlist{a}{\cdot}a_k) = \pmap(a_k)$, and
	\item \label{def:finunrav:tmap}$\tmap^u(\vlist{a}{\cdot}a_k,c) = \tmap(a_k,c)[ b \mapsto \mathsf{update}(\vlist{a}{\cdot}a_k,b) \mid b \in A]$ where $\mathsf{update}(a_0,\dots,a_k,b)$ is defined as~(a) the shortest prefix $a_0,\dots,a_i$ of $a_0,\dots,a_k,b$ such that $a_i=b$ and,~(b) for every $i< j \leq k$ we have that $\pmap(a_i) \leq \pmap(a_j)$; that is, the minimum parity encountered in the cycle $a_i,a_{i+1},\dots,a_i$ is $\pmap(a_i)$.
	%
\end{enumerate}
\end{definition}

\begin{remark}
	Condition~(\ref{def:finunrav:tmap}b) is there to ensure that the target of a back-edge is `of maximum priority' among the elements of the given cycle.	This condition is not necessary for $\aut^u$ to be a tree with back edges. However, it is necessary to make the (possible) alternation of fixpoints in the target formula mimic the parity game.

	In our case, as the automata that we use are \emph{weak}, all the parities of the elements of a given cycle are the same. Since the resulting formula will not have any alternation, we could have simply left condition~(\ref{def:finunrav:tmap}b) out. We chose to keep it for compatibility with the results of~\cite{Janin2006}, and for completeness.
\end{remark}


\begin{lemma}[{\cite[Lemma~3.1.2.3]{Janin2006}}]\label{lem:autunravel}
	$\aut \equiv \aut^u$ for any finite unravelling $\aut^u$ of $\aut$.
\end{lemma}

As we are working with additive-weak automata, we need to prove that the construction preserves the properties of weakness and additivity. In other words,

\begin{proposition}\label{prop:finunravwa}
	If $\aut \in \AutWA(\llang)$ then $\aut^u \in \AutWA(\llang)$.
\end{proposition}
\begin{proof}
	Define the projection $\last:A^+ \to A$ as $\last(a_0,\dots,a_k) := a_k$. For sets $B \subseteq A^+$ the projection is extended to $\last:\wp(A^+)\to\wp(A)$ by defining $\last(B) := \{\last(b) \mid b\in B\}$. The following observations will be useful:

	\begin{claimfirst}\label{prop:finunravwa:c1}
	 	If $\mccomp \subseteq A^u$ is a strongly connected component in $\aut^u$ then $\last(\mccomp)$ is a strongly connected component in $\aut$.
	\end{claimfirst}
	\begin{pfclaim}
		It is enough to prove that if $a_0,\dots,a_k \ord_{\aut^u} b_0,\dots,b_{k'}$ then $a_k \ord_\aut b_{k'}$, As the notion of strongly connected component is defined in terms of $\ord$.

		Now, because $\ord$ is the reflexive-transitive closure of $\reach$, it will actually be enough to prove that if $a_0,\dots,a_k \reach_{\aut^u} b_0,\dots,b_{k'}$ then $a_k \reach_\aut b_{k'}$. For this, just observe that if $a_0,\dots,a_k \reach_{\aut^u} b_0,\dots,b_{k'}$ then, by contruction of $\tmap^u$ in Definition~\ref{def:parunral}, we have that $b_{k'}$ occurs in $\tmap(a_k,c)$ for some $c\in\wp(\props)$. That is, $a_k \reach_\aut b_{k'}$.
	\end{pfclaim}

	\begin{claim}\label{prop:finunravwa:c2}
	For every strongly connected component $\mccomp \subseteq A^u$ we have $\pmap(\last(\mccomp)) = \pmap^u(\mccomp)$.
	\end{claim}
	\begin{pfclaim}
		By definition of $\pmap^u$.
	\end{pfclaim}

	For the \emph{weakness} condition we proceed as follows: by Claim~\ref{prop:finunravwa:c1} we know that if $C$ is a maximal strongly connected component in $\aut^u$ then $\last(C)$ will also be a strongly connected component in $\aut$. As $\aut$ is weak, then every element of $\last(C)$ will have the same parity, which we call $\pmap^u(\last(C))$. Using Claim~\ref{prop:finunravwa:c2}, we know that $\pmap(\last(\mccomp)) = \pmap^u(\mccomp)$, and therefore get that every element of $C$ has the same parity.
	
	For the \emph{additivity} condition let $\mccomp \subseteq A^u$ be a maximally connected component with $\pmap^u(\mccomp) = 1$ and let $\vlist{a}$ be an element of $\mccomp$. We want to prove that $\tmap^u(\vlist{a},c)$ is completely additive in $\mccomp$, for every color $c\in\wp(\props)$. Define $\varphi := \tmap(\last(\vlist{a}),c)$. It is not difficult to observe that, as $\last(\vlist{a})$ is in the connected component $\last(C)$, then $\varphi$ is completely additive in $\last(C)$. The key observation now is that if we substitute all the names in $\varphi$ from $\last(C)$ with some new set of names $A'$ then the new formula will be completely additive in $A'$. To conclude, we just recall that $\tmap^u(\vlist{a},c)$ is obtained by substituting the names from $\last(C)$ in $\varphi$ with new names that belong to $C$. Using the previous observation, we get that $\tmap^u(\vlist{a},c)$ is completely additive in $C$. We leave the case of $\pmap^u(\mccomp)=0$ to the reader.
\end{proof}

Next, we show that for every $\aut\in\Aut(\ofoe)$ it is possible to give an equivalent formula $\varphi_\aut(x)\in\mufoe$. Shortly after that, we will focus on the completely additive fragments of these formalisms. 

\begin{theorem}\label{thm:autofoetomuffoe}
	For every automaton $\aut\in\AutWA(\ofoe,\props)$ there is a formula $\varphi_\aut(x) \in\mucaffoe(\props)$ with exactly one free variable $x$, such that for every transition system $\model$,
	\[
	\model \models \aut \quad\text{iff}\quad \model \models \varphi_\aut(s_I).
	\]
\end{theorem}
\begin{proof}
	Because of Lemma~\ref{lem:autunravel} we assume that $\aut$ can be decomposed as a tree with back edges $(A,E,B)$. First we need the following definitions:
	\begin{align*}
	\beta_a(x) & :=  \bigvee_{c\in C} \big( \tau_c(x) \land \tmap^g_{a,c}(x) \big)\\
	\tau_c(x) & :=  \bigwedge_{p\in c}p(x) \land \bigwedge_{p\in \props\setminus c} \lnot p(x)
	\end{align*}
	where the $\tmap^g_{a,c}$ is a guarded version of $\tmap(a,c)$, defined as
	\begin{align*}
	\tmap^g_{a,c}(x) & := \tmap(a,c)[\exists y{:}\asort.\alpha \mapsto \exists y.R_\asort(x,y)\land \alpha; \forall y{:}\asort.\alpha \mapsto \forall y.R_\asort(x,y)\to \alpha].
	\end{align*}
	Now we define auxiliary formulas $\chi_a(x)$ and $\beta_a^\dag(z)$ by mutual induction, on the tree $(A,E)$. 
	\[
		\beta_a^\dag(z) :=
		\begin{cases}
			\beta_a(z) & \text{if $a$ is a leaf,}\\
			\beta_a(z)[a'(y) \mapsto \chi_{a'}(y) \mid (a,a') \in E] & \text{otherwise.}
		\end{cases}
	\]
	\[
	\chi_a(x) :=
	\begin{cases}
		\beta_a^\dag(x) & \text{if $a \notin \Ran(B)$} \\
		[\lfp_{a{:}z}. \beta_a^\dag(z)](x) & \text{if $a \in \Ran(B)$ and $\pmap(a)=1$} \\
		[\gfp_{a{:}z}. \beta_a^\dag(z)](x) & \text{if $a \in \Ran(B)$ and $\pmap(a)=0$}.
	\end{cases}
	\]
	Finally, we set $\varphi_\aut(x) := \chi_{a_I}(x)$. It is left to the reader to prove that $\aut \equiv \varphi_\aut(x)$. A very similar translation (for modal automata) is given in~\cite[Lemma~3.2.3.2--3]{Janin2006}. We still have to prove that $\varphi_\aut(x)$ lands in the appropriate fragment, i.e., that $\varphi_\aut(x)\in\mucafoe\cap\muffoe$.
	\begin{claimfirst}
		$\varphi_\aut(x)\in\mucafoe$.
	\end{claimfirst}
	\begin{pfclaim}
		It is not difficult to show, inductively, that if $a \in A$ belongs to a maximal strongly connected component $\mccomp\subseteq A$ of parity 1 (resp. 0) then $\beta_a^\dag(z)$ will be completely additive (resp. multiplicative) in $\mccomp\subseteq A$. This is enough, because then the fixpoint operators bind formulas of the right kind.
	\end{pfclaim}

	\begin{claim}
		$\varphi_\aut(x)\in\muffoe$.
	\end{claim}
	\begin{pfclaim}
	The formula $\varphi_\aut(x)$ can be seen to belong to $\muffoe$ by a simple inspection of the construction: more specifically, the definition of $\tmap^g_{a,c}(x)$ guards every quantifier, and the fixpoint operators introduced in every $\chi_a(x)$ are exactly of the form required by the fragment $\muffoe$. 
	\end{pfclaim}
	It is worth observing, as a consequence of the last claim, that the fixpoints operators of $\varphi_\aut(x)$ do not use parameters (that is, they don't have extra free --individual-- variables).
\end{proof}

\subsubsection{From $\mucaffoe$ to $\mucafoe$, on all models}

This inclusion is trivial because $\mucaffoe\subseteq\mucafoe$.

%% file: fig-backedges.tex

\begin{figure}[h]
\centering
\begin{tikzpicture}[
font=\small,
->,>=stealth',
n0/.style={circle,thick,draw=blue!75,fill=blue!20},
n1/.style={circle,thick,draw=red!75,fill=red!20},
n2/.style={circle,thick,draw=orange!75,fill=yellow!20},
n3/.style={circle,thick,draw=black!75,fill=black!20},
every node/.style={inner sep=1pt,minimum size=12pt}
]
\begin{scope}[level/.style={level distance=1.7cm, sibling distance=2.5cm},yshift=-0.25cm]
\node[n0] (o0) {$0$}
    child {
        node[n1] (o1) {$1$}
    }
    child {
        node[n2] (o2) {$2$}
    };

\node[below of=o0,n3] (o3) {$3$};
\draw (o1) -- (o2);
\draw (o3) edge [bend right=20] (o1);
\draw (o1) edge [bend right=20] (o3);
\draw (o2) edge [bend right] (o0);
\end{scope}
\begin{scope}[xshift=5.5cm,level/.style={level distance=1cm, sibling distance=2cm}]
\node[n0] (u0) {$0$}
    child {
        node[n1] (u1) {$1$}
        child {
            node[n2] (u2') {$2'$}
        }
        child {
            node[n3] (u3) {$3$}
        }
    }
    child {
        node[n2] (u2) {$2$}
    };

\draw[dashed,draw=black!75] (u3) edge [bend right=40] (u1);
\draw[dashed,draw=black!75] (u2) edge [bend right=40] (u0);
\draw[dashed,draw=black!75] (u2') edge [bend left=40] (u0);
\end{scope}
\begin{scope}[xshift=10.5cm,level/.style={level distance=1cm, sibling distance=2cm}]
\node {$\mu a_0.\varphi_0$}
    child {
        node {$\mu a_1.\varphi_1$}
        child {
            node {$\varphi_{2'}(a_0)$}
        }
        child {
            node {$\varphi_3(a_1)$}
        }
    }
    child {
        node {$\varphi_2(a_0)$}
    };
\end{scope}
\end{tikzpicture}
\caption{Automata, finite unravelling and formula structure.}
\label{fig:autunravelling}
\end{figure}
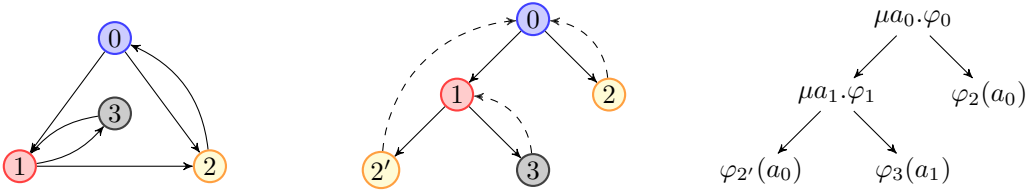

%% file: soaut-logic-caf2wcl.tex

\subsubsection{From $\mucafoe$ to $\wcl$, on trees}\label{ssec:cafoe2wcl}

In this case we will make use of the correspondence between one-sorted \wcl and the two-sorted version 2\wcl given in Section~\ref{sec:wcl} and give a translation from $\mucafoe$ to $2\wcl$. The translation is given inductively, it is clear that the interesting part is the simulation of the fixpoint operator using the weak chain quantifier.

Let $\varphi = [\lfp_{p{:}y}.\psi(p,y,\vlist{z})](x) \in \mucafoe$ be such that $FV(\psi) \subseteq \{y,\vlist{z}\}$.
Observe that this formula might have parameters in the fixpoint. That is, the variables $\vlist{z}$ will be free in $\varphi$ if they are free in $\psi$. It will be useful (and necessary) to get rid of them. First we define the auxiliary formula
\[
\psi'(Z_1,\dots,Z_k,p,y) := \exists \vlist{z}.\bigwedge_i Z_i(z_i) \land \psi,
\]
and we observe that, assuming that we could use second-order quantifiers, the following equivalence holds:
\[
[\lfp_{p{:}y}.\psi(p,y,\vlist{z})](x) \equiv \existswc \vlist{Z}.\bigwedge_i Z_i = \{z_i\} \land [\lfp_{p{:}y}.\psi'(Z_1,\dots,Z_k,p,y)](x),\tag{$*$}\label{eq:cafoe2wcl}
\]
where on the right-hand side the fixpoint operator does \emph{not} have (individual) parameters. Using an adequate complexity on the formulas we can still apply the induction hypothesis to $\psi'$. For example, assign a complexity to the formulas in $\mucafoe$ with a function $|\cdot|:\mucafoe\to\nat$ defined as follows:
\begin{align*}
	&|P(x_1,\dots,x_n)| = n+1 && |\lnot\alpha| = |\alpha|+1 \\
	&|[\lfp_{p{:}y}.\alpha(p,y,\vlist{z})](x)| = 10 * |\alpha| && |\alpha_1\land\alpha_2| = |\alpha_1|+|\alpha_2|+1 \\
	&|\existswc X.\alpha| = |\alpha|+1 && |\exists x.\alpha| = |\alpha|+1 .
\end{align*}
\begin{claimfirst}
	$|\varphi| > |\psi'|$.
\end{claimfirst}
\begin{pfclaim}
	By definition of the complexity we have that $|\varphi| = 10*|\psi|$. Now we approximate the complexity of $\psi'$. Recall that $\psi' = \exists \vlist{z}.\bigwedge_i Z_i(z_i) \land \psi$, then
	\[
		|\psi'| \leq \underbrace{|\psi|}_{\exists\vlist{z}} + \underbrace{3*|\psi|}_{\bigwedge_i Z_i(z_i)} + \underbrace{1}_{\land} + \underbrace{|\psi|}_{\psi}.
	\]
	Therefore, $|\psi'| \leq 6*|\psi| < 10*|\psi| = |\varphi|$.
\end{pfclaim}
By induction hypothesis we know that there is a formula $\psi'_w(Z_1,\dots,Z_k,p,y) \in 2\wcl$, which is equivalent to $\psi'(Z_1,\dots,Z_k,p,y)$.
\begin{claim}\label{claim:psiinv}
	The formula $\psi'_w(Z_1,\dots,Z_k,p,y)$ restricts to descendants (on trees).
\end{claim}
\begin{pfclaim}
	Combining the results of the last sections, that is, the translations
	\[
	\wcl \to \AutWA(\ofoe) \to \mucaffoe
	\]
	we get that $\psi'_w(Z_1,\dots,Z_k,p,y)$ is equivalent (on trees) to a formula in $\mucaffoe$. By Lemma~\ref{lem:restosuc} we know that these formulas restrict to descendants.
\end{pfclaim}

Now we know that $\psi'$ is completely additive in $p$ and restricts to descendants. Also, we modified it in such a way that there is only one free variable. We did this in order to use Proposition~\ref{prop:rsformrsmap} and obtain that the functional $F^{\psi'}$ induced by $\psi'$ is completely additive and restricts to descendants as well.
The following theorem, which will be critical for our translation, combines the content of Theorem~\ref{thm:propscafmap} and Theorem~\ref{thm:propssucmap} and gives a characterization of the fixpoint of maps that, at the same time, are completely additive and restrict to descendants.

\begin{theorem}\label{thm:restofinsucchain}
	Let $F:\wp(\npmoddom)\to\wp(\npmoddom)$ be completely additive and restrict to descendants. 
	For every $s\in\npmoddom$ we have that $s \in \lfp(F)$ iff $s \in \lfp(F_{\resto Y})$ for some finite chain $Y$.
\end{theorem}
\begin{proof}
	The proof of this theorem follows from a minor modification of Lemma~\ref{lem:charcaffp} for completely additive maps, which provides elements $t_1,\dots,t_k = s$ such that $t_i \in F^i(\nada)$. The key observation is that, because $F$ restricts to descendants, we can choose the elements in a way that $t_i R^* t_{i+1}$ for all $i$.
	
	The main change to the proof of Lemma~\ref{lem:charcaffp} is when we want to define $u_i$ in terms of $u_{i+1} \in F^{i+1}(\nada)$. By definition we have that $u_{i+1} \in G(F^i(\nada),\vlist{Y})$. Now we can use that \emph{$G$ restricts to descendants} and get that $u_{i+1} \in G(F^i(\nada) \cap R^*[u_{i+1}],\vlist{Y} \cap R^*[u_{i+1}])$.
	By complete additivity of $G$ there is a quasi-atom $(T,\vlist{Q}')$ of $(F^i(\nada)\cap R^*[u_{i+1}],\vlist{Y}\cap R^*[u_{i+1}])$ such that $u_{i+1} \in G(T,\vlist{Q}')$. This means that the element chosen from $T$ will be a descendant of $u_{i+1}$ and therefore we will get a (finite) chain.
\end{proof}

\begin{proposition}\label{prop:mucafoe2wcl}
There is an effective translation $(-)^t : \mucafoe \to 2\wcl$ such that
for every $\tmodel$ and $\varphi \in\mucafoe$ we have
$\tmodel \models \varphi$ iff $\tmodel \models \varphi^t$.
\end{proposition}
\begin{proof}
Clearly the interesing case is that of the fixpoint operator. We define the translation of the fixpoint as follows:
\begin{align*}
([\lfp_{p{:}y}.\psi(p,y)](x))^t & :=  \existswc \vlist{Z}.\bigwedge_i Z_i = \{z\} \land~ \\
								& \phantom{:=\;\;} \existswc Y.\big(\forallwc W \subseteq Y.W \in \mathsf{PRE}(F^{\psi'}_{\resto Y}) \to x\in W \big) \\
W \in \mathsf{PRE}(F^{\psi'}_{\resto Y}) & := \forall v. \psi_w'(Z_1,\dots,Z_k,W,v) \land v\in Y \to v\in W.
\end{align*}
The first conjunct of this translation is introduced to get rid of the parameters of $\psi$, and is justified by~\eqref{eq:cafoe2wcl}. To justify the second part we proceed as follows:
first recall that the translation of $[\lfp_{p{:}y}.\psi(p,y)](x)$ into \mso is given by
\[
\forall W.\big(W \in \mathsf{PRE}(F^\psi) \to x\in W \big).
\tag{t-MSO}\label{eq:tmso}
\]
where $W \in \mathsf{PRE}(F^\psi)$ expresses that $W$ is a prefixpoint of $F^\psi:\wp(\moddom)\to\wp(\moddom)$. 
This translation is based on the following fact about fixpoints of monotone maps:
\[
s \in \lfp(F^{\psi}) \quad\text{iff}\quad s \in \bigcap\{W\subseteq S \mid W \in \mathsf{PRE}(F^{\psi})\}.
\tag{$\mathsf{PRE}$}\label{eq:pre}
\]
It is easy to see that~\eqref{eq:tmso} exactly expresses that $\ass(x)$ has to belong to every prefixpoint of $F^\psi$. In our translation $(-)^t$, however, we cannot make use of the set quantifier $\exists W$, since we are dealing with \wcl. The crucial observation is that, as $F^{\psi'}:\wp(\tmoddom)\to\wp(\tmoddom)$ is completely additive and restricts to descendants, then we can use Theorem~\ref{thm:restofinsucchain} to prove that, without loss of generality, we can restrict ourselves to finite chains, in the following sense:
\begin{align*}
s \in \lfp(F^{\psi'})
& \quad\text{iff}\quad s \in \lfp(F^{\psi'}_{\resto Y}) \text{ for some f.c. $Y$}
\tag{Theorem~\ref{thm:restofinsucchain}}\\
& \quad\text{iff}\quad s \in \bigcap\{W\subseteq \tmoddom \mid W \in \mathsf{PRE}(F^{\psi'}_{\resto Y})\} \text{ for some f.c. $Y$}
\tag{\ref{eq:pre}}\\
& \quad\text{iff}\quad s \in \bigcap\{W \subseteq Y \mid W \in \mathsf{PRE}(F^{\psi'}_{\resto Y})\} \text{ for some f.c. $Y$}.
\tag{Image of $F^{\psi'}_{\resto Y}$}
\end{align*}
\noindent Therefore, the second part of the translation $(-)^t$ basically expresses the same as~\eqref{eq:tmso} but relativized to a finite chain $Y$. The correctness of the translation is then justified by the above equations.
\end{proof}

%% file: exp-intro.tex

In this section, we characterize the bisimulation-invariant fragment of the main formalisms that we have been using throughout the article. Our final objective is to show that
\[
\pdl \equiv \binfotc/{\bis}
\quad\text{and}\quad
\pdl \equiv \wcl/{\bis}.
\]
%
That is, we will prove Theorems~\ref{thm:mainbinfotc} and~\ref{thm:mainwcl}. Moreover, we show that the equivalences are effective. As a first step, the first subsection characterizes the bisimulation-invariant fragment of $\AutWA(\ofoe)$, which is our bridge between the logics of this paper. The following subsections prove the two main bisimulation-invariance results. Even though the technique is almost the same, these statements will be proved in three different subsections which also provide additional remarks about the logic in question. We think this is the clearest way to present it, in spite of some repetition of the arguments.

%% file: exp-autofoe.tex

In this subsection we will define a construction $(-)^{\bullet}: \AutWA(\ofoe) \to \AutWA(\ofo)$ such that
for every automaton $\aut$ and transition system $\model$ we have:
\[
\aut^{\bullet} \text{ accepts } \model \quad\text{iff}\quad \aut \text{ accepts } \model^{\omega}
\]
where $\model^\omega$ is the $\omega$-unravelling of $\model$ (defined in Section~\ref{ssec:prelim_trees}). From this, it is easy to prove, as a byproduct, that $\AutWA(\ofoe)/{\bis} \equiv \AutWA(\ofo)$. As we shall see, the map $(-)^{\bullet}$ is completely determined at the 
one-step level, that is, by some model-theoretic connection between 
$\ofoe$ and $\ofo$.


\begin{definition}
We define the one-step translation
$(-)_1^{\bullet} : \ofoe^+(A,\sorts) \pto \ofo^+(A,\sorts)$ on one-step formulas of $\ofoe^+(A,\sorts)$ which are in strict basic form, as follows:
\[
( \posdbnfofoe{\vlist{T}}{\Pi}_\asort )^{\bullet}_1 :=
\posdgbnfofo{\vlist{T}}{\Pi}_\asort
\]
where on the right hand side $\vlist{T}$ is seen as a set.
\end{definition}

\noindent The key property of this translation is the following.

\begin{proposition} 
\label{prop:pinvariant}
For every one-step model $(D_1,\dots,D_n,V)$ and $\alpha \in \ofoe^+(A)$ we have
\[
(D_1,\dots,D_n,V) \models \alpha^{\bullet} \text{ iff } (D_1\times\{1\}\times \omega,\dots,D_n\times\{n\}\times \omega,\val_\pi) \models \alpha,
\]
where $\val_{\pi}$ 
 is the induced valuation given by 
$\val_{\pi}(a) := \{ (d,i,k) \mid d \in \val(a), 1\leq i \leq n, k\in\omega\}$.
\end{proposition}

In the above proposition the model on the left (call it $\osmodel$) is an arbitrary (not necessarily strict) one-step model, whereas the one in the right (call it $\osmodel_\omega$) is a \emph{strict} one-step model. The strictness is obtained by tagging elements of $D_i$, so that the union $\bigcup_i D_i$ is disjoint. Also, unrelated to the strictness of the models, observe that $\osmodel_\omega$ has $\omega$-many copies of each element of $\osmodel$.

\begin{proof}
\fbox{$\Rightarrow$} Let $\osmodel \models \posdgbnfofo{\vlist{T}}{\Pi}_\asort$, we prove that $\osmodel_\omega \models \posdbnfofoe{\vlist{T}}{\Pi}_\asort$. The existential part ($\vlist{T}$) is straightforward, by observing that in $\osmodel_\omega$ we can choose as many distinct witnesses for each $T_i$ as we want, because of the $\omega$-expansion. For the universal part, observe that $\posdgbnfofo{\vlist{T}}{\Pi}_\asort$ states that \emph{every} $d\in D$ satisfies some type in $\Pi$. Therefore, the same happens with the elements of $\osmodel_\omega$. In particular, for the elements that are not witnesses for $\vlist{T}$. Therefore, $\osmodel_\omega \models \posdbnfofoe{\vlist{T}}{\Pi}_\asort$.

\medskip\noindent
\fbox{$\Leftarrow$} Let $\osmodel_\omega \models \posdbnfofoe{\vlist{T}}{\Pi}_\asort$, we prove that $\osmodel \models \posdgbnfofo{\vlist{T}}{\Pi}_\asort$. For the existential part, consider some $T_i$, we show that it has a witness in $\osmodel$. We know by hypothesis that there is some $(d,i,k) \in \osmodel_\omega$ which is a witness for $T_i$. It is easy to see that $d$ works as a witness for $T_i$ in $\osmodel$. For the universal part, consider $d\in D$, we show that it satisfies some type in $\Pi$. If there is some $(d,i,k) \in \osmodel_\omega$ such that $(d,i,k)$ is not a witness of $\vlist{T}$ then we are done, as it should satisfy some type in $\Pi$ by the semantics of $\posdbnfofoe{\vlist{T}}{\Pi}_\asort$. The key observation is that there is always such an element, because at most $|\vlist{T}|$ elements of $\{(d,i,n) \mid n\in\nat\}$ function as witnesses for $\vlist{T}$.
\end{proof}

\begin{remark}
The above proposition is shown in~\cite{Zanasi:Thesis:2012,DBLP:conf/lics/FacchiniVZ13} for formulas in a special normal form, that is, where the elements of $\vlist{T}$ and $\Pi$ are either singletons or empty. In their case this is enough because, thanks to the simulation theorem, they can assume that their (non-deterministic) automata have this kind of formulas in their transition map. In our case this is a priory not true. Moreover, we are also in a multi-sorted setting.

An unsorted version of Proposition~\ref{prop:pinvariant} is implied by the original proof of $\muML\equiv\mso/{\bis}$ given in~\cite{Jan96}. However, the first explicit presentation of these bisimulation-invariance results in terms of `one-step models' was published by Venema in~\cite{Venxx}. A similar approach can also be found in his lecture notes on the $\mu$-calculus~\cite{Ven12}.
\end{remark}

The following proposition will be crucial for the development of this section. It states that we can assume the transition map of our automata to be in normal form. This is easily achieved by transforming the transition map using Corollary~\ref{cor:osnormalize}.

\begin{definition}
We say that an automaton $\aut \in \AutWA(\ofoe)$ is \emph{normalized} if the formulas of the transition map of $\aut$ are in the basic normal forms of Section~\ref{sec:onestep}. Namely,
	\begin{itemize}
		\itemsep 0pt
		\item Every $\tmap(a,c)$ is of the form $\bigvee \bigwedge_\aSort \posdbnfofoe{\vlist{T}}{\Pi}_\aSort$ as given in Corollary~\ref{cor:ofoepositivenf}; and,
		\item If $\tmap(a,c)$ is completely additive in $A'\subseteq A$, the additional properties of the normal form stated in Corollary~\ref{cor:ofoeadditivenf} apply.
	\end{itemize}
We say that $\aut$ is \emph{strictly normalized} if the formulas of the transition map of $\aut$ are in the \emph{strict} normal forms of Section~\ref{sec:onestep}. Namely,
	\begin{itemize}
		\itemsep 0pt
		\item Every $\tmap(a,c)$ is of the form $\bigvee \bigwedge_\asort \posdbnfofoe{\vlist{T}}{\Pi}_\asort$ as given in Corollary~\ref{cor:ofoepositivenf}; and,
		\item If $\tmap(a,c)$ is completely additive in $A'\subseteq A$, the additional properties of the normal form stated in Corollary~\ref{cor:ofoeadditivenf} apply.
	\end{itemize}
\end{definition}
\begin{proposition}\label{prop:normalize}
	For every $\aut \in \AutWA(\ofoe)$ we can effectively construct:
	\begin{itemize}
		\itemsep 0pt
		\item A normalized automaton ${\aut'\in \AutWA(\ofoe)}$ such that $\aut\equiv\aut'$ over all models.
		\item A strictly normalized automaton ${\aut'\in \AutWA(\ofoe)}$ such that $\aut\equiv\aut'$ over strict trees.
	\end{itemize} 
\end{proposition}

Finally, we can give the main definition of this section. That is, we define $\aut^\bullet$ for every $\aut$. This definition will be tailored to satisfy the condition:
\[
\aut^{\bullet} \text{ accepts } \model \quad\text{iff}\quad \aut \text{ accepts } \model^{\omega}.
\]
Therefore, it is worth reminding that as $\aut$ is run on $\model^\omega$ (which is a strict tree) we can assume that the transition map is in \emph{strict} normal form.

\begin{definition}
Let $\aut = \tup{A,\tmap,\pmap,a_{I}}$ be an automaton in $\Aut(\ofoe)$.
Using Proposition~\ref{prop:normalize} we assume that $\aut$ is strictly normalized. We define 
the automaton $\aut^{\bullet} := \tup{A,\tmap^{\bullet},\pmap,a_{I}}$ in 
$\Aut(\ofo)$ by putting, for each $(a,c) \in A \times \wp(\props)$:
\[
\tmap^{\bullet}(a,c) := (\tmap(a,c))^{\bullet}_1.
\]
\end{definition}

First, it needs to be checked that the construction $(-)^{\bullet}$, which has
been defined for arbitrary automata in $\Aut(\ofoe)$, transforms 
the additive-weak automata of $\AutWA(\ofoe)$ into automata in the right class, that is, $\AutWA(\ofo)$.

\begin{proposition}
Let $\aut \in \Aut(\ofoe)$.
If $\aut \in \AutWA(\ofoe)$, then $\aut^{\bullet} \in \AutWA(\ofo)$.
\end{proposition}

\begin{proof}
This proposition can be verified by a straightforward inspection, at the 
one-step level, that if a formula $\alpha \in \ofoe^+(A)$ belongs to the fragment 
$\add{\ofoe^+}{A'}(A)$, then its translation $\alpha^{\bullet}_1$ lands in 
the fragment $\add{\ofo^+}{A'}(A)$. The same relationship holds for $\mult{\ofoe^+}{A'}(A)$
and $\mult{\ofo^+}{A'}(A)$.
\end{proof}

\noindent We are now ready to prove the main lemma of this section.

\begin{lemma}\label{lem:bisautofoe}
There is an effective construction $(-)^{\bullet}: \AutWA(\ofoe) \to \AutWA(\ofo)$ such that
for every automaton $\aut$ and transition system $\model$ we have:
\[
\aut^{\bullet} \text{ accepts } \model \quad\text{iff}\quad \aut \text{ accepts } \model^{\omega}.
\]
\end{lemma}
\begin{proof}
The proof of this lemma is based on a fairly routine comparison of the 
acceptance games $\agame(\aut^{\bullet},\model)$ and 
$\agame(\aut,\model^{\omega})$, using the fact that $\model^{\omega}$ is a strict tree.
A similar proof for the unsorted case can be found in~\cite{Jan96,Zanasi:Thesis:2012,DBLP:conf/lics/FacchiniVZ13}.
\end{proof}

\begin{remark}
	Actually, for each automaton $\aut$ it is possible to give a number $k$ such that
	\[\aut^{\bullet} \text{ accepts } \model \quad\text{iff}\quad \aut \text{ accepts } \model^{k}\]
	taking $k:=\max\{n \mid \diff{x_1,\dots,x_n} \text{ occurs in $\aut$} \} + 1$. This means that our results transfer to the class of finitely branching trees and finite trees as well.
\end{remark}

%% file: exp-binfotc.tex

In this section we prove the following equivalence:
\[
\pdl \equiv \binfotc/{\bis}.
\]
Moreover, we prove that the equivalence is effective.

\medskip
\noindent One of the inclusions is given by a straightforward translation from \pdl to \binfotc.

\begin{proposition}
	There is an effective translation $\st_x^{tc}:\pdl\to\binfotc$ such that $\varphi\equiv \st_x^{tc}(\varphi)$ for every $\varphi\in\pdl$.
\end{proposition}
\begin{proof}
The translation is defined by mutual induction on formulas and programs, as follows:
\begin{itemize}
	\itemsep 0 pt
	\item $\st_x^{tc}(p) := p(x)$
	\item $\st_x^{tc}(\lnot\varphi) := \lnot \st_x^{tc}(\varphi)$
	\item $\st_x^{tc}(\varphi \lor \psi) := \st_x^{tc}(\varphi) \lor \st_x^{tc}(\psi)$
	\item $\st_x^{tc}(\tup{\prog}\varphi) := \exists y. \relprog_\prog(x,y) \land \st_y^{tc}(\varphi)$
\end{itemize}
where the complex programs are translated as follows
\begin{itemize}
	\itemsep 0 pt
	\item $\relprog_{\prog\seq\prog'}(x,y) := \exists x'. \relprog_\prog(x,x') \land \relprog_{\prog'}(x',y)$
	\item $\relprog_{\prog\choice\prog'}(x,y) := \relprog_\prog(x,y) \lor \relprog_{\prog'}(x,y)$
	\item $\relprog_{\varphi?}(x,y) := x\foeq y \land \st_x^{tc}(\varphi)$
	\item $\relprog_{\prog^*}(x,y) := [\tc_{z,w}.\relprog_{\prog}(z,w)](x,y)$.
\end{itemize}
It is clear that the translation is truth-preserving.
\end{proof}

\noindent For the other inclusion we prove the following stronger lemma.

\begin{lemma}\label{lem:fotc2pdl}
	There is an effective translation $(-)_\tmodal:\binfotc\to\pdl$ such that
	for every $\varphi\in\binfotc$ we have that $\varphi \equiv \varphi_\tmodal$ iff $\varphi$ is bisimulation-invariant.
\end{lemma}
\begin{proof}
	The translation $(-)_\tmodal:\binfotc\to\pdl$ is defined as follows: given a formula ${\varphi\in\binfotc}$ we first translate it to $\mucafoe$ using Theorem~\ref{thm:fotcmucafoe} and construct an automaton $\aut_\varphi\in\AutWA(\ofoe)$ as done in Section~\ref{sec:soaut-logic}. Next, we compute the automaton $\aut_\varphi^\bullet\in\AutWA(\ofo)$ using Lemma~\ref{lem:bisautofoe}. To finish, we use Fact~\ref{fact:autwaofo2pdl} to get a formula $\varphi_\tmodal\in\pdl$.

	\begin{claimfirst}
		$\varphi \equiv \varphi_\tmodal$ iff $\varphi$ is invariant under bisimulation.
	\end{claimfirst}
	The left to right direction is trivial because $\varphi_\tmodal\in\pdl$, therefore if $\varphi\equiv\varphi_\tmodal$ it also has to be invariant under bisimulation. The opposite direction is obtained by the following chain of equivalences:
	\begin{align*}
	\model \models \varphi
		& \quad\text{iff}\quad \model^{\omega} \models \varphi
		& \tag{$\varphi$ bisimulation invariant}
	\\
		& \quad\text{iff}\quad \model^{\omega} \models \aut_\varphi.
		& \tag{Theorem~\ref{thm:allthesame}: $\mucafoe \equiv \AutWA(\ofoe)$ on trees}
	\\
		& \quad\text{iff}\quad \model \mmodels \aut_\varphi^\bullet
		& \tag{Lemma~\ref{lem:bisautofoe}}
	\\
		& \quad\text{iff}\quad \model \mmodels \varphi_\tmodal.
		& \tag{Fact~\ref{fact:autwaofo2pdl}}
	\end{align*}
	This finishes the proof for \binfotc.
\end{proof}

\noindent As a corollary of this lemma we get Theorem~\ref{thm:mainbinfotc}.

%% file: exp-wcl.tex

In this section we prove the following equivalence:
\[
\pdl \equiv \wcl/{\bis}.
\]
Moreover, we prove that the equivalence is effective.

\medskip
\noindent One of the inclusions is given by a translation from \pdl to \wcl.
We prove this through a detour via the modal $\mu$-calculus. In~\cite{AiML2014} it is shown that \pdl is equivalent to the fragment $\mucaML$ where the fixpoint operator $\mu p.\varphi$ is restricted to formulas which are completely additive in $p$. We will therefore give a translation $\st_x^{wc}:\mucaML \to 2\wcl$ which proves that $\pdl \leq \wcl$. The idea is to use basically the same translation as in Section~\ref{ssec:cafoe2wcl} where we prove that $\mucafoe \leq \wcl$ on trees.

The only interesting case of the translation is the fixpoint operator. Let $\varphi = \mu p.\psi(p)$ where $\psi$ is completely additive in $p$. We state the following claim

\begin{claimfirst}\label{claim:psiinv:2}
	The formula $\psi\in\mucaML$ restricts to descendants.
\end{claimfirst}
\begin{pfclaim}
	This is clear because the formula belongs to $\muML$. These formulas are invariant under generated submodels, in particular, they restrict to descendants.
\end{pfclaim}

\noindent To finish, define the translation of the fixpoint as follows:
\begin{align*}
\st_x^{wc}(\mu p.\psi) & := \existswc Y.\big(\forallwc W \subseteq Y.W \in \mathsf{PRE}(F^\psi_Y) \to x\in W \big) \\
W \in \mathsf{PRE}(F^\psi_Y) & := \forall v. \st_v^{wc}(\psi)[p\mapsto W] \land v\in Y \to v\in W.
\end{align*}
The correctness of this translation is a simplified version of the proof of Proposition~\ref{prop:mucafoe2wcl}, using Claim~\ref{claim:psiinv:2} and Theorem~\ref{thm:restofinsucchain}.

\medskip
\noindent For the other inclusion we prove the following stronger lemma.

\begin{lemma}
	There is an effective translation $(-)_\tmodal:\wcl\to\pdl$ such that
	for every $\varphi\in\wcl$ we have that $\varphi \equiv \varphi_\tmodal$ iff $\varphi$ is bisimulation-invariant.
\end{lemma}
\begin{proof}
	The translation $(-)_\tmodal:\wcl\to\pdl$ is defined as follows: given a formula ${\varphi\in\wcl}$ we first construct an automaton $\aut_\varphi\in\AutWA(\ofoe)$ as done in Section~\ref{sec:soaut-logic}. Next, we compute the automaton $\aut_\varphi^\bullet\in\AutWA(\ofo)$ using Lemma~\ref{lem:bisautofoe}. To finish, we use Fact~\ref{fact:autwaofo2pdl} to get a formula $\varphi_\tmodal\in\pdl$.

	\begin{claimfirst}
		$\varphi \equiv \varphi_\tmodal$ iff $\varphi$ is invariant under bisimulation.
	\end{claimfirst}
	The left to right direction is trivial because $\varphi_\tmodal\in\pdl$, therefore if $\varphi\equiv\varphi_\tmodal$ it also has to be invariant under bisimulation. The opposite direction is obtained by the following chain of equivalences:
	\begin{align*}
	\model \models \varphi
		& \quad\text{iff}\quad \model^{\omega} \models \varphi
		& \tag{$\varphi$ bisimulation invariant}
	\\
		& \quad\text{iff}\quad \model^{\omega} \models \aut_\varphi.
		& \tag{Theorem~\ref{thm:allthesame}: $\wcl \equiv \AutWA(\ofoe)$ on trees}
	\\
		& \quad\text{iff}\quad \model \mmodels \aut_\varphi^\bullet
		& \tag{Lemma~\ref{lem:bisautofoe}}
	\\
		& \quad\text{iff}\quad \model \mmodels \varphi_\tmodal.
		& \tag{Fact~\ref{fact:autwaofo2pdl}}
	\end{align*}
	This finishes the proof for \wcl.
\end{proof}

\noindent As a corollary of this lemma we get Theorem~\ref{thm:mainwcl}.

%% file: exp-versus.tex

In this section we prove a few results regarding the relative expressive power of \pdl, \wcl and \binfotc. Namely, we prove that:
\begin{itemize}
	\itemsep 0pt
	\item \pdl cannot be translated to a naive generalization of \wcl from trees to arbitrary models. This gives insight into the relationship of \pdl and generalized chains.
	\item \wcl and \binfotc are not expressively equivalent.
\end{itemize}

\paragraph{The letter of the law.}
Recall from the preliminaries (Section~\ref{ssec:prelim_trees}) that,

\begin{itemize}
	\itemsep 0 pt
	\item A \emph{chain} on $\model$ is a set $X\subseteq \moddom$ such that $(X,R^*)$ is a totally ordered set.
	\item A \emph{generalized chain} is a set $X \subseteq \moddom$ such that $X\subseteq P$, for some path $P$ of $\model$.
\end{itemize}
In the definition of \wcl in Section~\ref{sec:wcl} we chose to follow the spirit of the original definition of CL in~\cite{Thomas:1984:LAS:2868.2871}, and we required the quantifier to range over generalized finite chains, instead of finite chains, in the context of arbitrary models. 

There is another reason for this choice: if we had followed the letter of the definition and had required the quantifier to range over (non-generalized) finite chains in the context of arbitrary models, then \pdl \emph{would not} have been translatable to the resulting logic!
Suppose then that we define a variant $L$ of MSO with the quantifier:
\[
\model \models \tilde{\exists} p.\varphi \quad\text{ iff }\quad  \text{there is a \emph{finite chain} $X\subseteq \moddom$ such that $\model[p\mapsto X] \models \varphi$}.
\]
We show that $L$ cannot express the \pdl-formula $\varphi := \tup{\aact^*}p$. That is, $L$ cannot express the property ``I can reach an element colored with $p$.'' Intuitively, the problem is that chains are a lot more restricted than paths (on arbitrary models).

We will first define a class of models where the expressive power of $L$ is reduced to that of \foe, and then prove that \foe cannot express $\varphi$ on this class of models. Let $\mathbb{C}_i$ be defined as a model with $i$ elements laid out on a circle (see Fig.~\ref{fig:circles}) and $\mathbb{C}^p_i$ be as $\mathbb{C}_i$ but with one (any) element colored with $p$. We define the class of models $K := \{\mathbb{C}_i \uplus \mathbb{C}^p_i \mid i \geq 3\}$.

\input{fig-circles}

\begin{proposition}
	Over the class $K$, the logic $L$ is exactly as expressive as \foe.
\end{proposition}
\begin{proof}
	Every chain on a model of $K$ is either a singleton or empty.
\end{proof}

Observe now that our formula $\varphi = \tup{\aact^*}p$ is true exactly in the elements of $\mathbb{C}^p_i$, and false in $\mathbb{C}_i$, for every $i$. Assume towards a contradiction that there is a formula $\psi \in L$ such that $\varphi \equiv \psi$ on all models. If we focus on $K$, using the above proposition, we must also have a formula $\gamma \in \foe$ such that $\psi \equiv \gamma$ (on $K$). We show that such a $\gamma\in\foe$ cannot exist.

To do it, we rely on the fact that first-order logic is ``essentially local'', proved by Gaifman~\cite{gaifman}. Recall that an $n$-neighbourhood of an element $e$ is the set of all the elements $e'$ such that the undirected distance $\mathrm{dist}(e,e')$ is smaller or equal than $n$. The following fact is a corollary of Gaifman's theorem.

\begin{fact}\label{fact:gaifman}
	For every first-order formula $\gamma(\vlist{x})$ there is a number $t\in\nat$ (which depends only on the quantifier rank of $\gamma$) such that for every model $\npmodel$ and elements $\vlist{a},\vlist{a}' \in \npmoddom$: if the $t$-neighbourhoods of $a_i$ and $a'_i$ are isomorphic for every $i$ then $\npmodel \models \gamma(\vlist{a})$ iff $\npmodel \models \gamma(\vlist{a}')$.
\end{fact}

Let $t$ be the number obtained by the above fact applied to $\gamma(x)$. To finish, we prove the following fact.
\begin{claimfirst}
	$\mathbb{C}_{4t} \models \gamma(2t)$ iff
	$\mathbb{C}_{4t} \models \gamma(2t').$
\end{claimfirst}
\noindent Observe that this leads to a contradiction, since $\gamma$ should be false at $2t$ and true at $2t'$.
\begin{pfclaim}
	The $t$-neighbourhoods of $2t$ and $2t'$ are isomorphic, since no element is colored with $p$ with distance lower than $t$. Therefore by the above fact about first-order locality the two elements satisfy the same first-order formulas. A more detailed proof of a similar argument can be found in~\cite[Ex.~2]{libkinsurvey}.
\end{pfclaim}

\noindent As a consequence, we get the following proposition:

\begin{proposition}
	$\pdl \not \leq L$.
\end{proposition}

\paragraph{Separating \binfotc and \wcl on all models.}

We prove that $\binfotc \not\leq \wcl$ by showing that \emph{undirected reachability} is expressible in \binfotc but not in \wcl. First observe that in \binfotc the formula
\[
\varphi(x,y) := [\tc_{u,v}. R(u,v) \lor R(v,u)](x,y)
\]
is true iff $x=y$ or there is a way to get from $x$ to $y$ disregarding the direction of the edges.

Consider the model shown in Fig.~\ref{fig:separate}, which has two copies of the integers but with an alternating successor relation. The arrows denote the binary relation $R$ which is \emph{not} taken to be transitive.

\input{fig-zz}

It was observed by Yde Venema (private communication) that on the model of Fig.~\ref{fig:separate} the expressive power of \wcl collapses to that of plain first-order logic with equality. The reason for this is that every generalized chain (finite or not) has length at most one, and therefore the second order existential $\existswc X.\psi$ can be replaced by $\exists x_1,x_2.\psi'$ with a minor variation of $\psi$. Therefore, it will be enough to show that first-order logic cannot express undirected reachability over this model. Again, we will use Fact~\ref{fact:gaifman}.



Assume that $\varphi$ has an equivalent formulation $\varphi' \in \mathrm{FOE}$. Let $t$ be the number obtained by Fact~\ref{fact:gaifman}. To finish, we prove that
\begin{claimfirst}
	$\npmodel \models \varphi'((a,0),(b,0))$
	iff
	$\npmodel \models \varphi'((a,0),(a,2t)).$
\end{claimfirst}
Observe that this leads to a contradiction, since the first two elements are not connected and the second ones are.
\begin{pfclaim}
	The $t$-neighbourhoods of $(b,0)$ and $(a,2t)$ are isomorphic, therefore by the above fact about first-order locality the two elements satisfy the same first-order formulas. A more detailed proof of a similar argument can be found in~\cite[Ex.~2]{libkinsurvey}.
\end{pfclaim}

Observe also that in the above model $\npmodel$ we also have that non-weak chain logic (CL) collapses to first-order logic, therefore the same proof gives us that $\binfotc \not\leq \mathrm{CL}$. Therefore, we have proved the following proposition:

\begin{proposition}\label{prop:sepwclbinfotc}
	$\binfotc \not\leq \wcl$ and $\binfotc \not\leq \mathrm{CL}$.
\end{proposition}

The above results can also be proved with a finite part of $\npmodel$, for example, restricting it to the segments $(a,\pm 4t)$ and $(b,\pm 4t)$.
\fcwarning{Is $\wcl$ included in CL?}

%% file: fig-circles.tex

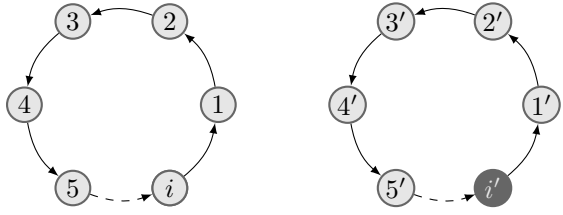
\begin{figure}[h]
\centering
\begin{tikzpicture}[
scale=0.85,
font=\small,
every node/.style={circle,draw=black!60,fill=black!10,thick,inner sep=1pt,minimum size=13pt},
->,>=latex
]

\begin{scope}[xshift=-2.5cm]
\def \n {6}
\def \radius {1.5cm}

\foreach \s in {1,...,\n}
{
  \node (n\s) at ({360/\n * (\s - 1)}:\radius) {$\s$};
}

\node (ni) at ({360/6 * (6 - 1)}:\radius) {$i$};

\draw (n1) edge[bend right=20,solid] (n2)
      (n2) edge[bend right=20,solid] (n3)
      (n3) edge[bend right=20,solid] (n4)
      (n4) edge[bend right=20,solid] (n5)
      (n5) edge[bend right=20,dashed] (ni)
      (n6) edge[bend right=20,solid] (n1);
\end{scope}

\begin{scope}[xshift=2.5cm]
\def \n {6}
\def \radius {1.5cm}

\foreach \s in {1,...,\n}
{
  \node (n\s) at ({360/\n * (\s - 1)}:\radius) {$\s'$};
}

\node[draw=black!60,fill=black!60] (ni) at ({360/6 * (6 - 1)}:\radius) {\textcolor{black!10}{$i'$}};

\draw (n1) edge[bend right=20,solid] (n2)
      (n2) edge[bend right=20,solid] (n3)
      (n3) edge[bend right=20,solid] (n4)
      (n4) edge[bend right=20,solid] (n5)
      (n5) edge[bend right=20,dashed] (ni)
      (n6) edge[bend right=20,solid] (n1);
\end{scope}

\end{tikzpicture}
\caption{Model $\mathbb{C}_i \uplus \mathbb{C}^p_i$. The element $i'$ is colored with $p$.}
\label{fig:circles}
\end{figure}

%% file: fig-zz.tex

\begin{figure}[h]
\centering
\tikzstyle{node}=[circle,fill=black,inner sep=1pt,minimum size=7pt]
\begin{tikzpicture}[thick]
\begin{scope}
	\node (0) at (0,0) [node,label=above:{$(a,0)$}] {};

	\node (1) [right=of 0,node,label=above:{$(a,1)$}] {};
	\node (2) [right=of 1,node,label=above:{$(a,2)$}] {};
	\node (3) [right=of 2,node,label=above:{$(a,3)$}] {};
	\node (end) [right=10pt of 3] {$\dots$};

	\node (-1) [left=of 0,node,label=above:{$(a,-1)$}] {};
	\node (-2) [left=of -1,node,label=above:{$(a,-2)$}] {};
	\node (-3) [left=of -2,node,label=above:{$(a,-3)$}] {};
	\node (-end) [left=10pt of -3] {$\dots$};

	\draw [->] (0) -- (1);
	\draw [<-] (1) -- (2);
	\draw [->] (2) -- (3);

	\draw [->] (0) -- (-1);
	\draw [<-] (-1) -- (-2);
	\draw [->] (-2) -- (-3);
\end{scope}
\begin{scope}[yshift=-1.5cm]
	\node (0') at (0,0) [node,label=above:{$(b,0)$}] {};

	\node (1') [right=of 0',node,label=above:{$(b,1)$}] {};
	\node (2') [right=of 1',node,label=above:{$(b,2)$}] {};
	\node (3') [right=of 2',node,label=above:{$(a,3)$}] {};
	\node (end') [right=10pt of 3'] {$\dots$};

	\node (-1') [left=of 0',node,label=above:{$(b,-1)$}] {};
	\node (-2') [left=of -1',node,label=above:{$(b,-2)$}] {};
	\node (-3') [left=of -2',node,label=above:{$(b,-3)$}] {};
	\node (-end') [left=10pt of -3'] {$\dots$};

	\draw [->] (0') -- (1');
	\draw [<-] (1') -- (2');
	\draw [->] (2') -- (3');

	\draw [->] (0') -- (-1');
	\draw [<-] (-1') -- (-2');
	\draw [->] (-2') -- (-3');
\end{scope}
\end{tikzpicture}
\caption{Separating example for $\binfotc$ and $\wcl$.}
\label{fig:separate}
\end{figure}
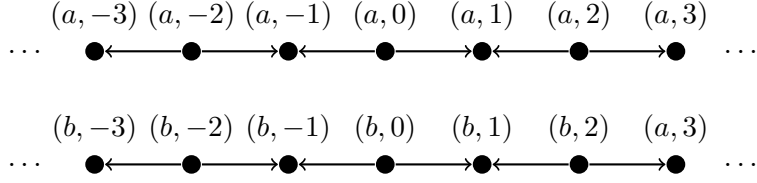

%% file: conclusions.tex

In this article we proved several characterization results for modal and classical fixpoint logics. The main results can be grouped as follows:

\begin{enumerate}
  \itemsep 0 pt
  \item
  We proved that $\binfotc \equiv \mucafoe$. That is, first-order logic with (binary) transitive closure is expressively equivalent to the fragment of $\unfolfp$ where the fixpoint is restricted to completely additive formulas (Section~\ref{sec:fotcinmufoe}).
  \item
  We introduced a new class $\AutWA(\ofoe)$ of additive-weak parity automata and proved that, on trees, the following formalisms are equivalent (Section~\ref{sec:soaut-logic}):
	\begin{itemize}
		\itemsep 0 pt
		\item $\wcl$: Weak chain logic,
		\item $\AutWA(\ofoe)$: Additive-weak automata based on $\ofoe$,
		\item $\mucafoe$: Completely additive restriction of $\mufoe$.
	\end{itemize}
  \item 
  We gave two characterization results for \pdl and, at the same time, solved the open problems of the bisimulation-invariant fragments of \binfotc and \wcl (Section~\ref{sec:char}). Namely, we proved that
  $\pdl \equiv \binfotc/{\bis}$ and $\pdl \equiv \wcl/{\bis}$.
\end{enumerate}

\noindent In order to obtain the results of the above list, we also had to develop many secondary results that are of independent interest. Among others, we proved:

\begin{enumerate}
	\itemsep 0pt
	\item Characterization of fixpoints of completely additive maps (Theorem~\ref{thm:propscafmap}).
	\item Characterization of fixpoints of maps that restrict to descendants (Theorem~\ref{thm:propssucmap}).
	\item Separation of \binfotc and CL/\wcl (Proposition~\ref{prop:sepwclbinfotc}).
	\item Normal forms and completely additive fragments for $\ofo$ and $\ofoe$ (Section~\ref{sec:onestep}).
\end{enumerate}

\paragraph{Open questions.}
There are many old and new open problems related to the content of this paper. In the following non-exhaustive list we name a few of them:
\begin{itemize}
	\item \textit{Automata for Monadic Path Logic}: Moller and Rabinovich~\cite{expressiveCTL} show that the full computation tree logic ($\mathrm{CTL}^*$) corresponds, on trees, to the bisimulation-invariant fragment of monadic path logic (MPL). The latter logic is a variant of MSO which quantifies over full (finite or infinite) paths. In contrast, \wcl quantifies over subsets of finite paths. Moller and Rabinovich do not use an automata-theoretic approach and leave this approach as an open question. Given the similarities between \pdl and $\mathrm{CTL}^*$, it would be interesting to see if the approach in this article can be applied to get an automata characterization of MPL and also characterize its bisimulation-invariant fragment on the class of all models.
	\item \textit{Characterization of full \fotc inside \folfp}: In this article we gave a precise characterization of the relationship between \binfotc and \unfolfp. It would be worth checking if this relationship lifts to \fotc and \folfp. As the logic \fotc captures NLOGSPACE queries, and \folfp captures PTIME this result could shed light on the relationship between NLOGSPACE and PTIME.
	\item \textit{Finitary versions of the bisimulation-invariance theorems}: It would be interesting to know if the bisimulation-invariance results of this article hold in the class of finite models. However, it is also not known whether the more fundamental equivalence $\muML \equiv \mso/{\bis}$ holds on finite models or not.
	\item \textit{The confusion conjecture}: In~\cite{Mikolaj-Thesis} Boja{\'n}czyk defines a notion of `confusion' and conjectures that a regular language (i.e., MSO definable) of finite trees is definable in CL iff it contains no confusion. A remarkable property of the notion of confusion is that it is \emph{decidable} whether a language has it or not. As the results of our paper transfer to finite trees (and $\mathrm{CL} \equiv \wcl$ in that class) the conjecture implies that a language definable in the mu-calculus (on finite trees) is definable in \pdl iff it contains no confusion. It is a major open problem whether we can decide if an arbitrary formula of $\muML$ is equivalent to some formula in \pdl. Therefore, it would be important to check the confusion conjecture.
	\item \textit{Characterization of $\AutWA(\ofoe)$ on all models}: It is known that, on all models, $\Aut(\ofoe)$ is equivalent to the graded $\mu$-calculus~\cite{Janin2006}. As we know that the equivalences $\Aut(\ofoe) \equiv \mso \equiv \unfolfp$ hold on trees~\cite{Walukiewicz96}, this means that, on trees, all these formalisms are equivalent to the graded $\mu$-calculus. It would be interesting to see if a similar result can be obtained for $\AutWA(\ofoe)$.

	Our conjecture is that $\AutWA(\ofoe)$ is, on all models, equivalent to a graded version of $\mucaML$. We think that, equivalently, $\AutWA(\ofoe)$ corresponds to a kind of graded \pdl where the number restrictions can only be applied to atomic programs. This may be better seen as the description logic $\mathcal{ALCQ}_\mathrm{reg}$~(see~\cite{Schild:1991:CTT:1631171.1631241,DeGiacomoLenzerini1996}). As a corollary, we would get that, on trees, \binfotc is equivalent to $\mathcal{ALCQ}_\mathrm{reg}$.
	%
	%
	\item \textit{Restriction to descendants}: The autors of~\cite{arec:hybr99a} show that the bounded fragment of first-order logic is equivalent to the hybrid language $\mathcal{H}(\downarrow,@)$. If $\muffoe$ is a proper generalization of this language, it should be possible to prove that $\muffoe \equiv \mu\mathcal{H}(\downarrow,@)$.
	\item \textit{Standalone proof of $\pdl \equiv \binfotc/{\bis}$}: One of the most cumbersome parts of this article is the logical characterization of $\AutWA(\ofoe)$ as \wcl on trees. In particular, the simulation theorem and the closure under projection are very technical. If we are only interested in proving that $\AutWA(\ofoe)\equiv \binfotc$ on trees, it shouldn't be strictly necessary to go through \wcl as we do in Section~\ref{sec:soaut-logic}.

	The crucial point is to show that from $\mucaffoe$ we can directly construct an equivalent automaton in $\AutWA(\ofoe)$. While doing it, we would have to prove that $\AutWA(\ofoe)$ is `closed under the $\lfp$ operation.' It is possible that a proof like~\cite[Theorem~3.2.2.1]{Janin2006} could be adapted to this setting. In that case, \wcl could be completely taken out of the picture.
	%
\end{itemize}

\paragraph{Acknowledgements.}
The author is grateful to Yde Venema for thoroughly reading and commenting on \emph{many} previous versions of this document; to Alessandro Facchini, Sumit Sourabh and Fabio Zanasi for fruitful discussions on the topics of this article; and to Balder ten Cate for his patience with my many emails on $\fotc$.